\RequirePackage[T1]{fontenc}
\documentclass[12pt]{article}

\usepackage[height=8.85in,width=6.25in]{geometry}

\newlength{\vertspace}
\usepackage[utf8]{inputenc}
\usepackage{times}
\usepackage{amsmath}
\usepackage{amssymb}
\numberwithin{equation}{section}

\usepackage{graphicx}
\usepackage[svgnames]{xcolor}
\usepackage[colorlinks,citecolor=DarkGreen,linkcolor=FireBrick,urlcolor=FireBrick,linktocpage,breaklinks=true]{hyperref}
\usepackage{cite}
\usepackage{amsthm}
\theoremstyle{definition}
\newtheorem{thm}{Theorem}[section]
\newtheorem{prop}[thm]{Proposition}
\newtheorem{lemma}[thm]{Lemma}
\newtheorem{question}{Question}
\newtheorem*{question*}{Question}
\newenvironment{questionprime}[1]
{\renewcommand{\thequestion}{\ref{#1}$'$}
\addtocounter{question}{-1}
\begin{question}}
{\end{question}}
\theoremstyle{remark}
\newtheorem*{rem*}{Remark}
\usepackage{listings}
\usepackage[all]{xy}
\usepackage{here}
\usepackage{lscape}
\usepackage[normalem]{ulem}

\usepackage{bm}
\usepackage[bbgreekl]{mathbbol}
\DeclareSymbolFontAlphabet{\mathbb}{AMSb}
\DeclareSymbolFontAlphabet{\mathbbl}{bbold}

\def\cA{\mathcal{A}}
\def\Aut{\mathrm{Aut}}
\def\C{\mathbb{C}}
\def\ch{\mathrm{ch}}
\def\Co{\mathrm{Co}}

\def\bbe{\mathbbl{e}}
\def\End{\mathrm{End}}
\def\F{\mathbb{F}}
\def\GL{\mathrm{GL}}
\def\h{\mathfrak{h}}
\def\Hom{\mathrm{Hom}}
\def\id{\mathrm{id}}
\DeclareMathOperator{\im}{im}
\def\LinAut{\mathrm{LinAut}}
\def\Map{\mathrm{Map}}
\def\cN{\mathcal{N}}

\def\R{\mathbb{R}}
\def\rR{\mathrm{R}}
\def\Span{\mathrm{Span}}
\def\Stab{\mathrm{Stab}}
\def\Tr{\mathrm{Tr}}
\def\Vir{\mathrm{Vir}}
\def\wt{\mathrm{wt}}
\def\Z{\mathbb{Z}}
\newcommand{\Zgen}[2]{\mathbb{Z}_{#1}}
\def\^{\string^}
\newcommand{\tildedot}{\mathbin{\widetilde{\cdot}}}

\begin{document}

\begin{titlepage}
\begin{center}
\vspace*{2cm}

\textbf{\Large On the Symmetry of Odd Leech Lattice CFT}

\vspace{1cm}

Masaki Okada

\vspace{1cm}

\begin{tabular}{l}
Kavli Institute for the Physics and Mathematics of the Universe (WPI), \\
University of Tokyo,  Kashiwa, Chiba 277-8583, Japan 
\end{tabular}

\vspace{2cm}

\begin{abstract}
We show that the Mathieu groups $M_{24}$ and $M_{23}$ in the isometry group of the odd Leech lattice
do not lift to subgroups of the automorphism group of its lattice vertex operator (super)algebra.
In other words,
the subgroups $2^{24}.M_{24}$ and $2^{23}.M_{23}$ of the automorphism group of the odd Leech lattice vertex operator algebra are non-split extensions.
Our method can also confirm a similar result for the Conway group $\Co_0$ and the Leech lattice,
which was already shown in \cite{MR0476878}.
This study is motivated by the moonshine-type observation on the $\cN=2$ extremal elliptic genus of central charge 24 by \cite{Benjamin:2015ria}.
We also investigate weight-1 and weight-$\frac{3}{2}$ currents invariant under the subgroup $2^{24}.M_{24}$ or $2^{23}.M_{23}$ of the automorphism group of the odd Leech lattice vertex operator algebra,
and revisit an $\cN=2$ superconformal algebra in it.
\end{abstract}

\end{center}
\end{titlepage}

\tableofcontents

\newpage

\section{Introduction}
\label{sec:intro}

Symmetries of lattices can be key ingredients in understanding moonshine phenomena.
The most prominent example is the moonshine module $V^\natural$,
which was constructed in \cite{MR0996026}
as a $\Z_2$ orbifold of the Leech lattice vertex operator algebra (VOA),
and provided an explanation for the appearance of the representation dimensions of the monster group $\mathbb{M}$ in the coefficients of the modular $J$-function.
The symmetry, or the automorphism (isometry) group $\Aut(\Lambda_{24})$ of the Leech lattice $\Lambda_{24}$ is the Conway group $\Co_0$,
and its information is inherited by the automorphism group $\Aut(V^\natural)$ of the moonshine module $V^\natural$,
which is the monster group $\mathbb{M}$,
in the form of the maximal subgroup $2^{1+24}.\Co_1\subset\mathbb{M}$,
where $\Co_1 = \Co_0 / \Z_2$ is the quotient of $\Co_0$ by its center.\footnote{
Here, $N.G$ denotes any (split or non-split) extension of a group $G$ by a group $N$.
In addition, we write $N:G$
if it splits, or it is a semidirect product of $N$ and $G$.
See the beginning of Section \ref{sec:preliminaries} or Appendix \ref{subsec:grp-ext-and-grp-coh} for more on notations of groups.
}
The proof that the automorphism group $\Aut(V^\natural)$ is the monster group $\mathbb{M}$ was first given by \cite{MR0996026}, and another simple proof was provided in \cite{MR2081435}.

Another case where we can expect the symmetry of a lattice accounts for a moonshine phenomenon was observed in \cite{Benjamin:2015ria}.
They pointed out that 
if we decompose a certain weak Jacobi form $Z_\mathrm{ext}^{m=4}(\tau, z)$ of weight 0 and index $m=4$,
called the $\cN=2$ extremal elliptic genus of central charge $c=6m=24$ \cite{Witten:2007kt, Gaberdiel:2008xb},
into the characters $\ch^{\cN=2,\widetilde{\rR}}_{l=\frac{c}{6}-\frac{1}{2};h,Q}(\tau,z)$ of the Ramond representations of $\cN=2$ superconformal algebra (SCA),
where the tilde of $\widetilde{\rR}$ denotes $(-1)^{F=J_0}$ inserted,
\begin{align}
Z_\mathrm{ext}^{m=4} 
= \ & \ch_{\frac{7}{2};1,4}^{\cN=2,\widetilde{\rR}} + 47 \ch_{\frac{7}{2};1,0}^{\cN=2,\widetilde{\rR}} \nonumber\\
& + 23\,\ch_{\frac{7}{2};2,4}^{\cN=2,\widetilde{\rR}} + 2024 (\ch_{\frac{7}{2};2,3}^{\cN=2,\widetilde{\rR}} + \ch_{\frac{7}{2};2,-3}^{\cN=2,\widetilde{\rR}}) + 14168 (\ch_{\frac{7}{2};2,2}^{\cN=2,\widetilde{\rR}} + \ch_{\frac{7}{2};2,-2}^{\cN=2,\widetilde{\rR}}) + \cdots,
\end{align}
then representation dimensions of the largest Mathieu group $M_{24}$ appear in a nontrivial way.
(See for example \cite{Cheng:2014owa} for explicit formulae of $\ch^{\cN=2,\widetilde{\rR}}_{l;h,Q}(\tau,z)$.)
For example, $M_{24}$ has irreducible representations of dimension 23 and 2024,
and 14168 is the sum $10395+3520+253$ of irreducible representation dimensions of $M_{24}$.
They further constructed a chiral $\cN=2$ superconformal field theory (SCFT)
whose partition function of the R sector coincides with the extremal elliptic genus $Z_\mathrm{ext}^{m=4}(\tau,z)$.
Their construction was the fermionization of the lattice conformal field theory (CFT)
of the $(A_1)^{24}$ lattice,
together with a special choice of $\cN=2$ SCA generators.
The automorphism group of the $(A_1)^{24}$ lattice is $2^{24}:M_{24}$ \cite[Table 16.1, Ch.\ 18 \S4]{MR1662447},
although the subgroup of the automorphism group of the resulting CFT fixing their $\cN=2$ SCA becomes smaller than $M_{24}$.
A similar analysis was also done for $\cN=4$ extremal CFT of central charge 24 in \cite{Harrison:2016hbq}.

This situation of \cite{Benjamin:2015ria} seems analogous to the celebrated K3 Mathieu moonshine.
The K3 Mathieu moonshine was first observed in \cite{Eguchi:2010ej} and studied extensively (see for example \cite{Gaberdiel:2011fg, Gannon:2012ck, Gaberdiel:2013psa} and the references listed below),
but is still missing some nice algebraic structure to explain all the properties such as the twining elliptic genera \cite{Cheng:2010pq, Gaberdiel:2010ch, Gaberdiel:2010ca, Eguchi:2010fg}, their twisted versions, and their modularities \cite{Gaberdiel:2012gf}.
Research to find such structure includes symmetry surfing \cite{Taormina:2011rr, Taormina:2013jza, Taormina:2013mda, Gaberdiel:2016iyz} and exploiting Duncan's module for the Conway moonshine \cite{MR2352133, Cheng:2014owa, Duncan:2014eha, Duncan:2015xoa, Taormina:2017zlm}.

\vspace{\vertspace}

Although the lattice symmetry seemingly offers some explanation for the appearance of representation dimensions of $M_{24}$ observed in \cite{Benjamin:2015ria},
if we take a closer look,
the situation turns out to be more complicated than expected as we will explain now.

We first note that the construction of the $\cN=2$ extremal SCFT in \cite{Benjamin:2015ria} was the fermionization with respect to the reflection $\Z_2$ symmetry of the $(A_1)^{24}$ lattice CFT,
but we can obtain an equivalent theory by performing the fermionization with respect to the shift $\Z_2$ symmetry of the same theory.\footnote{
\label{fn:equiv-of-reflection-and-shift-sym}
The equivalence between the reflection $\Z_2$ orbifold and the shift $\Z_2$ orbifold was proved in \cite{Dolan:1994st} in the cases of lattice CFTs constructed from doubly-even self-dual binary codes.
Here, the $(A_1)^{24}$ lattice CFT is the one constructed from the binary Golay code,
and fermionization and $\Z_2$ orbifold can be uniformly treated in the modern understanding (see Section \ref{subsec:lattice-CFT}).
Recently, the equivalence of the two $\Z_2$ orbifolds in the cases of fermionic lattice CFTs constructed from singly-even self-dual codes was reported in \cite{Kawabata:2024gek}.
}
The NS sector of the resulting CFT is a lattice CFT of the odd Leech lattice $O_{24}$,
and the R sector also admits a description in terms of lattices.
In fact, the construction of the extremal elliptic genus from an odd Leech lattice was also done in \cite{Kawabata:2023rlt} from the perspective of code CFT
(see also \cite[Example 4.5]{Gaiotto:2018ypj}).

One might expect that
we could take an $\cN=2$ SCA in the VOA of the NS sector $V_{O_{24}}$ such that the subgroup of the automorphism group $\Aut(V_{O_{24}})$ preserving the said $\cN=2$ SCA is the Mathieu group $M_{24}$.
If it were the case,
it might justify the appearance of the representation dimensions of $M_{24}$ in the partition function of the R sector,
under the spectral flow of the $\cN=2$ SCA.
In fact, the automorphism group of the odd Leech lattice is $\Aut(O_{24})=2^{12}:M_{24}$ and contains $M_{24}$ as a subgroup.

However, here we encounter a problem.
The symmetry $\Aut(V_L)$ of a lattice VOA inherits the information of the lattice symmetry $\Aut(L)$
in the form of the \emph{group extension} of $\Aut(L)$.
More precisely, it is the group extension $2^{\mathrm{rank} \, L}.\Aut(L)$ of $\Aut(L)$ by $\Hom(L,\Z_2)\cong 2^{\mathrm{rank}\, L}$ that is the subgroup of $\Aut(V_L)$ \cite{MR0996026, MR1745258},
and hence $\Aut(L)$ itself is not a subgroup of $\Aut(V_L)$ unless the group extension splits.
In the language of physics, this is because of the cocycle factors appearing in the operator product expansion (OPE) of vertex operators.

The main goal of this paper is to answer the following question,
in particular in the case of the odd Leech lattice.
\begin{question*}
Does the automorphism group $\Aut(L)$ (or its subgroup) of a lattice $L$
lift to a subgroup of the automorphism group $\Aut(V_L)$ of the lattice VOA $V_L$?
\end{question*}
\noindent We will analyze the group extension quite explicitly,
and finally prove that the answer is no in the cases of the subgroups $M_{24}$ and $M_{23}$ of the automorphism group of the odd Leech lattice by contradiction (Theorem \ref{thm:main-oddLeech}).
We will also check that the answer is no in the case of the Conway group $\Co_0$ of the Leech lattice by a similar method;
this result was already shown in \cite{MR0476878} (see also \cite{MR885104}, \cite[Lemma 1.8.8 (iii)]{MR2503090}),
but we hope that our method is more elementary and easy to follow for non-experts of group theory.
Of course, there is a possibility that
the techniques of \cite{MR0476878, MR885104, MR2503090} can also be used to show the result for the odd Leech lattice,
but this result is not explicitly available in literature to the author's knowledge.

As a result, the odd Leech lattice CFT does not have the $M_{24}$ symmetry in the first place,
unless we forget the OPE structure and regard it as just a momentum lattice.
Instead, it has a $2^{24}.M_{24}$ symmetry which does not split.
We further investigate the weight-1 and weight-$\frac{3}{2}$ currents of the odd Leech lattice CFT
invariant under the symmetry $2^{24}.M_{24}$ or its subgroup $2^{23}.M_{23}$,
and conclude that there is no such weight-$\frac{3}{2}$ invariant current.
Therefore, we cannot find $\cN=2$ SCA invariant under the $2^{24}.M_{24}$ or $2^{23}.M_{23}$ symmetry in the odd Leech lattice CFT,
because $\cN=2$ SCA needs weight-$\frac{3}{2}$ supercurrents.

This requires a slight modification on the $\cN=2$ SCA proposed by \cite{Benjamin:2015ria}.
Their discussion on how their supercurrents satisfy the OPE of $\cN=2$ SCA is based on the assumption that the CFT would have the $M_{24}$ symmetry,
so we have to revisit that point,
but we can still realize an $\cN=2$ SCA which gives the $\cN=2$ extremal elliptic genus
by retaking the supercurrents as in Appendix \ref{sec:N=2-SCA-in-odd-Leech}.
However, since such supercurrents are no longer invariant under the $2^{24}.M_{24}$ or $2^{23}.M_{23}$ symmetry
(and also not invariant under the $M_{24}$ or $M_{23}$ symmetry even if we forget the OPE structure), 
the situation surrounding this moonshine phenomenon seems to be more mysterious than we thought.

\vspace{\vertspace}

The organization of the rest of this paper is as follows.
In Section \ref{sec:preliminaries}, we briefly review fundamental concepts such as codes, lattices, and lattice CFTs, fix their definitions and notations, and collect some data for later use.
In Section \ref{sec:cocycle-factor-and-aut-lattice-VOA},
we review how cocycle factors disturb the group structure of a lattice isometry group directly lifting up to that of the automorphism group of the lattice VOA.
Looking into the details of the group extension,
we paraphrase our main Question into more concrete forms Question \ref{q:lift-group-hom}, \ref{q:ses-split}, and \ref{q:tangible-form-Z2}, step by step,
and slightly generalize it to Question \ref{q:tangible-form}.
In Section \ref{sec:answer-to-the-question},
we negatively settle Question \ref{q:tangible-form} in the cases of the odd Leech lattice and the Leech lattice.
In Section \ref{sec:inv-current},
we investigate weight-1 and weight-$\frac{3}{2}$ currents of the odd Leech lattice CFT invariant under the subgroups of its automorphism group.
In Section \ref{sec:discussions}, we revisit the observation by \cite{Benjamin:2015ria} in light of our results.
Appendix \ref{sec:group-extension} contains an elementary introduction to group extensions for readers not familiar with such topics.
Appendix \ref{sec:N=2-SCA-in-odd-Leech} explains how we can find supercurrents of the $\cN=2$ SCA in the odd Leech lattice CFT.
Appendices \ref{sec:data-of-eqs} and \ref{sec:proof-of-subgroup-of-extension} provide some details we omitted in the main part.

The Python codes used in analysis are also uploaded to arXiv as ancillary files,
in the form of Jupyter Notebook (.ipynb file).
These codes are used in the analysis of the following sections:
\vspace{-8pt}
\begin{itemize}
\setlength{\itemsep}{-5pt}
\leftskip -15pt
\item  Section \ref{subsec:answer-for-odd-Leech}: \texttt{M24\_of\_odd\_Leech.ipynb, M23\_of\_odd\_Leech.ipynb}.
\item Section \ref{subsec:answer-for-Leech}: \texttt{Co0\_of\_Leech.ipynb}.
\item Appendix \ref{sec:N=2-SCA-in-odd-Leech}: \texttt{check\_supercurrent\_OPE.ipynb, find\_supercurrent.ipynb}.
\end{itemize}

\section{Codes, Lattices, and CFTs}
\label{sec:preliminaries}
In this Section \ref{sec:preliminaries},
we briefly review fundamental concepts such as codes, lattices, and lattice CFTs,
partly in order to fix definitions and notations;
we mainly follow \cite{MR1662447}.
Readers already familiar with these topics can skip this section
and come back just to refer to some data collected here for later analysis if necessary.

Before proceeding, let us clarify the convention of permutation groups and its actions,
and some notations of groups here.

\vspace{\vertspace}

\noindent \textbf{Convention of permutation groups and its actions}

The symmetric group $S_n$ consists of permutations $\sigma:\Omega_n\to \Omega_n$ where $\Omega_n=\{1,\ldots,n\}$,
and there are two conventions for the definition of the multiplication:
\begin{align}
& (1,2) \cdot (1,3) = (1,3,2); \  1 \mapsto 3, 3 \mapsto 2, 2 \mapsto 1, \label{eq:mult-of-perm}\\
& (1,2) \tildedot (1,3) = (1,2,3); \  1 \mapsto 2, 2 \mapsto 3, 3 \mapsto 1. \label{eq:mult-of-perm-op}
\end{align}
They are related as $\tau \cdot \sigma = \sigma \tildedot \tau$.

In this paper, we will adopt the multiplication $\cdot$ (\ref{eq:mult-of-perm}) and the left action of $\sigma\in S_n$ on a vector $k=(k_1,\ldots,k_n)\in\R^n$ as
\begin{align}
& \sigma(k) = (k_{\sigma^{-1}(1)}, \ldots, k_{\sigma^{-1}(n)}).
\end{align}
This can also be written as $(\sigma(k))_{\sigma(i)}=k_i$.

In GAP \cite{GAP4} and the webpage of ATLAS of Finite Group Representations \cite{ATLASweb}, on the other hand, they adopt the multiplication $\tildedot$ (\ref{eq:mult-of-perm-op}).
So when we cite equations from them,
we will convert them into representations in terms of the multiplication $\cdot$ (\ref{eq:mult-of-perm}) in this paper.

\vspace{\vertspace}

\noindent \textbf{Notations of groups}

For groups $G$ and $N$, following ATLAS of Finite Groups \cite{MR827219}, we write $N.G$ for any extension of $G$ by $N$.
\begin{align}
1 \to N \to N.G \to G \to 1.
\end{align}
In addition, we write $N:G$ or $N\rtimes G$ if it is a split extension, or a semidirect product,
and $N\times G$ if a direct product.
$R^\times$ denotes the multiplicative group of a ring $R$.

\subsection{Binary Golay Code and Mathieu Groups}
\label{subsec:Golay-code-and-Mathieu-group}
A \emph{$q$-ary linear code} $C$ of \emph{dimension} $m$ and \emph{length} $n$ is an $m$-dimensional subspace of the $n$-dimensional $\F_q$-linear space $(\F_q)^n$,
where $q$ is a prime or a prime power and $\F_q$ is the finite field of order $q$.
In this paper, codes always refer to linear ones.
An element of a code is called a \emph{codeword},
and the \emph{(Hamming) weight} of a codeword $w=(w_0,\ldots,w_{n-1})$ is $\wt(w):=|\{i \mid w_i\neq0\}|$.
When the minimal nonzero weight $\min_{w\in C\setminus\{0\}}\wt(w)$ is $d$,
this linear code is denoted by $[n, m, d]_q$.
Two linear codes are said to be \emph{equivalent} when one is mapped to the other by a \emph{monomial} matrix,
which is a matrix containing exactly one nonzero element of $\F_q$ in each row and column.
An equivalent map from a code to itself is called an \emph{automorphism},
and the set of all the automorphisms of $C$ forms the \emph{automorphism group} $\Aut(C)$ of $C$.

When $q=p^a$ with $p$ prime, the \emph{dual code} $C^\ast$ of $C$ is defined by
\begin{align}
C^\ast := \{v\in(\F_q)^n \mid v\cdot\bar{w}=0\ \text{for any $w\in C$}\},
\end{align}
where $\bar{w}:=((w_0)^p,\ldots,(w_{n-1})^p)$ is the \emph{conjugate} of $w$,
and $v \cdot w:=\sum_iv_iw_i$.
A code $C$ is said to be \emph{self-dual} if $C^\ast=C$.
Since $\dim C^\ast = n - m$,
the length of a self-dual code must be even.

A binary code is said to be \emph{even} if the weight of any codeword is even,
and \emph{doubly-even} if a multiple of 4.
An even but not doubly-even code is said to be \emph{singly-even}.
A singly-even self-dual code is sometimes called \emph{Type I},
and a doubly-even self-dual code \emph{Type II}.
The length of a Type II code must be a multiple of 8 \cite[Ch.\ 7 \S6 Cor.\ 18]{MR1662447}.

\vspace{\vertspace}

\noindent \textbf{The binary Golay code}

An \emph{(extended) binary Golay code} $G_{24}$ is the unique linear code $[24, 12, 8]_2$ up to code equivalence.
This code is doubly-even and self-dual.
A basis of $G_{24}$ can be read off from the proof of \cite[Ch.\ 10 \S 2.1 Thm.\ 7]{MR1662447} as
\begin{align}
\begin{array}{r}
(1,1,1,1,1, 0,0,1,0,0, 1,0,1,0,0, 0,0,0,0,0, 0,0,0,0), \\
(0, 1,1,1,1,1, 0,0,1,0,0, 1,0,1,0,0, 0,0,0,0,0, 0,0,0), \\
(0,0, 1,1,1,1,1, 0,0,1,0,0, 1,0,1,0,0, 0,0,0,0,0, 0,0), \\
(0,0,0, 1,1,1,1,1, 0,0,1,0,0, 1,0,1,0,0, 0,0,0,0,0, 0), \\
(0,0,0,0, 1,1,1,1,1, 0,0,1,0,0, 1,0,1,0,0, 0,0,0,0,0), \\
(0, 0,0,0,0, 1,1,1,1,1, 0,0,1,0,0, 1,0,1,0,0, 0,0,0,0), \\
(0,0, 0,0,0,0, 1,1,1,1,1, 0,0,1,0,0, 1,0,1,0,0, 0,0,0), \\
(0,0,0, 0,0,0,0, 1,1,1,1,1, 0,0,1,0,0, 1,0,1,0,0, 0,0), \\
(0,0,0,0, 0,0,0,0, 1,1,1,1,1, 0,0,1,0,0, 1,0,1,0,0, 0), \\
(0,0,0,0,0, 0,0,0,0, 1,1,1,1,1, 0,0,1,0,0, 1,0,1,0,0), \\
(0, 0,0,0,0,0, 0,0,0,0, 1,1,1,1,1, 0,0,1,0,0, 1,0,1,0), \\
(1,1,1,1,1, 1,1,1,1,1, 1,1,1,1,1, 1,1,1,1,1, 1,1,1,1).
\end{array}
\label{eq:G24-basis}
\end{align}
The 24 columns are labeled by $0, 1, \ldots, 22, \infty$ in \cite{MR1662447},
and we will write $23$ instead of $\infty$ below.

A binary Golay code only has codewords with weight 0, 8, 12, 16, and 24.
Codewords with weight 8 and 12 are called an \emph{octad} and a \emph{dodecad}, respectively.

\vspace{\vertspace}

\noindent \textbf{The largest Mathieu group $M_{24}$}

The automorphism group $\Aut(G_{24})$ of a binary Golay code is the largest \emph{Mathieu group} $M_{24}$.
It is a sporadic simple group which is 5-transitive as a subgroup of $S_{24}$ acting on 24 points.
In the above basis (\ref{eq:G24-basis}), $M_{24}$ is generated by the following four permutations \cite[Ch. 10 \S2.1]{MR1662447}:
\begin{align}
\begin{array}{l}(0,1,2,3,4,5,6,7,8,9,10,11,12,13,14,15,16,17,18,19,20,21,22),\\
(15,7,14,5,10,20,17,11,22,21,19)(3,6,12,1,2,4,8,16,9,18,13),\\
(23,0)(15,3)(7,13)(14,18)(5,9)(10,16)(20,8)(17,4)(11,2)(22,1)(21,12)(19,6),\\
(14,17,11,19,22)(20,10,7,5,21)(18,4,2,6,1)(8,16,13,9,12).
\end{array}
\label{eq:M24-generator-from-CS}
\end{align}

A presentation of $M_{24}$ can be found in \cite{ATLASweb} as
\begin{align}
\hspace{-10pt}\scalebox{0.77}{$M_{24} = \langle a, b \mid a^2 = b^3 = (ba)^{23} = [b^{-1},a]^{12} = [(bab)^{-1},a]^5 = (b^{-1}ab^{-1}aba)^3(b^{-1}ababa)^3 = ((b^{-1}aba)^3ba)^4 = 1 \rangle$}, & \label{eq:presentation-of-M24}
\end{align}
where $[y^{-1},x^{-1}]=y^{-1}x^{-1}yx$ and note that the multiplication (\ref{eq:mult-of-perm-op}) in \cite{ATLASweb} is converted to (\ref{eq:mult-of-perm}) here.
In $M_{24}$ generated by (\ref{eq:M24-generator-from-CS}), these generators $a$ and $b$ can be taken as\footnote{
These generators (\ref{eq:specific-generator-of-M24-a}) and (\ref{eq:specific-generator-of-M24-b}) were found
in the result of the GAP \cite{GAP4} command
\texttt{IsomorphismGroups(MathieuGroup(24), G)},
where \texttt{G} is declared as a \texttt{Group} generated by the permutations (\ref{eq:M24-generator-from-CS}).
This command returns one explicit isomorphism in the form of a map between generators,
so it suffices to check the relations in (\ref{eq:presentation-of-M24}) for the displayed generators.
}
\begin{align}
& \scalebox{0.9}{$a = (0,7)(1,6)(2,18)(3,22)(4,19)(5,14)(8,15)(9,10)(11,21)(12,17)(13,23)(16,20),$} \label{eq:specific-generator-of-M24-a}\\
& \scalebox{0.9}{$b = (1,8,11)(2,12,23)(3,16,20)(4,15,14)(7,9,22)(10,13,21).$} \label{eq:specific-generator-of-M24-b}
\end{align}

\vspace{\vertspace}

\noindent \textbf{The second largest Mathieu group $M_{23}$}

The Mathieu Group $M_{23}$ is the stabilizer group of one point of the action of $M_{24}$ on 24 points.
Since $M_{24}$ is transitive, $M_{23}$ is unique up to isomorphism, regardless of the choice of the stabilized point.
It is also a sporadic simple group.

A presentation of $M_{23}$ can be found in \cite{ATLASweb} as
\begin{align}
\scalebox{0.9}{$M_{23} = \langle a, b \mid$} & \ \scalebox{0.9}{$a^2 = b^4 = (ba)^{23} = (b^2a)^6 = [b^{-1},a]^6 = (b^2ab^{-1}aba)^4 = 1,$} \nonumber \\
& \scalebox{0.9}{$(b^{-1}a)^3(ba)^3(b^{-1}aba)^2b^2ab^{-1}a(ba)^3 = b^2ab^{-1}abab^2aba(b^{-1}ab^2a)^2(b^2aba)^3 = 1 \rangle$}. \label{eq:presentation-of-M23}
\end{align}
If we take $M_{23}$ as the subgroup of $M_{24}$ generated by (\ref{eq:M24-generator-from-CS}) stabilizing the point 23,
these generators $a$ and $b$ can be taken as
\begin{align}
a & = (0,14)(1,19)(2,11)(3,18)(4,16)(6,12)(7,17)(10,13), \\
b & = (0,11,16,20)(1,5,2,7)(3,8,15,19)(4,22)(9,14,21,12)(10,18).
\end{align}

\subsection{Odd Leech Lattice, Leech Lattice, and Conway Groups}
\label{subsec:Leech-lattice}
A \emph{lattice} of rank $n$ is a free abelian group $L$ of rank $n$ whose basis is an $\R$-basis of a vector space $\R^n$ with a symmetric bilinear form $(-,-):L\times L\to\R$.
Such a lattice is denoted by the pair $(L,(-,-))$, or just $L$, and naturally regarded as a subset of $\R^n$.
A vector $k$ in $L$ is sometimes called a \emph{lattice point} of $L$,
and we also write the product $(k,k')$ of two vectors $k,k'\in L$ given by the symmetric bilinear form as $k\cdot k'$.
If any vectors $k,k'\in L$ satisfy $k\cdot k'\in\Z$, then the lattice is said to be \emph{integral}.
The \emph{squared length} of a vector $k\in L$ is defined as $|k|^2:=k\cdot k$.
A vector $k\in L$ with $|k|^2\in 2\Z$ is called an \emph{even vector},
and $k\in L$ with $|k|^2\in 2\Z+1$ is called an \emph{odd vector}.
An integer lattice is said to be \emph{even} if its vectors are all even,
and \emph{odd} otherwise.
An \emph{isometry} or \emph{isomorphism} $g:L\to L'$ of lattices $L$ and $L'$ is an isomorphism of free abelian groups compatible with their symmetric bilinear forms.
The group of all the isometries $L \to L$, or the \emph{automorphisms} of $L$, is denoted by $\Aut(L)$ or $O(L)$,
but note that $\Aut(L)$ sometimes denotes the automorphism group of just a free abelian group $L$ (e.g.\ Section \ref{subsec:review-of-FLM}).

The \emph{dual lattice} $L^\ast$ of $L$ is defined by
\begin{align}
L^\ast := \{l\in\R^n \mid (l,k)\in\Z \ \text{for any $k \in L$}\},
\end{align}
with the same symmetric bilinear form as that of $L$.
$L$ is integral if and only if $L\subset L^\ast$.
$L$ is said to be \emph{self-dual} or \emph{unimodular} if $L^\ast=L$.
An odd self-dual lattice is called \emph{Type I},
and an even self-dual lattice is called \emph{Type II}.

\vspace{\vertspace}

\noindent \textbf{Lattices from codes}

There are several ways to construct a lattice from a given code.
We focus on the case of binary codes here.
See \cite{MR1662447} for more on constructions.

Let $C \subset (\F_2)^n$ be a binary code of length $n$,
and regard each codeword of $C$ as a vector with entries 0 or 1.
We can construct a lattice $\Lambda(C)$ of rank $n$ as
\begin{align}
\Lambda(C) := \frac{1}{\sqrt{2}}C + \sqrt{2}\Z^n \subset \R^n,
\end{align}
with the standard Euclidean metric $k\cdot k'=\sum_ik_ik'_i$ for $k=(k_i)_i,k'=(k'_i)_i\in\R^n$ as a symmetric bilinear form.
This construction is called Construction A \cite[Ch.\ 7 \S2]{MR1662447}.
This lattice $\Lambda(C)$ satisfies $\Lambda(C^\ast)=\Lambda(C)^\ast$,
and hence $\Lambda(C)$ is integer if and only if $C$ satisfies $C\subset C^\ast$,
and self-dual if and only if $C$ is self-dual.
In addition, $\Lambda(C)$ is of Type I (odd self-dual) if and only if $C$ is of Type I (singly-even self-dual),
and Type II (even self-dual) if and only if $C$ is of Type II (doubly-even self-dual).

Another construction called the Construction B \cite[Ch.\ 7 \S5]{MR1662447} associates the sublattice
\begin{align}
\Lambda_\mathrm{B}(C) := \{k=(k_0,\ldots,k_{n-1})\in\Lambda(C) \mid \sqrt{2}\sum_{i=0}^{n-1}k_i \in 4\Z\}
\end{align}
of $\Lambda(C)$ to a binary code $C$.

For a doubly-even self-dual binary code $C$, we further consider the following constructions \cite[\S5.1]{Dolan:1994st}.
We first define
\begin{align}
& \Z^n_+ := \{x\in\Z^n \mid |x|^2 \in 2\Z\}, \\
& \Z^n_- := \{x\in\Z^n \mid |x|^2 \in 2\Z+1\}.
\end{align}
Recall that the length $n$ of a doubly-even self-dual code is always a multiple of 8,
and also define
\begin{align}
& \Lambda_0(C) := \frac{1}{\sqrt{2}} C + \sqrt{2} \Z^n_+, \label{eq:def-of-Lambda_0(C)} \\
& \Lambda_1(C) := \frac{1}{\sqrt{2}} C + \sqrt{2} \Z^n_-, \\
& \Lambda_2(C) := \frac{1}{2\sqrt{2}}\underline{1} + \frac{1}{\sqrt{2}} C + \sqrt{2} \Z^n_{(-)^{\frac{n}{8}+1}}, \\
& \Lambda_3(C) := \frac{1}{2\sqrt{2}}\underline{1} + \frac{1}{\sqrt{2}} C + \sqrt{2} \Z^n_{(-)^{\frac{n}{8}}}, \label{eq:def-of-Lambda_3(C)}
\end{align}
where $\underline{1}:=(1,\ldots,1)$.
Then we can construct some lattices as in Table \ref{tbl:lettice-constructions}.

\vspace{\vertspace}

\begin{table}
\caption{Lattices constructed from a doubly-even self-dual binary code $C$.
On the right side of the vertical line,
the names of lattices constructed from the binary Golay code $G_{24}$ and their isometry groups are shown.}
\hspace{-20pt}\begin{tabular}{lll|ll}
name of construction & lattice & property & for $C=G_{24}$ & Aut(lattice) \\\hline
Construction A & $\Lambda(C)=\Lambda_0(C)\cup\Lambda_1(C)$ & even self-dual & $(A_1)^{24}$ & $2^{24}:M_{24}$\\
Construction B & $\Lambda_\mathrm{B}(C)=\Lambda_0(C)$ & even & &\\
twisted construction & $\tilde{\Lambda}(C)=\Lambda_0(C)\cup\Lambda_3(C)$ & even self-dual & Leech lattice $\Lambda_{24}$ & $\Co_0$\\
$-$ & $\tilde{\Lambda}'(C)=\Lambda_0(C)\cup\Lambda_2(C)$ & odd self-dual & odd Leech lattice $O_{24}$ & $2^{12}:M_{24}$
\end{tabular}
\label{tbl:lettice-constructions}
\end{table}

\noindent \textbf{Odd Leech lattice $O_{24}$}

The \emph{odd Leech lattice} $O_{24}$ is the unique odd self-dual lattice of rank 24 without roots (vectors with squared length 2) up to isometry.
It can be constructed from the binary Golay code $G_{24}$ as
\begin{align}
O_{24} = (\frac{1}{\sqrt{2}}G_{24} + \sqrt{2}\Z^{24}_+) \cup (\frac{1}{2\sqrt{2}}\underline{1} + \frac{1}{\sqrt{2}}G_{24} + \sqrt{2}\Z^{24}_+).
\label{eq:odd-Leech-from-Golay-code}
\end{align}

If we use the binary Golay code $G_{24}$ with the basis (\ref{eq:G24-basis}),
then we can take a $\Z$-basis $e_0,\ldots,e_{23}$ of (\ref{eq:odd-Leech-from-Golay-code}) as
\begin{align}
\begin{array}{rl}
e_0 & = \frac{1}{2\sqrt{2}}\underline{1}, \\
e_1 & = \frac{1}{\sqrt{2}}\text{(the second line from the bottom of (\ref{eq:G24-basis}))}, \\
& \ \vdots \\
e_{11} & = \frac{1}{\sqrt{2}}\text{(the first line of (\ref{eq:G24-basis}))}, \\
e_{12} & = \sqrt{2}(1,0,0,0,0, 0,0,0,0,0, 0,1,0, \ldots, 0), \\
e_{13} & = \sqrt{2}(1,0,0,0,0, 0,0,0,0,0, 1,0,0, \ldots, 0), \\
& \ \vdots \\
e_{22} & = \sqrt{2}(1,1,0,\ldots,0), \\
e_{23} & = \sqrt{2}(2,0,\ldots,0).
\end{array}
\label{eq:odd-Leech-basis}
\end{align}
In fact, $O_{24}$ in (\ref{eq:odd-Leech-from-Golay-code}) obviously contains $\Span_\Z\{e_i\}_i$,
and the opposite direction of the inclusion can be checked as follows.
Since it is obvious that $\Span_\Z\{e_i\}_i$ contains $\frac{1}{2\sqrt{2}}\underline{1}$ and $\frac{1}{\sqrt{2}}G_{24}$,
it suffices to check that it also contains $\sqrt{2}\Z^{24}_+$.
Here, $\Z^{24}_+$ is generated by the 24 vectors
\begin{align*}
&(1,0,\ldots,0,0,1),\\
&(1,0,\ldots,0,1,0),\\
&\ \ \vdots\\
&(1,1,0,\ldots,0),\\
&(2,0,\ldots,0).
\end{align*}
Since the last 12 vectors of $\sqrt{2}\times$(above 24 vectors) are exactly $e_{12},\ldots,e_{23}$,
it suffices to check that the first 12 vectors of $\sqrt{2}\times$(above 24 vectors) can be written as $\Z$-linear combinations of $e_0,\ldots,e_{23}$,
which can be checked by computer.

The isometry group $\Aut(O_{24})$ of the odd Leech lattice is known to be $2^{12}:M_{24}$ \cite[Ch.\ 17]{MR1662447}.
In the construction (\ref{eq:odd-Leech-from-Golay-code}),
these automorphisms are apparent because $M_{24}=\Aut(G_{24})$ and $2^{12}$ are the maps $k=(k_i)_i\mapsto((-1)^{w_i}k_i)_i$ where $w=(w_i)_i\in G_{24}$.

If we use Construction A for ternary codes, an odd Leech lattice can also be constructed from any self-dual ternary code of length 24 with the minimal nonzero weight 9.
It is known that there are only two such ternary codes up to equivalence \cite{MR633414}:
the extended quadratic residue code $Q_{24}$ (see e.g. \cite[Ch.\ 16]{MR465510}),
and the symmetric code $P_{24}$ defined by Pless \cite{MR245455, MR290865}.
However, the automorphism groups of these ternary codes are $\Aut(Q_{24})=(\Z_3)^\times\cdot\mathrm{PSL}(23,\F_2)$ \cite[Ch.\ 16 \S5]{MR465510} and $\Aut(P_{24})=\Z_4\cdot\mathrm{PGL}(11,\F_2)$ \cite[\S5.2]{MR441541},
so the structure of $\Aut(O_{24})=2^{12}:M_{24}$ is not apparent in these constructions of odd Leech lattices.
See \cite[Example 4.5]{Gaiotto:2018ypj} for application of these constructions to $\cN=1$ supersymmetries of the lattice VOAs.


\vspace{\vertspace}

\noindent \textbf{Leech lattice $\Lambda_{24}$}

The \emph{Leech lattice} $\Lambda_{24}$ is the unique even self-dual lattice of rank 24 without roots (vectors with squared length 2) up to isometry.
It can be constructed from the binary Golay code $G_{24}$ as
\begin{align}
\Lambda_{24} = (\frac{1}{\sqrt{2}}G_{24} + \sqrt{2}\Z^{24}_+) \cup (\frac{1}{2\sqrt{2}}\underline{1} + \frac{1}{\sqrt{2}}G_{24} + \sqrt{2}\Z^{24}_-).
\label{eq:Leech-from-Golay-code}
\end{align}

\vspace{\vertspace}

\noindent \textbf{The largest Conway group $\Co_0$}

The isometry group $\Aut(\Lambda_{24})$ of the Leech lattice is the largest \emph{Conway group} $\Co_0$.
It is not a simple group, but its quotient by the center $\{\pm1\}$ is a sporadic simple group called the second largest Conway group $\Co_1$.

Generators of $\Co_0$ as a subgroup of $\GL(24, \Z)$ can be found on the $\Co_1$ page of \cite{ATLASweb} as
\begin{align}
A & = \scalebox{0.5}{$\left(\begin{array}{rrrrrrrrrrrrrrrrrrrrrrrr}
2 & 0 & 0 & -3 & -2 & -1 & 0 & -2 & 1 & 1 & 0 & 1 & 1 & 0 & 0 & 1 & -1 & -1 & 0 & 0 & -1 & 0 & 0 & 0 \\
-4 & 1 & 0 & 6 & 4 & 2 & 2 & 3 & -1 & -3 & -1 & 0 & -2 & -1 & -1 & -1 & 1 & 1 & 0 & 1 & 0 & 0 & 0 & 0 \\
4 & -1 & 0 & -6 & -4 & -2 & -3 & -4 & 2 & 4 & 3 & -3 & 2 & 1 & 2 & 0 & 1 & -1 & -1 & -1 & 1 & 0 & -1 & 1 \\
4 & 1 & 1 & -6 & -3 & -2 & -1 & -4 & 2 & 2 & 1 & 0 & 1 & 0 & 1 & 2 & 0 & -1 & -1 & 0 & -1 & 0 & 0 & 1 \\
5 & -2 & 0 & -8 & -5 & -3 & -2 & -4 & 2 & 4 & 2 & 0 & 2 & 2 & 1 & 1 & -1 & -1 & 0 & 0 & 0 & 1 & -1 & 0 \\
-3 & 0 & -1 & 4 & 1 & 2 & 1 & 1 & 0 & -2 & 0 & 0 & 0 & -1 & -1 & -1 & 0 & 0 & 0 & 0 & -1 & -1 & 0 & 0 \\
1 & 0 & -1 & -2 & -2 & -1 & 0 & -2 & 0 & 1 & 1 & 0 & 0 & 0 & 0 & 0 & -1 & 0 & 0 & -1 & 0 & 0 & 0 & 0 \\
-3 & 1 & -1 & 5 & 3 & 1 & 2 & 2 & -1 & -2 & -1 & 0 & -1 & -1 & -1 & -1 & 0 & 0 & 1 & 0 & 0 & 0 & 0 & 0 \\
-2 & 2 & 0 & 3 & 3 & 0 & 2 & 2 & 0 & -1 & -2 & 1 & 0 & -1 & 0 & 0 & 1 & -1 & 1 & 1 & 0 & 0 & 0 & 1 \\
9 & -1 & -1 & -13 & -9 & -5 & -4 & -9 & 4 & 6 & 4 & -1 & 3 & 1 & 2 & 2 & -1 & -3 & -1 & -1 & -2 & 1 & -1 & 1 \\
1 & -1 & -1 & -3 & -3 & -1 & -1 & -4 & 2 & 2 & 3 & -2 & 1 & 0 & 0 & 0 & 0 & 0 & -1 & -1 & 0 & 0 & -1 & 1 \\
4 & -1 & 1 & -6 & -4 & -2 & -2 & -3 & 2 & 3 & 2 & -1 & 2 & 1 & 1 & 1 & 0 & -1 & -1 & 0 & 0 & 0 & -1 & 1 \\
-4 & 0 & 0 & 5 & 3 & 2 & 2 & 3 & -2 & -2 & -1 & 0 & -2 & 0 & -1 & -1 & 0 & 2 & 0 & 0 & 1 & 0 & 0 & -1 \\
8 & 2 & 2 & -9 & -5 & -4 & 0 & -5 & 2 & 2 & 2 & 1 & 0 & 1 & 1 & 2 & -1 & 0 & -1 & 1 & -1 & 1 & 0 & 1 \\
-3 & 0 & 0 & 3 & 3 & 1 & 0 & 3 & -1 & -1 & -1 & 0 & 0 & 0 & 0 & 0 & 1 & 0 & 1 & 0 & 1 & -1 & 0 & 0 \\
-5 & 2 & 0 & 6 & 4 & 3 & 3 & 2 & -1 & -4 & -2 & 1 & -2 & -2 & -2 & 0 & 0 & 1 & 0 & 0 & -1 & -1 & 1 & 0 \\
3 & 0 & 0 & -3 & -3 & -1 & 0 & -2 & 0 & 1 & 1 & 0 & 0 & 0 & 0 & 0 & -1 & 0 & -1 & -1 & 0 & 0 & 0 & 0 \\
0 & 0 & 0 & 1 & 1 & 0 & 0 & 2 & -1 & 0 & -2 & 1 & 0 & 0 & 0 & 0 & 0 & -1 & 1 & 0 & 0 & 0 & 0 & -1 \\
-1 & 0 & -1 & 2 & 0 & 1 & -1 & 0 & 0 & 0 & 0 & -1 & 1 & -1 & 0 & -1 & 1 & -1 & 0 & -1 & 0 & -1 & 0 & 0 \\
-4 & -1 & -1 & 4 & 2 & 2 & 0 & 2 & -1 & -2 & -1 & 1 & -1 & -1 & -1 & 0 & 0 & 1 & 0 & 0 & -1 & 0 & 1 & -1 \\
-9 & 1 & 0 & 11 & 7 & 5 & 3 & 6 & -2 & -5 & -3 & 1 & -2 & -2 & -2 & -1 & 1 & 2 & 0 & 1 & 0 & -1 & 1 & 0 \\
-6 & 1 & 0 & 9 & 5 & 4 & 3 & 4 & -2 & -4 & -1 & -1 & -2 & -1 & -1 & -2 & 1 & 2 & 0 & 0 & 1 & -1 & 0 & 0 \\
5 & 1 & -1 & -7 & -5 & -3 & -1 & -6 & 2 & 3 & 3 & 0 & 1 & 0 & 1 & 1 & -1 & 0 & -1 & -1 & -1 & 1 & 0 & 1 \\
6 & 0 & 1 & -6 & -4 & -2 & -2 & -3 & 1 & 2 & 1 & 0 & 2 & 1 & 1 & 1 & -1 & -1 & 0 & 0 & -1 & 0 & 0 & 0
\end{array}\right)$}^T,
\label{eq:Co0-generator-A}
\end{align}
\begin{align}
B & = \scalebox{0.5}{$\left(\begin{array}{rrrrrrrrrrrrrrrrrrrrrrrr}
0 & 0 & 1 & 0 & 0 & 0 & -1 & 1 & 1 & 0 & -1 & 0 & 1 & 0 & 1 & 0 & 1 & 0 & 0 & 1 & 0 & 1 & 0 & 1 \\
3 & 1 & -1 & -3 & -2 & -2 & 0 & -3 & 1 & 2 & 2 & -1 & 1 & 0 & 1 & 0 & 0 & -1 & 0 & -1 & 0 & -1 & -1 & 0 \\
-4 & 0 & -1 & 7 & 4 & 3 & 1 & 4 & -3 & -3 & -2 & 0 & -1 & -1 & -1 & -1 & 0 & 0 & 1 & -1 & 0 & -1 & 1 & -1 \\
0 & -1 & 0 & 1 & 0 & 1 & -1 & 1 & 0 & 0 & 0 & -1 & 1 & 0 & 1 & -1 & 0 & 0 & 0 & 0 & 0 & 0 & 0 & 0 \\
6 & 0 & 2 & -6 & -4 & -2 & -3 & -2 & 1 & 2 & 1 & -1 & 1 & 1 & 2 & 1 & 1 & -1 & -1 & 0 & 0 & 1 & 0 & 1 \\
2 & 0 & -1 & -5 & -3 & -3 & -1 & -4 & 3 & 3 & 1 & 0 & 2 & 0 & 1 & 1 & 0 & -1 & 0 & 0 & 0 & 1 & -1 & 1 \\
-10 & 0 & 0 & 12 & 8 & 5 & 4 & 7 & -3 & -6 & -3 & 1 & -3 & -1 & -3 & -1 & 0 & 3 & 1 & 1 & 1 & -1 & 1 & -1 \\
-1 & 4 & 1 & 5 & 4 & 1 & 3 & 3 & -2 & -3 & -2 & 1 & -1 & -1 & 0 & 0 & 1 & 1 & 1 & 0 & 0 & -1 & 1 & 0 \\
-1 & 0 & 1 & 3 & 2 & 1 & -1 & 3 & -1 & -1 & -2 & 0 & 1 & 0 & 1 & 0 & 1 & 0 & 1 & 0 & 0 & 0 & 1 & 0 \\
2 & 2 & 1 & 0 & 0 & 0 & 0 & 0 & 0 & 0 & 0 & -1 & 1 & 0 & 2 & 0 & 1 & 0 & 0 & 0 & 0 & 0 & 0 & 1 \\
-5 & 3 & 0 & 11 & 7 & 4 & 4 & 6 & -4 & -5 & -3 & 0 & -2 & -2 & -1 & -2 & 1 & 1 & 1 & -1 & 1 & -2 & 1 & -1 \\
3 & 0 & 1 & -2 & -1 & -1 & -1 & 0 & 0 & 1 & -1 & 0 & 1 & 0 & 1 & 0 & 0 & -1 & 0 & 0 & 0 & 1 & 0 & 0 \\
-4 & 0 & -1 & 5 & 3 & 2 & 2 & 2 & -1 & -3 & 0 & 0 & -2 & -1 & -2 & -1 & 0 & 1 & 0 & 0 & 0 & -1 & 0 & -1 \\
2 & -1 & 1 & -2 & -1 & -1 & -1 & 0 & 1 & 1 & 1 & -1 & 1 & 1 & 1 & 0 & 0 & 0 & 0 & 1 & 1 & 0 & -1 & 0 \\
2 & -1 & 1 & -3 & -1 & -1 & -1 & 0 & -1 & 1 & -1 & 1 & -1 & 1 & 0 & 1 & -1 & 0 & 0 & 0 & 0 & 1 & 1 & -1 \\
-2 & 0 & -1 & 3 & 1 & 1 & 1 & 1 & 0 & -1 & 0 & 0 & 0 & -1 & 0 & -1 & 0 & 1 & 0 & 0 & 0 & 0 & 0 & 0 \\
-11 & -2 & -1 & 12 & 7 & 5 & 3 & 7 & -2 & -5 & -3 & 1 & -2 & -1 & -3 & -2 & 0 & 2 & 1 & 1 & 1 & 0 & 0 & -1 \\
-2 & 0 & 1 & 2 & 2 & 1 & 0 & 3 & -1 & -1 & -2 & 1 & -1 & 0 & 0 & 0 & 1 & 0 & 0 & 1 & 0 & 1 & 0 & 0 \\
-6 & 1 & -1 & 7 & 5 & 2 & 2 & 4 & -2 & -2 & -3 & 1 & -1 & -1 & -1 & -1 & 1 & 0 & 1 & 0 & 1 & 0 & 0 & 0 \\
-3 & 1 & -1 & 3 & 2 & 1 & 2 & 0 & 0 & -1 & 0 & 0 & -1 & -1 & -1 & -1 & 0 & 0 & 0 & 0 & 0 & -1 & 0 & 0 \\
2 & -2 & -2 & -6 & -5 & -2 & -2 & -6 & 3 & 4 & 3 & -1 & 2 & 0 & 0 & 0 & -1 & -1 & -1 & -1 & -1 & 0 & -1 & 0 \\
2 & -1 & -2 & -4 & -4 & -1 & -1 & -5 & 2 & 2 & 4 & -2 & 1 & 0 & 0 & 0 & -1 & 0 & -1 & -2 & 0 & -1 & -1 & 0 \\
-6 & 1 & 0 & 8 & 5 & 4 & 3 & 3 & -2 & -4 & -1 & 0 & -2 & -1 & -2 & -1 & 0 & 2 & 0 & 0 & 0 & -2 & 1 & -1 \\
14 & 1 & 2 & -18 & -11 & -7 & -4 & -10 & 4 & 7 & 4 & 0 & 3 & 2 & 3 & 3 & -1 & -2 & -2 & 0 & -1 & 2 & -1 & 2
\end{array}\right)$}^T,
\label{eq:Co0-generator-B}
\end{align}
where the transpose on the right-hand side is taken just for the convenience of notation.
These generators satisfy $A^2=-1$ and $B^3=1$.

The Conway group $\Co_0$ generated by these generators $A$ and $B$ is the isometry group of a lattice $(\Z^{24},q)$, 
where $q$ is some symmetric bilinear form which makes $(\Z^{24},q)$ isomorphic to the Leech lattice.
The matrix form of $q$ up to scalar multiplication can be found by solving the conditions
\begin{align}
A^T q A = q, \quad B^T q B = q,
\end{align}
by Mathematica \cite{Mathematica}.
Note that the elements of the Conway group $\Co_0$ in this notation act on the column vectors in the lattice $(\Z^{24},q)$ from the left.
The result is\footnote{
The author thanks Yuji Tachikawa for providing information on this symmetric bilinear form $q$.
}
\begin{align}
q & = \scalebox{0.5}{$\left(\begin{array}{rrrrrrrrrrrrrrrrrrrrrrrr}
4 & -2 & -2 & 2 & 2 & 2 & -1 & -1 & 2 & 2 & -2 & 2 & -2 & 1 & -1 & 1 & 1 & 2 & 0 & -1 & 0 & -2 & -2 & 2 \\
-2 & 4 & 0 & 0 & -2 & 0 & -1 & 2 & 0 & 0 & 2 & -1 & 1 & 1 & -1 & 1 & -2 & -2 & -1 & 1 & 1 & 2 & 1 & -2 \\
-2 & 0 & 4 & 0 & 0 & -2 & 0 & 0 & 0 & 0 & 2 & 0 & 0 & -2 & 1 & -1 & -1 & -1 & 1 & -1 & -1 & 1 & 0 & -1 \\
2 & 0 & 0 & 4 & 0 & 0 & -2 & -1 & 2 & 2 & 0 & 2 & -2 & 1 & 0 & 2 & -1 & 0 & -1 & 0 & 1 & -1 & -1 & 0 \\
2 & -2 & 0 & 0 & 4 & 0 & -1 & -1 & 1 & 1 & -1 & 1 & -1 & 0 & 0 & -1 & 0 & 2 & 0 & -2 & -2 & -1 & -2 & 2 \\
2 & 0 & -2 & 0 & 0 & 4 & 0 & 0 & 1 & 1 & -1 & 1 & -1 & 1 & -1 & 1 & 1 & 0 & 1 & 0 & 1 & 0 & -1 & 1 \\
-1 & -1 & 0 & -2 & -1 & 0 & 4 & 1 & -1 & -1 & 0 & -2 & 2 & 0 & 0 & 0 & 2 & -1 & 0 & 0 & -1 & 1 & 2 & -1 \\
-1 & 2 & 0 & -1 & -1 & 0 & 1 & 4 & 1 & 1 & 2 & -1 & 1 & 0 & -1 & 1 & -1 & -1 & -1 & -1 & -1 & 1 & 1 & -1 \\
2 & 0 & 0 & 2 & 1 & 1 & -1 & 1 & 4 & 2 & 0 & 2 & -2 & 1 & 0 & 1 & -1 & 0 & -1 & -2 & 0 & -1 & -1 & 0 \\
2 & 0 & 0 & 2 & 1 & 1 & -1 & 1 & 2 & 4 & 0 & 1 & -2 & 0 & -1 & 1 & -1 & 1 & 0 & -1 & -1 & -1 & -1 & 1 \\
-2 & 2 & 2 & 0 & -1 & -1 & 0 & 2 & 0 & 0 & 4 & 0 & 1 & -1 & 0 & 1 & -1 & -2 & 0 & 0 & 0 & 2 & 1 & -2 \\
2 & -1 & 0 & 2 & 1 & 1 & -2 & -1 & 2 & 1 & 0 & 4 & -2 & 0 & 1 & 0 & 0 & 1 & 0 & -1 & 1 & -2 & -2 & 1 \\
-2 & 1 & 0 & -2 & -1 & -1 & 2 & 1 & -2 & -2 & 1 & -2 & 4 & 0 & -1 & 0 & 1 & -1 & -1 & 1 & 0 & 2 & 2 & -2 \\
1 & 1 & -2 & 1 & 0 & 1 & 0 & 0 & 1 & 0 & -1 & 0 & 0 & 4 & -1 & 1 & 0 & -1 & -2 & 0 & 0 & 0 & 0 & 0 \\
-1 & -1 & 1 & 0 & 0 & -1 & 0 & -1 & 0 & -1 & 0 & 1 & -1 & -1 & 4 & -1 & -1 & 0 & 0 & 0 & 0 & -1 & 0 & 0 \\
1 & 1 & -1 & 2 & -1 & 1 & 0 & 1 & 1 & 1 & 1 & 0 & 0 & 1 & -1 & 4 & 0 & -1 & -1 & 0 & 1 & 1 & 0 & -1 \\
1 & -2 & -1 & -1 & 0 & 1 & 2 & -1 & -1 & -1 & -1 & 0 & 1 & 0 & -1 & 0 & 4 & 1 & 1 & 0 & 0 & 0 & 0 & 0 \\
2 & -2 & -1 & 0 & 2 & 0 & -1 & -1 & 0 & 1 & -2 & 1 & -1 & -1 & 0 & -1 & 1 & 4 & 1 & -1 & -1 & -2 & -2 & 2 \\
0 & -1 & 1 & -1 & 0 & 1 & 0 & -1 & -1 & 0 & 0 & 0 & -1 & -2 & 0 & -1 & 1 & 1 & 4 & 0 & 0 & 0 & -1 & 1 \\
-1 & 1 & -1 & 0 & -2 & 0 & 0 & -1 & -2 & -1 & 0 & -1 & 1 & 0 & 0 & 0 & 0 & -1 & 0 & 4 & 2 & 0 & 2 & -1 \\
0 & 1 & -1 & 1 & -2 & 1 & -1 & -1 & 0 & -1 & 0 & 1 & 0 & 0 & 0 & 1 & 0 & -1 & 0 & 2 & 4 & 0 & 1 & -1 \\
-2 & 2 & 1 & -1 & -1 & 0 & 1 & 1 & -1 & -1 & 2 & -2 & 2 & 0 & -1 & 1 & 0 & -2 & 0 & 0 & 0 & 4 & 1 & -2 \\
-2 & 1 & 0 & -1 & -2 & -1 & 2 & 1 & -1 & -1 & 1 & -2 & 2 & 0 & 0 & 0 & 0 & -2 & -1 & 2 & 1 & 1 & 4 & -2 \\
2 & -2 & -1 & 0 & 2 & 1 & -1 & -1 & 0 & 1 & -2 & 1 & -2 & 0 & 0 & -1 & 0 & 2 & 1 & -1 & -1 & -2 & -2 & 4
\end{array}\right)$}.
\label{eq:sym-bilin-form-of-Leech}
\end{align}

\subsection{Lattice CFT, \texorpdfstring{$\Z_2$}{Z2} Orbifold, and Fermionization}
\label{subsec:lattice-CFT}
We can construct a modular invariant bosonic CFT from an even self-dual lattice.
Its mathematical description as a lattice VOA will be reviewed in Section \ref{subsec:autom-of-lattice-VOA}.
Here, we only present a brief explanation in the language of physics.

If the lattice $L$ is Euclidean, that is, the symmetric bilinear form is positive-definite,
then the resulting CFT is chiral.
We only deal with such cases in this paper.
The states of such a CFT are $\C$-linear combinations of the states in the form of
\begin{align}
\alpha^{i_1}_{-m_1}\cdots\alpha^{i_l}_{-m_l}|k\rangle,
\label{eq:general-element-of-lattice-CFT}
\end{align}
where $|k\rangle$ is the state with momentum vector $k\in L$, $m_1,\ldots,m_l\in \Z_{>0}$, and $\alpha^0_m,\ldots,\alpha^{n-1}_m$ ($m\in\Z$) are the creation-annihilation operators corresponding to an $\Z$-basis $e_0,\ldots, e_{n-1}$ of $L$, satisfying the commutation relations $[\alpha^i_m,\alpha^{i'}_{m'}]=(e_i,e_{i'})m\delta_{m+m',0}$.
We also often use the creation-annihilation operators $\bbalpha^0_m,\ldots,\bbalpha^{n-1}_m$
with respect to the standard basis $\bbe_0,\ldots,\bbe_{n-1}$ of $\R^n$ where the lattice $L$ is embedded.
They are related to $\alpha^0_m,\ldots,\alpha^{n-1}_m$ under a proper $\R$-linear transformation,
and their commutation relation is of course $[\bbalpha^i_m,\bbalpha^{i'}_{m'}] = (\bbe_i,\bbe_{i'})m\delta_{m+m',0}$.

If we introduce the chiral bosons $X(z)=(X^i(z))_{i=0,\ldots,n-1}$ which satisfy the OPE
\begin{align}
X^i(z_1)\cdot X^j(z_2) \sim -\delta^{ij}\log(z_1-z_2),
\end{align}
and whose mode expansions are given as
\begin{align}
\partial X^i(z)=-\sqrt{-1}\sum_{m=-\infty}^\infty\frac{\bbalpha^i_m}{z^{m+1}},
\end{align}
then we can describe the state-operator correspondence as
(see e.g. \cite[\S2.8]{Polchinski:1998rq})
\begin{align}
\bbalpha^i_{-m}|0\rangle & \quad \longleftrightarrow \quad \frac{\sqrt{-1}}{(m-1)!}\partial^{m}X^i(z),\\
|k\rangle & \quad \longleftrightarrow \quad V_k(z) \propto \ :e^{\sqrt{-1}k\cdot X(z)}:,
\end{align}
where $\propto$ denotes that we ignore the cocycle factor,
which will be treated in Section \ref{subsec:cocycle-factor}.

\vspace{\vertspace}

If we consider an even self-dual lattice $\Lambda(C)$ constructed by Construction A from a doubly-even self-dual code $C$,
the resulting CFT has a shift $\Z_2$ symmetry with respect to the shift vector $\chi := \frac{1}{\sqrt{2}}\underline{1}$ as
\begin{align}
& X(z) \to X(z) + \pi\chi,\\
& V_k(z) \propto\ :e^{\sqrt{-1}k\cdot X(z)}:\ \to e^{\sqrt{-1}\pi k\cdot\chi} V_k(z). \label{eq:shift-Z2-of-vertex-op}
\end{align}
Then we can also consider the states under the twisted boundary condition by this shift $\Z_2$ symmetry,
in addition to the untwisted ones.
Such twisted states with respect to a $\Z_2$ symmetry
are built up on the twisted ground states,
which constitute the $2^{\frac{n}{2}}$-dimensional irreducible representation of a certain gamma matrix algebra,
and these twisted ground states are of weight $\frac{n}{16}$ and get a sign $(-1)^{\frac{n}{8}}$ under the $\Z_2$ symmetry \cite[\S5.3]{Dolan:1994st}.
As a result, the even and odd sectors of untwisted and twisted sectors under the action of the shift $\Z_2$ symmetry can be described as
\begin{align}
\begin{array}{c|cc}
 & \text{untwisted} & \text{twisted} \\\hline
\text{even} & \Lambda_0(C) & \Lambda_3(C) \\
\text{odd} & \Lambda_1(C) & \Lambda_2(C)
\end{array}
\quad ,
\label{sectors-under-shift-Z2}
\end{align}
where $\Lambda_i(C)$ denotes the sector consisting of the states in the form of (\ref{eq:general-element-of-lattice-CFT}) with $k\in\Lambda_i(C)$.

In the modern understanding of fermionization \cite{TachikawaLecture, Karch:2019lnn} (see also \cite{Hsieh:2020uwb} and references therein), 
we can uniformly treat $\Z_2$ orbifolding and fermionization.
By orbifolding the $\Z_2$ symmetry of (\ref{sectors-under-shift-Z2}), a new $\Z_2$ symmetry emerges, and we obtain the following orbifold theory:
\begin{align}
\begin{array}{c|cc}
\text{orbifold} & \text{untwisted} & \text{twisted} \\\hline
\text{even} & \Lambda_0(C) & \Lambda_1(C) \\
\text{odd} & \Lambda_3(C) & \Lambda_2(C)
\end{array}
\quad .
\end{align}
By fermionizing the $\Z_2$ symmetry of (\ref{sectors-under-shift-Z2}), on the other hand, we obtain a fermionic CFT.
There are two ways of fermionization:
\begin{align}
\begin{array}{c|cc}
\text{fermionic} & \text{NS} & \text{R} \\\hline
(-1)^F=+1 & \Lambda_0(C) & \Lambda_1(C) \\
(-1)^F=-1 & \Lambda_2(C) & \Lambda_3(C)
\end{array}
\quad , \qquad
\begin{array}{c|cc}
\text{fermionic}' & \text{NS} & \text{R} \\\hline
(-1)^F=+1 & \Lambda_0(C) & \Lambda_3(C) \\
(-1)^F=-1 & \Lambda_2(C) & \Lambda_1(C)
\end{array}
\quad . \label{eq:fermionization-of-lattice-CFT}
\end{align}
We refer the reader to \cite{Kawabata:2023iss, Kawabata:2024gek} for the shift $\Z_2$ symmetries of more general lattice CFTs, and their orbifold and fermionization.

For example, in the case where $C$ is the binary Golay code,
we can see from Table \ref{tbl:lettice-constructions} that the bosonic CFT constructed from the $(A_1)^{24}$ lattice is mapped to the bosonic CFT of the Leech lattice $\Lambda_{24}$ by the orbifold,
and mapped to the fermionic CFT whose NS sector is described by the odd Leech lattice $O_{24}$ by the fermionization,
with respect to the shift $\Z_2$ symmetry.

Lattice CFTs constructed from more general codes over finite fields are investigated in \cite{Gaiotto:2018ypj, Yahagi:2022idq, Kawabata:2023nlt, Kawabata:2023iss, Kawabata:2024gek}.
Recently, the construction of CFTs from quantum codes through Lorentzian lattices was established in \cite{Dymarsky:2020bps, Dymarsky:2020qom},
and has been developed in many directions;
see for example \cite{Dymarsky:2020pzc, Dymarsky:2021xfc, Buican:2021uyp, Angelinos:2022umf, Dymarsky:2022kwb}, \cite{Henriksson:2022dnu}, \cite{Furuta:2022ykh, Furuta:2023xwl}, \cite{Buican:2023ehi}, \cite{Kawabata:2022jxt, Alam:2023qac, Kawabata:2023usr, Kawabata:2023iss, Ando:2024gcf}.
We also note that \cite{Harvey:2020jvu} established the relation between a certain quantum code and a particular K3 CFT studied in \cite{Gaberdiel:2013psa}.

\section{Cocycle Factor and Automorphism of Lattice VOA}
\label{sec:cocycle-factor-and-aut-lattice-VOA}
In this Section \ref{sec:cocycle-factor-and-aut-lattice-VOA},
we review that the automorphism group $\Aut(V_L)$ of a lattice VOA inherits the information of the lattice isometry group $O(L)$
in the form of its group extension.
This is because of the cocycle factors appearing in the OPE of vertex operators,
which disturb the group structure of $O(L)$ directly lifting up to that of $\Aut(V_L)$.

We introduce the concept of cocycle factor in the language of physics in Section \ref{subsec:cocycle-factor},
and explain the outline of the rest of this section with a simplified argument in Section \ref{subsec:brief-description}.
They are refined in a mathematical language in the following Sections \ref{subsec:review-of-FLM} and \ref{subsec:paraphrase-question}.
Looking into the details of the group extension,
we paraphrase our main Question raised in Section \ref{sec:intro}
into more concrete forms Question \ref{q:lift-group-hom}, \ref{q:ses-split}, and \ref{q:tangible-form-Z2}, step by step,
and slightly generalize it to Question \ref{q:tangible-form}.

\subsection{Cocycle Factor}
\label{subsec:cocycle-factor}
To realize the appropriate commutation relations of the vertex operators $V_k(z)\sim\ :e^{\sqrt{-1}k\cdot X(z)}:\,$, in accordance with whether $V_k(z)$ is bosonic or fermionic,
we have to introduce a correction factor $c_k(p)$ to modify the commutation relations of $:e^{\sqrt{-1}k\cdot X(z)}:$'s.
As a result, an additional factor $\varepsilon(k,k')$ called a \emph{cocycle factor} appears in the OPEs of the vertex operators.
In this Section \ref{subsec:cocycle-factor}, we will review this story in the language of physics.
A good reference for cocycle factors is \cite[\S\S6.4.4-6.4.5]{Green:2012oqa}, but it deals with only even lattices.
The cases including odd lattices are discussed in for example \cite[Appendix]{Goddard-Olive}, \cite[Appendix A]{Gaberdiel:2013psa}.

Let $L$ be an integral lattice of rank $n$.
We would like to construct a VOA satisfying
\begin{align}
V_k(z_1) \cdot V_{k'}(z_2) \sim (-1)^{|k|^2|k'|^2} V_{k'}(z_2) \cdot V_{k}(z_1) \sim \varepsilon(k,k') (z_1-z_2)^{k\cdot k'} V_{k+k'}(z_2),
\label{eq:V_k-OPE}
\end{align}
for some $\varepsilon:L\times L\to\{\pm1\}$,
where $k,k'\in L$ and we dropped $O((z_1-z_2)^{k\cdot k'+1})$ terms.
Since $V_k(z)$ is a vertex operator of conformal dimension $\frac{1}{2}|k|^2$,
it is bosonic when $k$ is even and fermionic when $k$ is odd,
which accounts for the factor $(-1)^{|k|^2|k'|^2}$ of the commutation relation.
It immediately follows from (\ref{eq:V_k-OPE}) that $\varepsilon$ must satisfy
\begin{align}
\varepsilon(k,k') = (-1)^{k\cdot k'+|k|^2|k'|^2}\varepsilon(k',k).
\label{eq:2-cocycle-symmetricity}
\end{align}

Recalling that the OPE of operators $:e^{\sqrt{-1}k\cdot X(z)}:$ is
\begin{align}
:e^{\sqrt{-1}k\cdot X(z_1)}: \cdot :e^{\sqrt{-1}k'\cdot X(z_2)}: & = (z_1-z_2)^{k\cdot k'} :e^{\sqrt{-1}k\cdot X(z_1)}e^{\sqrt{-1}k'\cdot X(z_2)}:,
\end{align}
we can observe
\begin{align}
\hspace{-3pt}\scalebox{0.8}{
$:e^{\sqrt{-1}k\cdot X(z_1)}: \cdot :e^{\sqrt{-1}k'\cdot X(z_2)}: \ = (-1)^{k\cdot k'} :e^{\sqrt{-1}k'\cdot X(z_2)}: \cdot :e^{\sqrt{-1}k\cdot X(z_1)}: \ \sim (z_1-z_2)^{k\cdot k'} :e^{\sqrt{-1}(k+k')\cdot X(z_2)}:$
},
\label{eq:e^ik-OPE}
\end{align}
where $O((z_1-z_2)^{k\cdot k'+1})$ terms are dropped.
This (\ref{eq:e^ik-OPE}) differs from the desired commutation relation (\ref{eq:V_k-OPE}) only in the sign.
To modify it, let us assume $V_k(z)$ to be in the form of
\begin{align}
V_k(z) = \ :e^{\sqrt{-1}k\cdot X(z)}:c_k(p),
\end{align}
where $c_k(p)$ is an operator in the form of a function of the momentum operators $p^i = \alpha^i_0$ ($i=0,\ldots, n-1$).
Then (\ref{eq:V_k-OPE}) translates to the condition on this correction factor $c_k(p)$ as
\begin{align}
c_k(p+k')c_{k'}(p) = (-1)^{k\cdot k'+|k|^2|k'|^2} c_{k'}(p+k)c_k(p) = \varepsilon(k,k')c_{k+k'}(p),
\label{eq:c_k(p)-OPE}
\end{align}
where we used $c_k(p) \, \cdot :e^{\sqrt{-1}k'\cdot X(z)}: \, = \, :e^{\sqrt{-1}k'\cdot X(z)}: \cdot \, c_k(p+k')$, because $[p^i, :e^{\sqrt{-1}k'\cdot X(z)}:] =$ $(k')^i:e^{\sqrt{-1}k'\cdot X(z)}:$.

In addition, if we impose the associativity on $c_k(p)$ as
\begin{align}
(c_k(p+k'+k'')c_{k'}(p+k''))c_{k''}(p) = c_k(p+k'+k'')(c_{k'}(p+k'')c_{k''}(p)),
\label{eq:associativity-of-c_k(p)}
\end{align}
then the factor $\varepsilon$ should satisfy the condition
\begin{align}
\varepsilon(k,k')\varepsilon(k+k',k'') = \varepsilon(k,k'+k'')\varepsilon(k',k'').
\label{eq:2-cocycle-condition}
\end{align}
This means that the factor $\varepsilon:L\times L\to\{\pm1\}$ is a 2-cocycle in the language of group cohomology
(where $\{\pm1\}$ is regarded as an $L$-module by the trivial $L$-action),
and hence $\varepsilon$ is called a \emph{cocycle factor}.
A 2-cocycle satisfying (\ref{eq:2-cocycle-symmetricity}) is unique up to coboundary;
see Theorem \ref{thm:cent-ext-and-comm-map}.
In the language of Section \ref{subsec:cent-ext-and-comm-map},
$\varepsilon$ is a 2-cocycle for the commutator map $c(k,k')=(-1)^{k\cdot k'+|k|^2|k'|^2}$.

\vspace{\vertspace}

To construct $c_k(p)$ satisfying (\ref{eq:c_k(p)-OPE})  and (\ref{eq:associativity-of-c_k(p)}),
we choose a basis $\{e_i\}_{i=0,\ldots,n-1}$ of $L$
and introduce a bilinear non-commutative product $\ast:L\times L\to\Z$ on $L$ by
\begin{align}
k \ast k' := \sum_{i>j} k^i k'^j (e_i\cdot e_j + |e_i|^2 |e_j|^2),
\end{align}
for $k=\sum_i k^i e_i$ and $k'=\sum_j k'^j e_j$.
In the language of Section \ref{subsec:cent-ext-and-comm-map},
this definition follows the construction (\ref{eq:cocycle-from-comm-map}) of a cocycle from the commutator map.
This product $\ast$ depends on the choice of the basis $\{e_i\}_i$.
If the lattice $L$ is even,
or if the lattice $L$ is odd and we choose the basis so that $e_0$ is odd and $e_1,\ldots,e_{n-1}$ are even,\footnote{
This is always possible because (odd vector) $-$ (odd vector) $=$ (even vector).
}
then this product reduces modulo 2 to
\begin{align}
k \ast k' \equiv \sum_{i>j} k^i k'^j e_i\cdot e_j \pmod{2}.
\label{eq:ast-prod-reduced}
\end{align}

We can now construct $c_k(p)$ satisfying (\ref{eq:c_k(p)-OPE}) and (\ref{eq:associativity-of-c_k(p)}) as
\begin{align}
c_k(p) = (-1)^{k \ast p}, \label{eq:c_k(p)}
\end{align}
and then $\varepsilon$ is given by
\begin{align}
& \varepsilon(k,k') = (-1)^{k \ast k'},
\label{eq:varepsilon}
\end{align}
which of course satisfies (\ref{eq:2-cocycle-symmetricity}) and (\ref{eq:2-cocycle-condition}).

\begin{proof}[Proof that $c_k(p)$ in (\ref{eq:c_k(p)}) satisfy (\ref{eq:c_k(p)-OPE}) and (\ref{eq:associativity-of-c_k(p)})]
For the first equation of (\ref{eq:c_k(p)-OPE}), it suffices to check
\begin{align}
(-1)^{k\ast k'}=(-1)^{k\cdot k'+|k|^2|k'|^2}(-1)^{k'\ast k}.
\end{align}
This is already discussed around (\ref{eq:cocycle-from-comm-map}),
but we can check it explicitly as follows:
\begin{align}
(-1)^{k\ast k'+k'\ast k} & = (-1)^{\sum_{i\neq j}(k^ik'^je_i\cdot e_j+|k^ie_i|^2|k'^je_j|^2)} \\
 & = (-1)^{\sum_{i,j}(k^ik'^je_i\cdot e_j+|k^ie_i|^2|k'^je_j|^2)} = (-1)^{k\cdot k' + |k|^2|k'|^2},
\end{align}
where the first equation follows from
$k^i \equiv (k^i)^2 \mod 2$,
the second equation follows from
$k^ik'^i|e_i|^2 + (k^ik'^i|e_i|^2)^2 \equiv 0 \mod 2$,
and the last equation follows from
$\sum_i|k^ie_i|^2 \equiv |k|^2 \mod 2$.

The second equation of (\ref{eq:c_k(p)-OPE}), and (\ref{eq:associativity-of-c_k(p)}), are obvious.
\end{proof}

\subsection{Brief Description of This Section}
\label{subsec:brief-description}
In this Section \ref{subsec:brief-description}, we will outline what we are doing in the following Sections \ref{subsec:review-of-FLM}-\ref{subsec:paraphrase-question} with a simplified argument.
The discussions here will be refined in a mathematical language in the following Sections \ref{subsec:review-of-FLM} and \ref{subsec:paraphrase-question}.

Let us assume that an isometry $g$ of the lattice $L$ can be lifted to an automorphism of the lattice VOA $V_L$, which consists of $V_k(z)$'s in (\ref{eq:V_k-OPE}), as
\begin{align}
g(V_k(z)) = \zeta_g(k) V_{g(k)}(z),
\end{align}
where $\zeta_g(k)$ is some factor taking values in $\Z_2=\{\pm1\}$, or more generally, in $U(1)$.
Here, we abused notation and also wrote the automorphism of $V_L$ lifted from the isometry $g$ of $L$ as $g$.
Since $g$ is an automorphism, it must preserve the multiplication (\ref{eq:V_k-OPE}), and hence $\zeta_g$ must satisfy
\begin{align}
\varepsilon(k,k') \zeta_g(k+k') = \zeta_g(k) \zeta_g(k') \varepsilon(g(k),g(k')).
\label{eq:zeta-vs-varepsilon-easy}
\end{align}

If we specify the values of $\zeta_g(e_0),\ldots,\zeta_g(e_{n-1})$ for a basis $\{e_i\}_{i=0,\ldots,n-1}$ of $L$,
then (\ref{eq:zeta-vs-varepsilon-easy}) recursively determines all the values $\zeta_g(k)$ for $k\in L$.
The resulting $\zeta_g$ is well-defined, because
\begin{align}
\zeta_g((k+k')+k'') & = \zeta_g(k+k')\zeta_g(k'')\frac{\varepsilon(g(k+k'),g(k''))}{\varepsilon(k+k',k'')} \label{eq:recursive-well-defined-1}\\
& = \zeta_g(k)\zeta_g(k')\zeta_g(k'')\frac{\varepsilon(g(k),g(k'))}{\varepsilon(k,k')}\frac{\varepsilon(g(k+k'),g(k''))}{\varepsilon(k+k',k'')}
\end{align}
and
\begin{align}
\zeta_g(k+(k'+k'')) & = \zeta_g(k)\zeta_g(k'+k'')\frac{\varepsilon(g(k),g(k'+k''))}{\varepsilon(k,k'+k'')}\\
& = \zeta_g(k)\zeta_g(k')\zeta_g(k'')\frac{\varepsilon(g(k),g(k'+k''))}{\varepsilon(k,k'+k'')}\frac{\varepsilon(g(k'),g(k''))}{\varepsilon(k',k'')}
\label{eq:recursive-well-defined-4}
\end{align}
are equal since $\varepsilon$ is a 2-cocycle (\ref{eq:2-cocycle-condition}).
Conversely, any $\zeta_g$ determined in such a way by (\ref{eq:zeta-vs-varepsilon-easy}) can define an automorphism of $V_L$.
Therefore, there are $\Z_2^n$, or more generally $U(1)^n$, degrees of freedom in the way of lifting a lattice isometry $g$ to an automorphism of the lattice VOA.

\begin{rem*}
Here is another derivation of $\Z_2^n$.
If both $\zeta_g$ and $\zeta'_g$ satisfy (\ref{eq:zeta-vs-varepsilon-easy}),
then their difference $\eta(k)=\zeta'_g(k)\zeta_g(k)^{-1}$ turns out to be just a linear map in $\Hom(L,\Z_2)$:
\begin{align}
\eta(k+k') = \eta(k)\eta(k').
\end{align}
Therefore, the degree of freedom in choosing $\zeta_g$ is $\Hom(L,\Z_2)\cong\Z_2^n$.
\textit{(Remark ends.)}
\end{rem*}

Our main Question raised in Section \ref{sec:intro} was whether we can lift each isometry of $L$ to an automorphism of $V_L$, while preserving the whole group structure of the isometry group $O(L)$.
If this is the case, we should have
\begin{align}
g_2(g_1(V_k(z)) = (g_2g_1)(V_k(z)) \quad \text{for any $g_1,g_2\in O(L)$},
\end{align}
and hence
\begin{align}
\zeta_{g_2}(g_1(k))\zeta_{g_1}(k) = \zeta_{g_2g_1}(k) \quad \text{for any $g_1,g_2\in O(L)$}.
\label{eq:grp-hom-condition-on-zeta-easy}
\end{align}

Therefore, our Question can be restated as follows.
``Can we find a collection $\{\zeta_g:L\to\Z_2 \text{\ (or $U(1)$)} \mid \text{$\zeta_g$ satisfies (\ref{eq:zeta-vs-varepsilon-easy})}\}_{g\in O(L)}$ satisfying the condition (\ref{eq:grp-hom-condition-on-zeta-easy})?
Here, we can adjust the values of $\zeta_g(e_0),\ldots,\zeta_g(e_{n-1})\in\Z_2\text{\ (or $U(1)$)}$ for each $g\in O(L)$, from which $\zeta_g$ is completely determined by (\ref{eq:zeta-vs-varepsilon-easy}).''
This is the paraphrased Question \ref{q:tangible-form-Z2} or \ref{q:tangible-form}
in the following Sections \ref{subsec:review-of-FLM} and \ref{subsec:paraphrase-question}.

\subsection{Lift of Lattice Isometries to Automorphisms of Lattice VOA: Part 1}
\label{subsec:review-of-FLM}
An isometry of the lattice $L$ can be lifted to an automorphism of the lattice VOA $V_L$ \cite{MR0996026};
our Question was whether such lifts preserve the group structure or not.
That is,

\begin{question}
\label{q:lift-group-hom}
Can we lift the isometry group $O(L)$ of a lattice $L$ to a subgroup of the automorphism group $\Aut(V_L)$ of the lattice VOA $V_L$?
\end{question}

To begin with, we review some general theory on the lift from \cite{MR0996026} in this Section \ref{subsec:review-of-FLM}.
See Appendix \ref{sec:group-extension} for some basic facts on group extensions, although this section is intended to be self-contained.
We also note that we will postpone reviewing the definition of the lattice VOA $V_L$ itself until a later Section \ref{subsec:autom-of-lattice-VOA};
for a while, it suffices to consider it to be some algebra consisting of $V_k(z)$'s in (\ref{eq:V_k-OPE}).

\vspace{\vertspace}

Let $\hat{L}$ be the central extension\footnote{
\label{fn:more-general-extension}
More generally, \cite{MR0996026} deals with the central extension $\hat{L}$ of $L$ by $\Zgen{s}{\kappa}=\langle\kappa\mid\kappa^s=1\rangle$.
Furthermore, we would like to consider the central extension by $U(1) \cong \R/2\Z = \langle\kappa\rangle_\R / \langle \kappa \rangle_{2\Z}$ later.
Therefore, the discussions below avoid using special properties for $\Z_2$, such as $\hat{\varepsilon}(k,k')^{-1}=\hat{\varepsilon}(k,k')$.
} 
of the free abelian group $L$ by $\Zgen{2}{\kappa}=\langle\kappa\mid\kappa^2=1\rangle$
\begin{align}
1 \to \Zgen{2}{\kappa} \to \hat{L} \overset{\bar{}}{\to} L \to 0,
\label{eq:lift-of-L}
\end{align}
specified by a 2-cocycle $\hat{\varepsilon}:L\times L\to\Zgen{2}{\kappa}$.
In other words, $\hat{L}$ is $\Zgen{2}{\kappa}\times L$ as a set, and if we write its element as $\kappa^me^k$ $(\kappa^m\in\Zgen{2}{\kappa}, k\in L)$, then $\hat{L}$ is a group specified by the multiplication
\begin{align}
\kappa^m e^k \cdot \kappa^{m'} e^{k'} = \hat{\varepsilon}(k,k') \kappa^{m+m'} e^{k+k'}.
\end{align}
So, $V_k(z)$ and $\varepsilon$ in (\ref{eq:V_k-OPE}) correspond to $e^k\in\hat{L}$ and $\hat{\varepsilon}$ with $\kappa=e^{\sqrt{-1}\pi}$ here, respectively.

Here, it follows from (\ref{eq:2-cocycle-condition}) that any 2-cocycle $\hat{\varepsilon}:L\times L\to\Zgen{2}{\kappa}$ satisfies
\begin{align}
\hat{\varepsilon}(k,0) = \hat{\varepsilon}(0,k) = \hat{\varepsilon}(0,0) \quad \text{for any $k\in L$}.
\label{eq:2-cocycle-for-0}
\end{align}
Furthermore, it is known that there exists a 2-cocycle $\hat{\varepsilon}$ satisfying the normalization condition
\begin{align}
\hat{\varepsilon}(0,0) = \kappa^0,
\label{eq:normalized-cocycle-lattice}
\end{align}
in any cohomology class in $H^2(L,\Zgen{2}{\kappa})$,
and it defines an equivalent extension $\hat{L}$ to any 2-cocycle in the same cohomology class (see the last paragraph of Section \ref{subsec:grp-ext-and-grp-coh}).
Therefore, we will always assume that $\hat{\varepsilon}$ is a normalized one as in (\ref{eq:normalized-cocycle-lattice}), without loss of generality.
Also note that there is no problem in applying the general theory here to the specific cocycle $\varepsilon$ in (\ref{eq:varepsilon}), because it also satisfies the normalization $\varepsilon(0,0)=1$.
Then, we can observe that the multiplication of $\hat{L}$ is well-behaved in the sense that
\begin{align}
& \kappa^me^0 \cdot \kappa^{m'}e^k = \kappa^{m+m'}e^k,\\
& \kappa^me^k \cdot \kappa^{m'}e^0 = \kappa^{m+m'}e^k,
\end{align}
and hence there is no confusion if we just write $\kappa^m$ instead of $\kappa^me^0$.
We will also just write $e^k$ instead of $\kappa^0e^k$.

We define the \emph{commutator map} $\hat{c}:L\times L\to\Zgen{2}{\kappa}$ by
\begin{align}
\hat{c}(k,k') := \hat{\varepsilon}(k,k') \hat{\varepsilon}(k',k)^{-1}.
\label{eq:commutator-vs-2-cocycle}
\end{align}
If $\hat{\varepsilon}$ here is the specific one $\varepsilon$ in (\ref{eq:V_k-OPE}), we have $\hat{c}(k,k')=\kappa^{k\cdot k'+|k|^2|k'|^2}\mid_{\kappa=-1}$ from (\ref{eq:2-cocycle-symmetricity}).

\vspace{\vertspace}

Forgetting the symmetric bilinear form on $L$, for a moment, we focus on the automorphism group $\Aut(L)$ of just a free abelian group $L$, instead of the isometry group $O(L)$ of the lattice $L$.
The proposition \cite[Prop.\ 5.4.1]{MR0996026} states that
\begin{align}
1 \to \Hom(L,\Zgen{2}{\kappa}) \overset{\widetilde{}}{\to} \Aut(\hat{L},\kappa) \overset{\bar{}}{\to} \Aut(L,\hat{c}) \to 1
\label{eq:lift-of-Aut(L,c)}
\end{align}
is exact.
(The proof is also reviewed in Appendix \ref{subsec:thm-on-aut-grp-of-central-ext}.)
The details of (\ref{eq:lift-of-Aut(L,c)}) are as follows.
$\Hom(L,\Zgen{2}{\kappa}) \tilde{\to} \Aut(\hat{L},\kappa)$ maps
$\eta\in\Hom(L,\Zgen{2}{\kappa})$ to
\begin{align}
\tilde{\eta}: \hspace{20pt} \hat{L} & \to \hat{L} \label{eq:def-of-tilde-pre}\\
\kappa^me^k & \mapsto \eta(k)\kappa^{m}e^k. \label{eq:def-of-tilde}
\end{align}
Note that $\eta\in\Hom(L,\Z_2)$ is determined only from the values of $\eta(e_0), \ldots, \eta(e_{n-1})$,
where $e_0,\ldots,e_{n-1}$ is a basis of $L$,
and hence
\begin{align}
\Hom(L,\Z_2) \cong (\Z_2)^n.
\end{align}
$\Aut(\hat{L},\kappa)$ is defined as
\begin{align}
\Aut(\hat{L},\kappa) := \{f\in\Aut(\hat{L}) \mid f(\kappa) = \kappa\},
\end{align}
and $\Aut(L,\hat{c})$ is defined as
\begin{align}
\Aut(L,\hat{c}) := \{g \in \Aut(L) \mid \hat{c}(g(k), g(k')) = \hat{c}(k, k')\}.
\label{eq:def-of-Aut(L,c)}
\end{align}
$\Aut(\hat{L},\kappa) \bar{\to} \Aut(L,\hat{c})$ maps $f\in\Aut(\hat{L})$ to
\begin{align}
\bar{f} : L & \to L\\
k & \mapsto \overline{f(e^k)}, \label{eq:def-of-barf}
\end{align}
where the natural projection $\hat{L}\bar{\to} L$ of (\ref{eq:lift-of-L}) is used.

Since $1$ and $\kappa$ are the only elements of finite order in $\hat{L}$,
it follows that any $f\in\Aut(\hat{L})$ satisfies $f(\kappa)=\kappa$,
and hence $\Aut(\hat{L},\kappa)$ reduces to $\Aut(\hat{L})$ in the case at hand.
However, this is the special property for the extension by $\Z_{s=2}$.
If we consider a more general extension, say by $\Z_{s>2}$,
then we cannot reduce $\Aut(\hat{L},\kappa)$ to $\Aut(\hat{L})$
(see footnote \ref{fn:more-general-extension}).

\vspace{\vertspace}

Let us recall that $L$ is a lattice, more than just a free abelian group, and move on to the isometry group $O(L)$ of the lattice from the automorphism group $\Aut(L)$ of the free abelian group.
If the commutator map $\hat{c}$ depends only on the bilinear form of the lattice $L$, say $\hat{c}(k,k')=\kappa^{k\cdot k'+|k|^2|k'|^2}$,
then the isometry group $O(L)$ is a subgroup of $\Aut(L,\hat{c})$ defined in (\ref{eq:def-of-Aut(L,c)}).
In addition, if we define
\begin{align}
O(\hat{L}) := \{f\in\Aut(\hat{L},\kappa) \mid \bar{f}\in O(L)\},
\label{eq:def-of-O(Lhat)}
\end{align}
then we obtain the exact sequence \cite[Prop.\ 6.4.1]{MR0996026}
\begin{align}
1 \to \Hom(L,\Zgen{2}{\kappa}) \overset{\tilde{}}\to O(\hat{L}) \overset{\bar{}}{\to} O(L) \to 1
\label{eq:lift-of-O(L)}
\end{align}
from (\ref{eq:lift-of-Aut(L,c)}).

As we will review in Section \ref{subsec:autom-of-lattice-VOA},
any element of $O(\hat{L})$ can be naturally extended to an automorphism of the lattice VOA $V_L$,
and then $O(\hat{L})$ can be regarded as a subgroup of $\Aut(V_L)$.
The whole $\Aut(V_L)$ is determined in \cite{MR1745258}.

Now, Question \ref{q:lift-group-hom} is reduced to how $O(L)$ is lifted into $O(\hat{L})$ in (\ref{eq:lift-of-O(L)}).
More precisely,

\begin{question}
\label{q:ses-split}
Does the exact sequence (\ref{eq:lift-of-O(L)}) have a section $O(L)\to O(\hat{L})$ such that it is a group homomorphism?
\end{question}
\noindent If the answer is yes, we say that the exact sequence (\ref{eq:lift-of-O(L)}) splits,
and $O(\hat{L})$ is a semidirect product of $\Hom(L,\Zgen{2}{\kappa})$ and $O(L)$.
Otherwise, $O(\hat{L})$ is some non-split extension of $O(L)$,
and $O(L)$ does not lift to a subgroup of $O(\hat{L})$ or $\Aut(V_L)$.

\vspace{\vertspace}

\begin{rem*}
For any $g\in O(L)$, it is known that there exists a lift $\hat{g}\in O(\hat{L})$ of $g$ such that $\hat{g}(e^k)=e^k$ for any $k\in L$ fixed by $g$ as $g(k)=k$ \cite[\S5]{MR820716}.
Such lift $\hat{g}$ is called the \emph{standard lift} of $g$,
and is sometimes of use in research.
See for example \cite[Lemma 12.1]{MR1172696}, \cite[\S5.3]{Moller:2016wzp}, \cite[\S7]{MR4058176} for some properties of the standard lift.
\textit{(Remark ends.)}
\end{rem*}

\subsection{Making the Question More Concrete}
\label{subsec:paraphrase-question}
We will paraphrase Question \ref{q:ses-split} into a more concrete form.

Let $S:O(L)\to O(\hat{L})\ ;g\mapsto S_g$ be a section of (\ref{eq:lift-of-O(L)}).
Then for each $g\in O(L)$, we can define $\zeta_g:L\to \Zgen{2}{\kappa}$ by
\begin{align}
S_g(e^k) = \zeta_g(k)e^{g(k)},
\label{eq:def-of-zeta}
\end{align}
because $\overline{S_g(e^k)}=g(k)$.
Since $S_g$ is an automorphism of $\hat{L}$, it must satisfy
\begin{align}
S_g(e^k \cdot e^{k'}) = S_g(e^k) \cdot S_g(e^{k'}),
\end{align}
which translates to the condition on $\zeta_g$ as
\begin{align}
\hat{\varepsilon}(k,k')\zeta_g(k+k') = \zeta_g(k)\zeta_g(k')\hat{\varepsilon}(g(k),g(k')),
\label{eq:zeta-vs-varepsilon}
\end{align}
where we also used $S_g\in O(\hat{L})$ satisfies $S_g(\kappa)=\kappa$.

\begin{rem*}
The equation (\ref{eq:zeta-vs-varepsilon}) is equivalent to $\hat{\varepsilon}(k,k')=\hat{\varepsilon}(g(k),g(k'))d\zeta_g(k,k')$,
which means that the cocycles $\hat{\varepsilon}({-},{-})$ and $\hat{\varepsilon}(g({-}),g({-}))$ are in the same cohomology class.
\textit{(Remark ends.)}
\end{rem*}

Take another section $S':O(L)\to O(\hat{L})$ of (\ref{eq:lift-of-O(L)})
and define $\zeta'_g:L\to\Zgen{2}{\kappa}$ in the same way.
Since $S'_g(S_g)^{-1}$ is in the kernel of $O(\hat{L})\bar{\to}O(L)$ in (\ref{eq:lift-of-O(L)}),
there exists $\eta_g\in\Hom(L,\Zgen{2}{\kappa})$ satisfying
\begin{align}
\tilde{\eta}_g = S'_g(S_g)^{-1},
\label{eq:DoF-in-S_g}
\end{align}
which shows that the degree of freedom in taking the value $S_g$ of the section $S:O(L)\to O(\hat{L})$ is the linear map $\eta_g\in\Hom(L,\Zgen{2}{\kappa})$ for each $g\in O(L)$.
In terms of the factors $\zeta_g$ and $\zeta_g':L\to\Z_2$,
(\ref{eq:DoF-in-S_g}) reduces to
\begin{align}
\eta_g = \zeta'_g(\zeta_g)^{-1},
\end{align}
where the degree of freedom in taking the value $S_g$ is just parametrized by $\zeta_g$.

We can rephrase this degree of freedom $\Hom(L,\Zgen{2}{\kappa}) \cong (\Z_2)^n$ in terms of the factor $\zeta_g:L\to\Z_2$ in the following way.
If we define the values of $\zeta_g(e_0), \ldots, \zeta_g(e_{n-1})\in\Zgen{2}{\kappa}$, where $e_0,\ldots,e_{n-1}$ is a basis of $L$,
then the condition (\ref{eq:zeta-vs-varepsilon}) recursively determines all the values of $\zeta_g(k)$ for $k\in L$.
This recursive way to determine $\zeta_g$ is well-defined thanks to the fact that $\hat{\varepsilon}$ is a 2-cocycle, as we saw in (\ref{eq:recursive-well-defined-1})--(\ref{eq:recursive-well-defined-4}).
Therefore, the value $S_g\in O(\hat{L})$ is parametrized by only $\zeta_g(e_0), \ldots, \zeta_g(e_{n-1})$, and hence the parameter space is $(\Z_2)^n$.

\vspace{\vertspace}

If there exists a section $S:O(L)\to O(\hat{L})$ such that it is a group homomorphism, as in Question \ref{q:ses-split},
then it should satisfy
\begin{align}
S_{g_2} \circ S_{g_1} = S_{g_2g_1} \quad \text{for any $g_1,g_2\in O(L)$}.
\end{align}
This translates to the condition on $\zeta_g$ as
\begin{align}
\zeta_{g_2}(g_1(k))\zeta_{g_1}(k) = \zeta_{g_2g_1}(k) \quad \text{for any $g_1,g_2\in O(L)$}.
\label{eq:grp-hom-condition-on-zeta}
\end{align}
Therefore, Question \ref{q:ses-split} can be paraphrased as

\begin{question}
\label{q:tangible-form-Z2}
Can we find a collection $\{\zeta_g:L\to\Zgen{2}{\kappa} \mid \text{$\zeta_g$ satisfies (\ref{eq:zeta-vs-varepsilon})}\}_{g\in O(L)}$ satisfying the condition (\ref{eq:grp-hom-condition-on-zeta}),
by adjusting each $\zeta_g$ within the degree of freedom $\Hom(L,\Zgen{2}{\kappa})$?
More specifically, we are allowed to adjust the values of $\zeta_g(e_0),\ldots,\zeta_g(e_{n-1})\in\Zgen{2}{\kappa}$ for each $g\in O(L)$, from which $\zeta_g$ is completely determined by (\ref{eq:zeta-vs-varepsilon}).
\end{question}

\begin{rem*}
Note that even if we take another cocycle $\hat{\varepsilon}'$ cohomologous to $\hat{\varepsilon}$, in the equation (\ref{eq:zeta-vs-varepsilon}) from which $\zeta_g$ is determined,
the answer to Question \ref{q:tangible-form-Z2} does not change.
This is just because cocycles in the same cohomology class define equivalent extensions $\hat{L}$,
but we can also check it explicitly as follows.
If $\zeta_g$ satisfies (\ref{eq:zeta-vs-varepsilon}),
then $\zeta'_g$ satisfying (\ref{eq:zeta-vs-varepsilon}) for $\hat{\varepsilon}'= \hat{\varepsilon}d\delta$,
where $d\delta(k,k')=\delta(k')\delta(k+k')^{-1}\delta(k)$ is a coboundary,
can be constructed as
\begin{align}
\zeta'_g(k) = \zeta_g(k)\delta(k)\delta(g(k))^{-1}.
\label{eq:zeta-for-cohomologous-cocycle}
\end{align}
In addition, if $\{\zeta_g \mid (\ref{eq:zeta-vs-varepsilon})\}_g$ satisfies (\ref{eq:grp-hom-condition-on-zeta}), then it is easy to check that $\{\zeta'_g\}_g$ constructed by (\ref{eq:zeta-for-cohomologous-cocycle}) also satisfies (\ref{eq:grp-hom-condition-on-zeta}).
\textit{(Remark ends.)}
\end{rem*}

\vspace{\vertspace}

Finally, we mention one slight generalization of the question.
We can generalize the discussions so far
from $\Zgen{2}{\kappa} = \langle \kappa \mid \kappa^2 = 1 \rangle$
to $U(1) \cong \R/2\Z = \{\kappa^r \mid r\in\R \} / \langle \kappa^2 \rangle_\Z$ (see footnote \ref{fn:more-general-extension}).
As a result, we may consider another question

\begin{questionprime}{q:tangible-form-Z2}
\label{q:tangible-form}
Can we find a collection $\{\zeta_g:L\to U(1) \mid \text{$\zeta_g$ satisfies (\ref{eq:zeta-vs-varepsilon})}\}_{g\in O(L)}$ satisfying the condition (\ref{eq:grp-hom-condition-on-zeta}),
by adjusting each $\zeta_g$ within the degree of freedom $\Hom(L,U(1))$?
More specifically, We are allowed to adjust the values of $\zeta_g(e_0),\ldots,\zeta_g(e_{n-1})\in U(1)$ for each $g\in O(L)$, from which $\zeta_g$ is completely determined by (\ref{eq:zeta-vs-varepsilon}).
\end{questionprime}

Of course, if the answer to Question \ref{q:tangible-form-Z2} is yes,
then so is the answer to Question \ref{q:tangible-form},
and if the answer to Question \ref{q:tangible-form} is no,
then so is that to Question \ref{q:tangible-form-Z2}.

\section{Determine Whether the Lift Preserves Group Structures or Not}
\label{sec:answer-to-the-question}
In this Section \ref{sec:answer-to-the-question}, we negatively settle Question \ref{q:tangible-form} in the cases of the odd Leech lattice and the Leech lattice.

\subsection{Closed Form of the Factor \texorpdfstring{$\zeta_g$}{zeta g}}
\label{subsec:closed-form-of-zeta}
Although we have pointed out that the relation (\ref{eq:zeta-vs-varepsilon}) determines the value of $\zeta_g(k)$ for any $k\in L$ recursively from the values of $\zeta_g(e_0),\ldots,\zeta_g(e_{n-1})$,
it would be convenient if we have a closed form of $\zeta_g(k)$.
It can be done at least in the case where the cocycle $\hat{\varepsilon}$ is the specific one $\varepsilon$ in (\ref{eq:varepsilon}), and the lattice $L$ satisfies the following assumption.

\begin{prop}
\label{prop:closed-form-of-zeta}
Let $L$ be an integral lattice of rank $n$.
Assume that we can take an integral basis $e_0,\ldots,e_{n-1}$ of $L$ such that,\footnote{
The existence of a basis $e_0,\ldots,e_{n-1}$ of a lattice $L$ satisfying $|e_i|^2\in 4\Z$ for any $i=0,\ldots,n-1$ does not imply that $|k|^2\in 4\Z$ for any $k\in L$.
In fact, we can take such a basis of the Leech lattice,
as we can see from the symmetric bilinear form (\ref{eq:sym-bilin-form-of-Leech}),
but the Leech lattice has vectors whose squared lengths are not multiples of 4.
}
if $L$ is even, $|e_i|^2\in 4\Z$ for any $i=0,\ldots,n-1$,
and if $L$ is odd, $e_0$ is odd and $|e_i|^2\in 4\Z$ for any $i=1,\ldots,n-1$.
If the 2-cocycle $\hat{\varepsilon}:L\times L\to U(1)$ in (\ref{eq:zeta-vs-varepsilon}) (with $\kappa=e^{\sqrt{-1}\pi}$) is the one $\varepsilon$ in (\ref{eq:varepsilon})
\begin{align}
\varepsilon(k,k') = (-1)^{k \ast k'}
\label{eq:varepsilon-again}
\end{align}
where the product $\ast$ is specified by the above choice of a basis,
then the equation (\ref{eq:zeta-vs-varepsilon}) determines the value of $\zeta_g(k)$ for $g\in O(L)$ and $k=\sum_ik^ie_i\in L$ as
\begin{align}
\zeta_g(k) = \prod_{i=0}^{n-1}\zeta_g(e_i)^{k^i} \cdot \prod_{\substack{ i'<j' \\ i',j' \in \{i \mid \text{$k^i$ is odd}\} }}\varepsilon(g(e_{i'}),g(e_{j'})).
\label{eq:closed-form-of-zeta}
\end{align}
\end{prop}

Before getting into the proof of this Proposition \ref{prop:closed-form-of-zeta}, we first recall that the bilinear non-commutative product $\ast$ used in (\ref{eq:varepsilon-again}) can be reduced modulo 2 to (\ref{eq:ast-prod-reduced})
\begin{align}
k \ast k' \equiv \sum_{i>j} k^i k'^j e_i\cdot e_j \pmod{2}
\label{eq:ast-prod-reduced-again}
\end{align}
under the choice of the basis $e_0,\ldots,e_{n-1}$ in Proposition \ref{prop:closed-form-of-zeta}.

We also prepare the following lemma.

\begin{lemma}
\label{lemma:property-of-varep-and-zeta}
If the 2-cocycle $\hat{\varepsilon}:L\times L\to U(1)$ in (\ref{eq:zeta-vs-varepsilon}) is the one $\varepsilon$ in (\ref{eq:varepsilon-again}) where the product $\ast$ is specified by the choice of the basis $e_0,\ldots,e_{n-1}$ in Proposition \ref{prop:closed-form-of-zeta},
then we have
\begin{enumerate}
\item $\varepsilon(2k,k') = 1$ for any $k,k'\in L$.
\label{lemma:property-of-varep-and-zeta-1}
\item $\zeta_g(2k+k') = \zeta_g(2k)\zeta_g(k')$ for any $k,k'\in L$ and $g\in O(L)$.
\label{lemma:property-of-varep-and-zeta-2}
\item $\varepsilon(k,k) = \varepsilon(g(k),g(k))$ for any $k\in L$ and $g\in O(L)$.\\
(More generally, we can show $\varepsilon(k,k) = \varepsilon(k',k')$ for $k,k'\in L$ satisfying $|k|^2 \equiv |k'|^2 \mod 4$.)
\label{lemma:property-of-varep-and-zeta-3}
\item $\zeta_g((t_1+t_2)k) = \zeta_g(t_1k)\zeta_g(t_2k)$ for any $t_1,t_2\in\Z$, $k\in L$, and $g\in O(L)$.
\label{lemma:property-of-varep-and-zeta-4}
\end{enumerate}
\end{lemma}

\begin{proof}
(\ref{lemma:property-of-varep-and-zeta-1}) This is because $2k \ast k' \equiv 0 \pmod{2}$.

(\ref{lemma:property-of-varep-and-zeta-2}) Use (\ref{lemma:property-of-varep-and-zeta-1}) to (\ref{eq:zeta-vs-varepsilon}).

(\ref{lemma:property-of-varep-and-zeta-3}) We first observe that
\begin{align}
|k|^2 = 2\sum_{i>j}k^ik^j e_i \cdot e_j +\sum_{i=0}^{n-1}(k^i)^2|e_i|^2.
\label{eq:norm-in-terms-of-ast-prod}
\end{align}
If $L$ is even, recall that we took the basis satisfying $|e_i|^2 \equiv 0 \mod 4$.
If $L$ is odd, since we took the basis such that $e_0$ is odd and $e_{i\neq0}$ are even, $k$ is odd if and only if $k^0$ is odd.
Moreover, we have $|e_{i\neq0}|^2 \equiv 0 \mod 4$.
By applying these facts to the second term of the right-hand side of (\ref{eq:norm-in-terms-of-ast-prod}),
we have
\begin{align}
|k|^2 \equiv 2\sum_{i>j}k^ik^j e_i \cdot e_j + \left\{\begin{array}{ll}
0 & \text{($k$ is even)}\\
|e_0|^2 & \text{($k$ is odd)}
\end{array}\right. \pmod{4}.
\end{align}
Therefore, $|k|^2=|g(k)|^2$ leads to
\begin{align}
\sum_{i>j}k^ik^j e_i \cdot e_j \equiv \sum_{i>j}(g(k))^i(g(k))^j e_i \cdot e_j \pmod{2},
\end{align}
and hence $\varepsilon(k,k)=\varepsilon(g(k),g(k))$ follows from (\ref{eq:ast-prod-reduced-again}).

(\ref{lemma:property-of-varep-and-zeta-4}) Use (\ref{lemma:property-of-varep-and-zeta-3}) to (\ref{eq:zeta-vs-varepsilon}).
Note that $\varepsilon(t_1k,t_2k)=\varepsilon(k,k)^{t_1t_2}$ follows from the bilinearity of the product $\ast$.
\end{proof}

Now, we prove Proposition \ref{prop:closed-form-of-zeta}.
\begin{proof}[Proof of Proposition \ref{prop:closed-form-of-zeta}]
By using Lemma \ref{lemma:property-of-varep-and-zeta} (\ref{lemma:property-of-varep-and-zeta-2}) and (\ref{lemma:property-of-varep-and-zeta-4}),
we have
\begin{align}
\zeta_g(k) = \Biggl(\prod_{i=0}^{n-1}\zeta_g(e_i)^{2\lfloor\frac{k^i}{2}\rfloor}\Biggr) \cdot \zeta_g\Bigl(\sum_{i'\in\{i\mid\text{$k^i$ is odd}\}}e_{i'}\Bigr).
\label{eq:zeta-after-lemma}
\end{align}
If we write the elements of $\{i\mid\text{$k^i$ is odd}\}$ as $i'_1, \ldots, i'_l$ where $i'_1<\cdots<i'_l$,
then by using (\ref{eq:zeta-vs-varepsilon}) and $\varepsilon(e_i,e_j)=1$ for $i<j$ repeatedly,
the last factor of (\ref{eq:zeta-after-lemma}) becomes
\begin{align}
\zeta_g(e_{i'_1}+\cdots+e_{i'_l}) & = \zeta_g(e_{i'_1}) \zeta_g(e_{i'_2}+\cdots+e_{i'_l}) \varepsilon(g(e_{i'_1}),g(e_{i'_2}+\cdots+e_{i'_l})) \\
& = \zeta_g(e_{i'_1}) \zeta_g(e_{i'_2}+\cdots+e_{i'_l}) \prod_{1<m'\leq l}\varepsilon(g(e_{i'_1}),g(e_{i'_{m'}})) \\
& \ \ \vdots \\
& = \prod_{m=1}^l\zeta_g(e_{i'_m}) \prod_{1\leq m<m'\leq l}\varepsilon(g(e_{i'_m}),g(e_{i'_{m'}})).
\end{align}
Therefore, we obtain (\ref{eq:closed-form-of-zeta}).
\end{proof}

\subsection{The Mathieu Group of Odd Leech Lattice}
\label{subsec:answer-for-odd-Leech}
Let us see that the answer to Question \ref{q:tangible-form} is no
in the case of the odd Leech lattice $O_{24}$.

We take $e_0,\ldots,e_{23}$ in (\ref{eq:odd-Leech-basis}) as a basis of the odd Leech lattice $O_{24}$.
This basis satisfies the assumption of Proposition \ref{prop:closed-form-of-zeta}.
The isometry group $O(O_{24})$ contains the automorphism group $M_{24}$ of the binary Golay code $G_{24}$ as a subgroup,
and its element $g \in M_{24} \subset S_{24}$ acts on $k=(k_0,\ldots,k_{23})\in O_{24} \subset \R^{24}$ from left as
\begin{align}
g(k) = (k_{g^{-1}(0)},\ldots,k_{g^{-1}(23)}), \quad \text{or equivalently,} \quad (g(k))_{g(i)} = k_i.
\end{align}

Recall that the Mathieu group $M_{24}$ is generated by the two generators $a,b$
satisfying the relations in (\ref{eq:presentation-of-M24}):
\begin{align}
a^2 & = 1, \label{eq:rel-of-M24-1}\\
b^3 & = 1,\\
(ba)^{23} & = 1,\\
(b^2aba)^{12} & = 1,\\
(b^2ab^2ababa)^5 & = 1,\\
(b^2ab^2aba)^3(b^2ababa)^3 & = 1, \label{eq:rel-of-M24-6}\\
((b^2aba)^3ba)^4 & = 1. \label{eq:rel-of-M24-7}
\end{align}
In our choice of a basis of $O_{24}$, we can take explicit elements of $a, b$ as in (\ref{eq:specific-generator-of-M24-a}), (\ref{eq:specific-generator-of-M24-b}).

For the answer to Question \ref{q:tangible-form} to be yes,
we must at least be able to find the values of $\zeta_a(e_i)$ and $\zeta_b(e_i)\in U(1)$ for $i=0,\ldots,23$,
which are compatible with the relations (\ref{eq:rel-of-M24-1})--(\ref{eq:rel-of-M24-7}).
For example, by applying the relation (\ref{eq:rel-of-M24-1}) to
the condition (\ref{eq:grp-hom-condition-on-zeta}) that $\{\zeta_g\}_g$ preserve the group structure, we have
\begin{align}
\zeta_a(a(e_i))\zeta_a(e_i) = \zeta_1(e_i),
\label{eq:example-of-eq-to-find-zeta}
\end{align}
for any $i=0,\ldots,23$.
By using the explicit form (\ref{eq:closed-form-of-zeta}) of $\zeta_g(k)$, the left-hand side of (\ref{eq:example-of-eq-to-find-zeta}) can be factorized in the form of
\begin{align}
(-1)^{\mu^i}\zeta_a(e_0)^{m^i_0}\cdots\zeta_a(e_{23})^{m^i_{23}},
\label{eq:LHS-of-example-of-eq-to-find-zeta}
\end{align}
where $\mu^i\in\Z_2, m^i_0,\ldots,m^i_{23}\in\Z$.
We further write $\zeta_a(e_j)\in U(1)$ as $(e^{\sqrt{-1}\pi})^{\rho_j}$ with $\rho_j\in\R/2\Z$.
Then (\ref{eq:LHS-of-example-of-eq-to-find-zeta}) is rewritten as
\begin{align}
(e^{\sqrt{-1}\pi})^{m^i_0\rho_0+\cdots+m^i_{23}\rho_{23}+\mu^i}.
\end{align}
On the other hand, the right-hand side of (\ref{eq:example-of-eq-to-find-zeta}) is 1, because we can show that $\zeta_1(k)=1$ for any $k\in L$ from the condition (\ref{eq:grp-hom-condition-on-zeta}).
Therefore, the equation (\ref{eq:example-of-eq-to-find-zeta}) is equivalent to
\begin{align}
m^i_0\rho_0+\cdots+m^i_{23}\rho_{23} = \mu^i \ \text{in $\R/2\Z$}.
\label{eq:eq-to-find-zeta-example-a}
\end{align}
In this way, from the relation (\ref{eq:rel-of-M24-1}), we have obtained 24 necessary conditions (\ref{eq:eq-to-find-zeta-example-a}; $i=0,\ldots,23$) for $\zeta_a(e_j)=(e^{\sqrt{-1}\pi})^{\rho_j}$ to preserve the group structure of $M_{24}$.

In the same way, from the seven relations (\ref{eq:rel-of-M24-1})--(\ref{eq:rel-of-M24-7}),
we obtain $7\times 24$ equations, which we label by $\iota=1,\cdots,7\times 24$,
\begin{align}
m^\iota_0\rho_0+\cdots+m^\iota_{23}\rho_{23} + m^\iota_{24}\rho_{24}+\cdots+m^\iota_{47}\rho_{47} = \mu^\iota \ \text{in $\R/2\Z$},
\end{align}
with respect to variables $\rho_i\in\R/2\Z$ defined by $\zeta_a(e_j)=(e^{\sqrt{-1}\pi})^{\rho_j}$ and $\zeta_b(e_j)=(e^{\sqrt{-1}\pi})^{\rho_{24+j}}$,
with coefficients $m^\iota_j\in\Z$ and $\mu^\iota_j\in\Z/2\Z\subset\R/2\Z$.
We write these equations in a matrix form
\begin{align}
M \vec{\rho} = \vec{\mu},
\label{eq:eq-to-find-zeta}
\end{align}
where the coefficient matrix $M$ is $\Z$-valued and has $7\times24$ rows and 48 columns,
the variable $\vec{\rho}$ is the $\R/2\Z$-valued column vector with 48 entries,
and $\vec{\mu}$ is the $\Z/2\Z$-valued and hence $\R/2\Z$-valued column vector with $7\times24$ entries.

It is known that for any $t \times u$ integer matrix $M\in\Z^{t\times u}$,
there exist invertible square matrices $P\in\GL(t,\Z)$ and $Q\in\GL(u,\Z)$,
such that $D=PMQ$ is in the form of
\begin{align}
D = \left(\begin{array}{ccc|c}
d_0 & & & \\
 & \ddots & & O\\
 & & d_{r-1} & \\\hline
 & O & & O
\end{array}\right),\quad d_0,\ldots,d_{r-1}\in\Z_{\neq0},\quad d_i\mid d_{i+1}\ \text{for any $i$}.
\end{align}
The entries $d_0,\ldots,d_{r-1}$ are unique for $M$ up to multiplication by elements in $\Z^\times=\{\pm1\}$,
and called the \emph{elementary divisors} of $M$.
$D$ is called the \emph{Smith normal form} of $M$,
and $M=P^{-1}DQ^{-1}$ is called the \emph{Smith decomposition} of $M$.
The matrices $P, D, Q$ for a given $M$ can be calculated by Mathematica command \texttt{SmithDecomposition[$M$]}.

Let $D=PMQ$ be the Smith normal form of the coefficient matrix $M$ in (\ref{eq:eq-to-find-zeta}).
Then the equation (\ref{eq:eq-to-find-zeta}) is equivalent to
\begin{align}
DQ^{-1}\vec{\rho} = P\vec{\mu}.
\label{eq:eq-to-find-zeta-deformed}
\end{align}
If the rank of $D$, or equivalently that of $M$, is $r$,
and if $P\vec{\mu}$ contains any non-zero element of $\R/2\Z$ somewhere in its $(r+1),\ldots,(7\times24)$-th entries,
then the equation (\ref{eq:eq-to-find-zeta-deformed}) does not have any solution $\vec{\rho}$,
and hence we cannot find $\zeta_a(e_i)$ and $\zeta_b(e_i)$ $(i=0,\ldots,23)$ preserving the group structure of $M_{24}$ and compatible with the relations (\ref{eq:rel-of-M24-1})--(\ref{eq:rel-of-M24-7}).
This means that the answer to Question \ref{q:tangible-form} is no,
and hence the Mathieu group $M_{24}$, a subgroup of the isometry group of the odd Leech lattice $O_{24}$,
does not lift to a subgroup of the automorphism group of the lattice VOA $V_{O_{24}}$.

In fact, we only need the $2\times 24$ equations from the relations (\ref{eq:rel-of-M24-1}) and (\ref{eq:rel-of-M24-6}) to bring this matter to a conclusion.
The coefficient matrix $M$ and $\vec{\mu}$ for these $2\times 24$ equations are calculated by computer (see \texttt{M24\_of\_odd\_Leech.ipynb} for the source code)
and recorded in (\ref{eq:M24-coef-mat}, \ref{eq:M24-RHS-vec}).
By the Smith decomposition $M=P^{-1}DQ^{-1}$,
we obtain the elementary divisors of $M$ as
\begin{align}
1, 1, 1, 1, 1, 1, 1, 1, 1, 1, 1, 1, 1, 1, 1, 1, 1, 1, 1, 1, 1, 1, 1, 2, 54,
\end{align}
and hence the rank of $D$ is 25, whereas 
\begin{align}
\begin{array}{rrrr rrrr l}
P\vec{\mu} = (\ 16018 & 16019 & 22266 & 16022 & 22269 & 22266 & 16020 & 22268 & \\
16020 & 16019 & 16020 & 16020 & 31230 & 77430 & 48739 & 61107 & \\
51039 & 77168 & 54837 & 70020 & 78108 & 53946 & 76206 & 62456 & \\
31928 & 0 & 0 & -2 & 0 & 0 & 22171 & 13702 & \\
23183 & 2 & 5974 & 17416 & 17762 & 0 & 0 & 7460 & \\
28878 & -2 & 2 & 0 & 5156 & 23924 & 0 & 33476 & )^T,
\end{array}
\end{align}
where the 31st and 33rd entries are nonzero in $\R/2\Z$.
Therefore, we cannot find the values of $\zeta_a(e_i)$ and $\zeta_b(e_i)$ $(i=0,\ldots,23)$ compatible with the relations (\ref{eq:rel-of-M24-1}) and (\ref{eq:rel-of-M24-6}),
and hence the subgroup $M_{24}$ of the isometry group of $O_{24}$
does not lift to a subgroup of the automorphism group of the lattice VOA $V_{O_{24}}$.

\vspace{\vertspace}

In a similar way, we can see that the subgroup $M_{23}$ of $M_{24}$ also does not lift to a subgroup of the automorphism group of $V_{O_{24}}$.
Recall that the Mathieu group $M_{23}$ is generated by two generators $a,b$ satisfying the relations in (\ref{eq:presentation-of-M23}):
\begin{align}
a^2 & = 1, \\
b^4 & = 1, \\
(ba)^{23} & = 1, \\
(b^2a)^6 & = 1, \\
(b^3aba)^6 & = 1, \\
(b^2ab^3aba)^4 & = 1, \\
(b^3a)^3(ba)^3(b^3aba)^2b^2ab^3a(ba)^3 & = 1, \label{eq:rel-of-M23-7}\\
b^2ab^3abab^2aba(b^3ab^2a)^2(b^2aba)^3 & = 1. \label{eq:rel-of-M23-8}
\end{align}

From the relations (\ref{eq:rel-of-M23-7}) and (\ref{eq:rel-of-M23-8}),
we obtain the equation $M\vec{\rho} = \vec{\mu}$ which
$\zeta_a(e_i) = (e^{\sqrt{-1}\pi})^{\rho_i}$ and $\zeta_b(e_i) = (e^{\sqrt{-1}\pi})^{\rho_{24+i}}$ should satisfy,
where the coefficient matrix $M$ and $\vec{\mu}$ are calculated by computer (see \texttt{M23\_of\_odd\_Leech.ipynb} for the source code)
and recorded in (\ref{eq:M23-coef-mat}, \ref{eq:M23-RHS-vec}).
By the Smith decomposition $M=P^{-1}DQ^{-1}$,
we obtain the elementary divisors of $M$ as
\begin{align}
1, 1, 1, 1, 1, 1, 1, 1, 1, 1, 1, 1, 1, 1, 1, 1, 1, 1, 1, 1, 1, 1, 1, 1,
\end{align}
and hence the rank of $D$ is 24, whereas 
\begin{align}
\scalebox{0.8}{$
\begin{array}{r rrrr rrrr l}
P\vec{\mu} = ( & 0 & -184045 & -57967 & -121500 & -94368 & -45774 & -62078 & -153328 & \\
& -2276 & -572 & -94178 & -81760 & -132488 & -197704 & -197578 & -182388 & \\
& -59707 & -185810 & -9243 & -5733 & -86582 & -170898 & -161384 & -99556 & \\
& 0 & -2422 & -198226 & -188029 & -118824 & -38728 & -109273 & -199800 & \\
& -172362 & -200168 & -11403 & -113189 & -102530 & -173862 & -50724 & -20268 & \\
& -165812 & -169514 & -123120 & -137854 & -136654 & -81506 & -25294 & -201016 & )^T,
\end{array}
$}
\end{align}
where the 28th, 31st, 35th, and 36th entries are nonzero in $\R/2\Z$.
Therefore, we cannot find the values of $\zeta_a(e_i)$ and $\zeta_b(e_i)$ $(i=0,\ldots,23)$ compatible with the relations (\ref{eq:rel-of-M23-7}) and (\ref{eq:rel-of-M23-8}),
and hence the subgroup $M_{23}$ of the isometry group of $O_{24}$
does not lift to a subgroup of the automorphism group of the lattice VOA $V_{O_{24}}$.

\vspace{\vertspace}

Let us restate our conclusion in the form of a theorem.
\begin{thm}
\label{thm:main-oddLeech}
Let $L$ be an odd Leech lattice $O_{24}$.
The subgroups $(\Z_2)^{24}.M_{24}$ and $(\Z_2)^{24}.M_{23}$ of the automorphism group $\Aut(V_L)$ of the lattice VOA
are non-split extensions.
In particular, the subgroups $M_{24}$ and $M_{23}$ of the isometry group $O(L)$ do not lift to a subgroup of $\Aut(V_L)$.
\end{thm}

\noindent If we use $U(1)$ instead of $\Z_2$ in the construction of $\hat{L}$ in (\ref{eq:lift-of-L}) and hence in the construction of $V_L$,
then the subgroups $U(1)^{24}.M_{24}$ and $U(1)^{24}.M_{23}$ of $\Aut(V_L)$ are non-split.
We will see that $(\Z_2\text{ or }U(1))^{24}.M_{23}$ contains $(\Z_2\text{ or }U(1))^{23}.M_{23}$ as a subgroup
around (\ref{eq:def-of-stabsubgrp}) and Appendix \ref{sec:proof-of-subgroup-of-extension},
and of course it is also non-split.

\subsection{The Conway Group of Leech Lattice}
\label{subsec:answer-for-Leech}
Let us see that the answer to Question \ref{q:tangible-form} is no
in the case of the Leech lattice $\Lambda_{24}$.
In fact, it was already shown in \cite{MR0476878};
we will check it by the similar discussion to that of the previous Section \ref{subsec:answer-for-odd-Leech}.

Recall that the isometry group $O(\Lambda_{24})$ of a Leech lattice is the Conway group $\Co_0$,
and it is generated by the two generators $A,B$ in (\ref{eq:Co0-generator-A}, \ref{eq:Co0-generator-B}) as a subgroup of $\GL(24,\Z)$.
From the matrix form (\ref{eq:sym-bilin-form-of-Leech}) of the symmetric bilinear form $q$,
we can see that the standard basis $\{\bbe_i\}_{i=0,\ldots,23}$ is a basis of the Leech lattice $(\Z^{24}, q)$
satisfying the assumption of Proposition \ref{prop:closed-form-of-zeta}.

We need some relations of the generators $A$ and $B$ to use the discussion similar to that of Section \ref{subsec:answer-for-odd-Leech}.
Unfortunately, unlike the case of the Mathieu groups,
we could not find information on the set of relations constituting the presentation of $\Co_0$ in terms of $A$ and $B$,
but the way of finding relations is not important.
After a lot of trials, we found some sets of relations
which lead to the system of equations showing that
there are no values $\zeta_A(e_0),\ldots,\zeta_A(e_{23})$ and $\zeta_B(e_0),\ldots,\zeta_B(e_{23})$
for the isometry group $\Co_0$ of the Leech lattice
to lift to a subgroup of the automorphism group of its lattice VOA.
(It would be valuable if there exists some systematic way to find a set of such relations.)

For example, a set of relations\footnote{
The relation (\ref{eq:rel1-of-Co0-1}) is taken from the semi-presentation of $\Co_1$ on the Version 3 webpage of ATLAS of Finite Group Representations \url{https://brauer.maths.qmul.ac.uk/Atlas/v3/spor/Co1/},
(\ref{eq:rel1-of-Co0-2}) is taken from the representative of the conjugacy class 22A of $\Co_1$ on the same webpage,
and the relation (\ref{eq:rel1-of-Co0-3}) was found through a brute-force search by computer.
}
\begin{align}
(Z^{21}BBAB)^{11} & = 1, \label{eq:rel1-of-Co0-1}\\
((YX^2YX)^2XYX)^{22} & = 1, \label{eq:rel1-of-Co0-2}\\
Y^3(YX)^5Y^2XB & = 1, \label{eq:rel1-of-Co0-3}
\end{align}
where $X=BA$, $Y=BBA$, and $Z=(BBABA)^2BA$,
leads to the equation $M\vec{\rho}=\vec{\mu}$ with the coefficient matrix $M$ and $\vec{\mu}$ recorded in (\ref{eq:Co0-coef-mat}, \ref{eq:Co0-RHS-vec});
these $M$ and $\vec{\mu}$ are calculated by computer
(see \texttt{Co0\_of\_Leech.ipynb} for the source code).
By the Smith decomposition $M=P^{-1}DQ^{-1}$,
we obtain the elementary divisors of $M$ as
\begin{align}
1, 1, 1, 1, 1, 1, 1, 1, 1, 1, 1, 1, 1, 1, 1, 1, 1, 1, 1, 1, 1, 1, 1, 2,
\end{align}
and hence the rank of $D$ is 24, whereas 
\begin{align}
\scalebox{0.7}{$
\begin{array}{rrrrrrrrl}
P\vec{\mu} = (\ 27709164 & 2529357 & 7489045 & 30268789 & 4526426 & 20722962 & 15136149 & 30380229 & \\
5661159 & 9704077 & 3615025 & 17793711 & 1073374 & 8631464 & 30885927 & 27407271 & \\
7011428 & 2117131 & 25166374 & 9655442 & 645641 & 12408278 & 26657215 & -5342424 & \\
0 & 2 & 0 & 0 & 0 & 0 & 0 & 0 & \\
0 & 0 & 0 & 0 & 0 & 0 & 0 & 0 & \\
0 & 0 & 0 & 0 & 0 & 0 & 0 & 0 & \\
0 & 2 & 2 & 2 & 0 & 0 & 2 & 2 & \\
0 & 0 & 0 & 0 & 0 & 2 & 0 & 0 & \\
10378735 & 0 & 9006865 & 16286994 & 0 & 25479941 & 27049788 & 31316233 &\ )^T,
\end{array}
$}
\end{align}
where for example the last entry is nonzero in $\R/2\Z$.
Therefore, we cannot find the values of $\zeta_A(e_i)$ and $\zeta_B(e_i)$ $(i=0,\ldots,23)$ compatible with the relations (\ref{eq:rel1-of-Co0-1})--(\ref{eq:rel1-of-Co0-3}),
and hence the isometry group $\Co_0$ of the Leech lattice $\Lambda_{24}$
does not lift to a subgroup of the automorphism group of the lattice VOA $V_{\Lambda_{24}}$.

\section{Currents Invariant under Specified Symmetry}
\label{sec:inv-current}
Now that we understand the symmetry of the odd Leech lattice CFT contains $U(1)^{24}.M_{24}$, instead of $M_{24}$,
our next interest is in investigating what kind of structure can be preserved under its action or the action of its subgroup.
Specifically, we determine weight-1 and weight-$\frac{3}{2}$ elements of the odd Leech lattice VOA,
which could be ingredients of a superconformal algebra,
invariant under the action of $U(1)^{24}.M_{24}$ or its subgroup $U(1)^{23}.M_{23}$.

\subsection{Lift of Lattice Isometries to Automorphisms of Lattice VOA: Part 2}
\label{subsec:autom-of-lattice-VOA}
To determine which elements of the VOA are invariant under given automorphisms,
of course we have to know how automorphisms act on the elements.
It was halfway done in Section \ref{subsec:review-of-FLM},
but we have not seen how an element of $O(\hat{L})$ extends to an element of $\Aut(V_L)$;
in the first place, we have not explained the definition of $V_L$, either.
So we will review them in this Section \ref{subsec:autom-of-lattice-VOA}.

The foundational literature is \cite{MR0996026},
and \cite{MR1745258, Lam} also contain readable summaries.
These references deal with ordinary lattice VOAs constructed from even lattices,
but we also consider odd lattices in this paper,
so the resulting VOA could be a vertex operator \emph{super}algebra (VOSA) to be precise (see e.g.\ \cite[Remark 12.38]{MR1233387}).
However, this distinction is not significant for our purposes,
so let us be unconcerned about that point.
We assume that the lattice is Euclidean, and hence the resulting lattice VOA formulates a chiral CFT in the language of physics.
As for Lorentzian lattice VOAs or more general VOAs to describe non-chiral CFTs (full CFTs in other words),
see \cite{Huang:2005gz, Moriwaki:2020cxf, Singh:2023mom}.

As always, we assume $L$ is an integral lattice.
Let $\h=L\otimes_\Z\C$ be the abelian Lie algebra,
and $\hat{\h} = \h\otimes\C[t,t^{-1}] \oplus \C K$ be its affine Lie algebra with the Lie bracket
\begin{align}
[\alpha(m), \alpha'(m')] = (\alpha,\alpha') m \delta_{m+m',0}  K,
\end{align}
where $\alpha(m) := \alpha\otimes t^m$ with $\alpha\in\h$, $m\in\Z$,
and the symmetric bilinear form $({-},{-})$ on the lattice $L$ is extended to $\h$ by $\C$-linearity.

Let $U(\hat{\h})$ be the universal enveloping algebra of $\hat{\h}$
(the algebra where $\alpha(m)\alpha'(m')-\alpha'(m')\alpha(m) = [\alpha(m),\alpha'(m')]$ holds),
and define the $\hat{\h}$-module $M(1)$ as
\begin{align}
M(1) := U(\hat{\h}) \otimes_{\h\otimes\C[t] \oplus \C K} \C,
\end{align}
where $\h\otimes\C[t]$ acts on $\C$ as $\alpha(m)\cdot \C = 0$ ($m\geq0$),
and $K$ acts on $\C$ as multiplication by 1.
As a vector space, $M(1)$ is isomorphic to $U(\hat{\h}^-)$,
which is the universal enveloping algebra of the abelian subalgebra
\begin{align}
\hat{\h}^- := \h \otimes t^{-1}\C[t^{-1}],
\end{align}
of $\hat{\h}$.

Recall that $\hat{L}$ is the central extension of $L$ by $\Z_2$ or $U(1)$
whose generator was denoted by $\kappa$.
We may consider a representation $\C$ of the subgroup $\Z_2$ or $U(1)$ of $\hat{L}$,
where the generator $\kappa$ acts on $\C$ as multiplication by $e^{\sqrt{-1}\pi}$.
We define the $\hat{L}$-module $\C\{L\}$ as the induced representation of $\hat{L}$ from it:
\begin{align}
\C\{L\} := \mathrm{Ind}^{\hat{L}}_{\text{$\Z_2$ or $U(1)$}} \C = \C[\hat{L}] \otimes_{\C[\text{$\Z_2$ or $U(1)$}]} \C,
\end{align}
where $\C[G]$ denotes the group algebra of $G$ over $\C$.
Roughly speaking, $\C\{L\}$ is just a group algebra $\C[\hat{L}]$ with $\kappa^r \cdot e^k = e^{\sqrt{-1}\pi r}e^k$.

Now, the lattice VOA $V_L$ is defined as the $\C$ vector space
\begin{align}
V_L := M(1) \otimes_\C \C\{L\},
\end{align}
with the VOA structure such as the vacuum vector and the state-field correspondence, which we do not write down here.
Any element of $V_L$ therefore can be written as a $\C$-linear combination of elements in the form of
\begin{align}
\alpha_1(-m_1)\cdots\alpha_l(-m_l) e^k,
\label{eq:VOA-monomial}
\end{align}
where $\alpha_1,\ldots,\alpha_l \in \h = L\otimes_\Z\C$, $m_1,\ldots,m_l \in \Z_{>0}$, and $k\in L$.
The \emph{(conformal) weight} or the \emph{(conformal) dimension} of the element (\ref{eq:VOA-monomial}) is defined as
\begin{align}
\wt(\alpha_1(-m_1)\cdots\alpha_l(-m_l) e^k) := m_1 + \cdots + m_l + \frac{1}{2}|k|^2,
\end{align}
and we can introduce the grading by this weight to the VOA $V_L$.
An element of a VOA is sometimes called a \emph{current} because of some physical background.

Let $e_0,\ldots,e_{n-1}$ be an integral basis of $L$.
The element $e_{i_1}(-m_1)\cdots e_{i_l}(-m_l) e^k$ of $V_L$ is usually denoted by $\alpha^{i_1}_{-m_1}\cdots\alpha^{i_l}_{-m_l}|k\rangle$ in physics literature,
where the commutation relation of the creation-annihilation operators $\alpha^i_m$ is $[\alpha^i_m,\alpha^{i'}_{m'}] = (e_i,e_{i'})m\delta_{m+m',0}$.
We also often use the creation-annihilation operators $\bbalpha^i_m:=\bbe_i(m)$
with respect to the standard basis $\bbe_0,\ldots,\bbe_{n-1}$ of $\R^n$ where the lattice $L$ is embedded.
Their commutation relation is of course $[\bbalpha^i_m,\bbalpha^{i'}_{m'}] = (\bbe_i,\bbe_{i'})m\delta_{m+m',0}$.

The definition of an \emph{automorphism} of a VOA $V$ is
a map $F:V\to V$ such that
\vspace{-8pt}
\begin{enumerate}
\setlength{\itemsep}{-3pt}
\leftskip -8pt
\item it is a $\C$-linear automorphism on the $\C$-vector space $V$.
\item it preserves OPE; more precisely, $F(v_m(v')) = (F(v))_m(F(v'))$ for any $v,v'\in V$.
Here, $v_m\in\End(V)$ denotes the $m$-th mode of the mode expansion of the operator corresponding to the state $v\in V$.
\item it preserves the Virasoro element; $F(\omega)=\omega$.
Here, the \emph{Virasoro element} $\omega\in V$ is a specified state such that the modes of the corresponding operator, called the \emph{energy-momentum tensor} $T(z)$, satisfy the Virasoro algebra.
In a lattice VOA, $\omega$ is given by $\omega=\frac{1}{2}\sum_i\bbe_i(-1)^2 e^0$,
and the corresponding operator is $T(z)=-\frac{1}{2}\sum_i:(\partial X^i(z))^2:$ (up to scalar multiplication).
\end{enumerate}
\vspace{-8pt}
We refer the reader to \cite[\S8.10]{MR0996026} or \cite[\S2.3]{MR1745258} for more details.

Any element $f$ of $O(\hat{L})$, which was introduced in (\ref{eq:def-of-O(Lhat)}),
induces an automorphism $F$ of $V_L$ by
\begin{align}
F(\alpha_1(-m_1)\cdots\alpha_l(-m_l) e^k) = \bar{f}(\alpha_1)(-m_1) \cdots \bar{f}(\alpha_l)(-m_l) f(e^k),
\label{eq:autom-of-lattice-VOA}
\end{align}
where $\bar{f}\in O(L)$ was defined in (\ref{eq:def-of-barf}) and extended by $\C$-linearity here.
This preserves the group structure of $O(\hat{L})$,
and hence $O(\hat{L}) \cong (\text{$\Z_2$ or $U(1)$})^n.O(L)$ can be regarded as a subgroup of $\Aut(V_L)$.
The whole $\Aut(V_L)$ is determined in \cite{MR1745258}.

\subsection{Invariant Currents of Odd Leech Lattice VOA}
\label{subsec:inv-current-odd-Leech}
In this Section \ref{subsec:inv-current-odd-Leech},
let $L$ be the odd Leech lattice $O_{24}$ constructed as in (\ref{eq:odd-Leech-from-Golay-code}).
The automorphism group $\Aut(V_L)$ of the lattice VOA contains $O(\hat{L})=U(1)^{24}.O(L)$, and hence $U(1)^{24}.M_{24}$, as a subgroup.
Let us investigate the weight-1 and weight-$\frac{3}{2}$ currents of $V_L$, invariant under the action of $U(1)^{24}.M_{24}$.

To begin with, any weight-1 current of $V_L$ is in the form of
\begin{align}
\alpha(-1)e^0,
\label{eq:wt-1-in-odd-Leech}
\end{align}
where $\alpha\in L\otimes_\Z\C$.
Since $e^0$ is the identity element of $\hat{L}$, we have $f(e^0)=e^0$ for any $f\in O(\hat{L})$.
Therefore, an element (\ref{eq:wt-1-in-odd-Leech}) is invariant under the action (\ref{eq:autom-of-lattice-VOA}) of $U(1)^{24}.M_{24} \subset \Aut(V_L)$,
if and only if $\alpha\in L\otimes_\Z\C$ is invariant under the action of $M_{24}\subset O(L)$ extended by $\C$-linearity.
Since the action of $M_{24}$ on $L=O_{24}$ constructed as in (\ref{eq:odd-Leech-from-Golay-code}) is the permutation of the coordinates with respect to the standard basis $\bbe_0,\ldots,\bbe_{23}$ of $\R^{24}$ where $L$ is embedded,
and it is transitive,
such $\alpha$ is only $\underline{1}=(1,\ldots,1)=\sum_i\bbe_i$, up to multiplication by $\C$.
In conclusion, the weight-1 current of $V_L$ invariant under the action of $U(1)^{24}.M_{24}$ is
\begin{align}
j^{M_{24}} = \sum_{i=0}^{23}\bbe_i(-1)e^0,
\end{align}
up to multiplication by $\C$.
In physics notation,
\begin{align}
j^{M_{24}} = \sum_{i=0}^{23}\bbalpha^i_{-1}|0\rangle.
\end{align}

Next, any weight-$\frac{3}{2}$ current of $V_L$ is a linear combination of elements in the form of
\begin{align}
e^k \quad \text{with} \quad k = \frac{1}{2\sqrt{2}}(-1)^w, \quad w\in G_{24},
\label{eq:wt-3/2-in-odd-Leech}
\end{align}
where the codeword $w$ of the binary Golay code $G_{24}$ is regarded as a vector $(w_i)_i$ with entries 0 or 1, and $(-1)^w$ denotes the vector $((-1)^{w_i})_i$.
Such elements are not invariant under the action of $\Hom(L,U(1))\cong U(1)^{24}\subset \Aut(V_L)$, recalling its action (\ref{eq:def-of-tilde}).
Therefore, there is no weight-$\frac{3}{2}$ current of $V_L$ invariant under the action of $U(1)^{24}.M_{24}$.
Even if we imposed the invariance under the action of $\Z_2^{24}.M_{24}$ only,
there is still no invariant weight-$\frac{3}{2}$ current of $V_L$.
This is because $k=\frac{1}{2\sqrt{2}}(-1)^w$ is an odd vector,
and hence its component with respect to the basis $e_0$ in (\ref{eq:odd-Leech-basis}) is odd,
so its dual $(e_0)^\vee \in \Hom(L,\Z_2)$ always gives a minus sign to (\ref{eq:wt-3/2-in-odd-Leech}).

\vspace{\vertspace}

How about a smaller subgroup?
For example, the subgroup of $M_{24}$ stabilizing one coordinate, say $k_0$ of $k=\sum_ik_i\bbe_i\in L$, is $M_{23}$.
In addition, if we consider the subgroup
\begin{align}
\Stab_{\Hom(L,U(1))}(\bbe_0) := \{\eta\in\Hom(L,U(1)) \mid \eta(L\cap\R\bbe_0)=1\} \cong U(1)^{23},
\label{eq:def-of-stabsubgrp}
\end{align}
of $\Hom(L,U(1))$,
then $\Stab_{\Hom(L,U(1))}(\bbe_0).M_{23}$ is a subgroup of $\Hom(L,U(1)).M_{24} \subset O(\hat{L})$.
To see it,\footnote{
The fact that $\Stab_{\Hom(L,U(1))}(\bbe_0).M_{23}$ is a subgroup of $\Hom(L,U(1)).M_{24}$ is nontrivial in the following sense.
Let $\hat{G}$ be an group extension $N.G$ of $G$ by $N$.
For any subgroup $G'$ of $G$, $\hat{G}$ contains an extension $N.G'$ as a subgroup.
However, for a subgroup $N'$ of $N$, $N'.G$ is not a subgroup of $\hat{G}$ in general.
This is because the multiplication law of $\hat{G}$ is like
$(n_1, g_1) \cdot (n_2, g_2) = (n_1g_1(n_2)\epsilon(g_1,g_2), g_1g_2)$;
see (\ref{eq:multi-law-of-grp-ext}) for the precise expression.
} we check that the 2-cocycle $E:M_{24}\times M_{24}\to\Hom(L,U(1))$ associated to the extension $O(\hat{L})$ in (\ref{eq:lift-of-O(L)})
takes value in $\Stab_{\Hom(L,U(1))}(\bbe_0)$ when restricted to $E|_{M_{23}\times M_{23}}$.
See Appendix \ref{sec:proof-of-subgroup-of-extension} for the detailed proof.

Let us investigate the weight-1 and weight-$\frac{3}{2}$ currents of $V_L$,
invariant under the action of\\
$\Stab_{\Hom(L,U(1))}(\bbe_0).M_{23} \cong U(1)^{23}.M_{23}$.
By an analysis similar to the case of $U(1)^{24}.M_{24}$,
the weight-1 invariant currents are
\begin{align}
j^{M_{23}}_1 = \bbe_0(-1)e^0,\quad j^{M_{23}}_2 = \sum_{i=1}^{23}\bbe_i(-1)e^0,
\end{align}
and their $\C$-linear combinations.
In physics notation,
\begin{align}
j^{M_{23}}_1 = \bbalpha^0_{-1}|0\rangle,\quad j^{M_{23}}_2 = \sum_{i=1}^{23}\bbalpha^i_{-1}|0\rangle.
\label{eq:M23-inv-weight-1-physics}
\end{align}
There is no weight-$\frac{3}{2}$ invariant current again.
Even if we imposed the invariance under the action of $\Z_2^{23}.M_{23}$ only,
there is still no invariant weight-$\frac{3}{2}$ current.
This is because the basis $e_{23}$ in (\ref{eq:odd-Leech-basis}) is in fact a primitive vector of $L\cap\R\bbe_0$,
and hence $\Stab_{\Hom(L,\Z_2)}(\bbe_0)=\Span_{\Z_2}\{(e_0)^\vee,\ldots,(e_{22})^\vee\}$,
where $(e_i)^\vee$ are the dual basis of (\ref{eq:odd-Leech-basis}).
Therefore, $(e_0)^\vee$ is in $\Stab_{\Hom(L,\Z_2)}(\bbe_0)$,
and it always gives a minus sign to (\ref{eq:wt-3/2-in-odd-Leech})
by the same discussion as in the case of $\Z_2^{24}.M_{24}$.

\section{Discussions}
\label{sec:discussions}
In this paper, we have seen that
the symmetry $\Aut(V_{O_{24}})$ of the odd Leech lattice VOA
inherits the $M_{24}$ or $M_{23}$ symmetry of the odd Leech lattice $O_{24}$
in the form of the non-split extension $U(1)^{24}.M_{24}$ or $U(1)^{23}.M_{23}$,
and therefore $M_{24}$ and $M_{23}$ are not genuine symmetries of the lattice VOA.
We also determined weight-1 and weight-$\frac{3}{2}$ currents in the odd Leech lattice VOA
invariant under the symmetry $U(1)^{24}.M_{24}$ or its subgroup $U(1)^{23}.M_{23}$.
In particular, we concluded that there is no such weight-$\frac{3}{2}$ current.

Let us revisit the observation by \cite{Benjamin:2015ria}.
They pointed out that representation dimensions of $M_{24}$ appear in the decomposition of the $\cN=2$ extremal elliptic genus $Z_\mathrm{ext}^{m=4}(\tau, z)$ of central charge $c=6m=24$ into the $(-1)^F$-inserted characters of the Ramond representations of an $\cN=2$ SCA.
They further constructed a chiral $\cN=2$ SCFT
whose partition function of the R sector coincides with the extremal elliptic genus $Z_\mathrm{ext}^{m=4}(\tau, z)$,
although a slight modification is needed for their supercurrents as in Appendix \ref{sec:N=2-SCA-in-odd-Leech}.
Their theory is equivalent to the fermionization of the $(A_1)^{24}$ lattice CFT with respect to the shift $\Z_2$ symmetry,
and its NS sector can be described as the odd Leech lattice VOA $V_{O_{24}}$,
denoted by $V_L$ for short below.

\vspace{\vertspace}

Now, we notice that there are more obstacles in trying to explain the moonshine-type observation by this fermionized $(A_1)^{24}$ theory than we naively expected.
The apparent problem is why it is the representation dimensions of $M_{24}$ that appear in the decomposition of the partition function of the R sector into the $\cN=2$ SCA characters,
even though the genuine symmetry of the VOA of the NS sector is the non-split extension $(\Z_2\text{ or }U(1))^{24}.M_{24}$.
Let us still try to explore a perspective in which $M_{24}$ or its subgroup $M_{23}$ takes on a meaningful role.
As we will see, this is not straightforward at all.

Here, we note that the difference between the NS sector and the R sector is not important.
One reason is that, using the spectral flow of the $\cN=2$ SCA,
we see that the partition function of the NS sector should also exhibit the representation dimensions of $M_{24}$ under the decomposition into $\cN=2$ SCA characters.
Another reason is that,
since the R sector is not closed under the OPE,
we cannot consider a new CFT consisting of the R sector only,
and the NS sector always comes into the story,
so there is no hope that the symmetry of CFT contains $M_{24}$ as a subgroup.

\begin{rem*}
In light of VOA,
the R sector can be formulated as a twisted module $M_\rR$ of the NS sector $V_L$.
We may consider the automorphism group of this module $\Aut_{V_L}(M_\rR)=(\End_{V_L}(M_\rR))^\times$,
but it would be too small to provide any useful implications.
For example, if $M_\rR$ is an irreducible module of $V_L$,
then $\Aut_{V_L}(M_\rR)=\C^\times$ by Schur's lemma.
We can try to find a subalgebra $\cA$ of $V_L$,
say some algebra containing an $\cN=2$ SCA,
such that $\End_\cA(M_\rR)$ contains $\C[M_{24}]$ or $\C[M_{23}]$,
but this attempt is not much different from the following approach focusing on the NS sector.
\emph{(Remark ends.)}
\end{rem*}

\vspace{\vertspace}

As such, to explain the appearance of the representation dimensions of $M_{24}$,
we have to exclude the effect of the group extension.
One possible way to do so
is to construct some VOA-like object,
for example by using $\hat{L}/\Z_2\cong L$,
on which $(2^{24}.M_{24})/2^{24}\cong M_{24}$ can act.
The discussion on whether this idea works is beyond the scope of this paper,
and it might require further study.

\vspace{\vertspace}

Another way\footnote{
The author thanks the anonymous referee for providing this viewpoint
of focusing on a subVOA having an $M_{24}$ symmetry.
} is to consider a subVOA of $V_L$ fixed by $2^{24}\subset 2^{24}.M_{24}$.
Such a subVOA admits the action of $(2^{24}.M_{24})/2^{24}\cong M_{24}$.
Recalling how this $2^{24}\cong\Hom(L,\Z_2)$ acts on $V_L$ (\ref{eq:def-of-tilde}),
we can see that the fixed subVOA is the lattice VOA $V_{2L}$,
where $2L$ is the sublattice of $L$ spanned by twice the vectors of $L$.
We can also check that $M_{24}$ in fact acts on this $V_{2L}$
by observing that the cocycle $\varepsilon:L\times L\to\Z_2$ restricted to $2L$ is trivial.

This lattice subVOA $V_{2L}$ has $4^{24}$ irreducible modules $V_{\lambda+2L}$ labeled by $[\lambda] \in (2L)^\ast/2L$.
If we write a basis of the odd Leech lattice $L$ as $\{e_i\}_{i=0,\ldots,23}$ and its dual basis of $L^\ast$ as $\{(e_i)^\vee\}_{i=0,\ldots,23}$,
since the odd Leech lattice is self-dual,
we can take a representative of $[\lambda] \in (2L)^\ast/2L$ in the form of $\lambda = \sum_{i=0}^{23}\lambda^i(e_i)^\vee$ with $\lambda^i\in\{0,\frac{1}{2}, 1, \frac{3}{2}\}$.
Then, the original lattice VOA $V_L$ as a module of $V_{2L}$ is the sum $\oplus_{\lambda^i\in\{0,1\}}V_{\lambda+2L}$ of the $2^{24}$ irreducible modules.

\begin{rem*}
More generally, for each $(e_{j_0})^\vee+\cdots+(e_{j_k})^\vee \in \Hom(L,\Z) \cong 2^{24}$,
the sum\\
$V_{\frac{1}{2}(e_{j_0})^\vee+\cdots+\frac{1}{2}(e_{j_k})^\vee+L}:=\oplus_{\lambda^i\in\{0,1\}}V_{\lambda+\frac{1}{2}(e_{j_0})^\vee+\cdots+\frac{1}{2}(e_{j_k})^\vee+2L}$ of the $2^{24}$ irreducible modules
is the $((e_{j_0})^\vee+\cdots+(e_{j_k})^\vee)$-twisted module of the original lattice VOA $V_L$.
\textit{(Remark ends.)}
\end{rem*}

Since this lattice VOA $V_{2L}$ is rational and $C_2$-cofinite,
the characters (bosonic partition functions) of these irreducible modules of $V_{2L}$ are closed under the modular transformations.
Note that $V_{2L}$ itself is also a module of $V_{2L}$.
Since $V_{2L}$ has a genuine $M_{24}$ symmetry,
the character of $V_{2L}$ itself exhibits representation dimensions of $M_{24}$ under the decomposition into Virasoro characters.
We may even define the analogues of the McKay-Thompson series, and their twisted versions, for this module $V_{2L}$.
However, since the modules other than $V_{2L}$ do not admit the action of $M_{24}$ in general
(otherwise, $V_L$ would admit the $M_{24}$-action, which is a contradiction),
it is difficult to directly explain the appearance of representation dimensions of $M_{24}$ in the partition function of $V_L$,
or to define such $M_{24}$-inserted characters for every module of $V_{2L}$ and expect their modular properties.
In fact, $M_{24}$ acts\footnote{
Let $V_{\lambda+2L}$ be an irreducible module of $V_{2L}$, that is, a vector space with a $V_{2L}$-action.
By applying $g\in M_{24}$ on $v\in V_{2L}$ before $v$ acts on $V_{\lambda+2L}$,
we can define another $V_{2L}$-action on $V_{\lambda+2L}$,
which makes it isomorphic to another irreducible module $V_{\lambda'+2L}$.
This is how $M_{24}$ acts on the set of the irreducible modules of $V_{2L}$.
} on the set $(2L)^\ast/2L$ of (the labels of) the irreducible modules,
but the module $V_{\lambda+2L}$ is just a projective representation of the subgroup $G_{[\lambda]} \subset M_{24}$ stabilizing $[\lambda]\in (2L)^\ast/2L$
because of the non-trivial cocycle factor,
and by putting all these information together,
we just come back to the $2^{24}.M_{24}$-action on $V_L$.

\vspace{\vertspace}

Another rather haphazard way
is to forget the information of cocycle factors,
even though it destroys the OPE structure.
Let us investigate this possibility step by step.

\vspace{\vertspace}

\noindent \textbf{Review of the definition of an automorphism of VOA}

As reviewed in Section \ref{subsec:autom-of-lattice-VOA},
the definition of an automorphism of the VOA $V_L$ is a map $F:V_L\to V_L$ such that
\begin{enumerate}
\setlength{\itemsep}{-3pt}
\leftskip -8pt
\item \label{cond:autom-C-linearity} it is a $\C$-linear automorphism on the $\C$-vector space $V_L$,
\item it preserves OPE; more precisely, $F(v_m(v')) = (F(v))_m(F(v'))$ for any $v,v'\in V_L$,
\item \label{cond:autom-preserve-Virasoro} it preserves the Virasoro element; $F(\omega)=\omega$.
\end{enumerate}
Such maps constitute the automorphism group $\Aut(V_L)$,
and it contains $\Hom(L,\Z_2).O(L)$ as a subgroup as reviewed in Section \ref{subsec:review-of-FLM}.

\vspace{\vertspace}

\noindent \textbf{Modifying the conditions to discard the effect of cocycle factors}

To discard the effect of cocycle factors,
let us modify these conditions (\ref{cond:autom-C-linearity})--(\ref{cond:autom-preserve-Virasoro})
and consider a map $F:V_L\to V_L$ such that
\begin{enumerate}
\setlength{\itemsep}{-3pt}
\leftskip -8pt
\item \label{cond:C-linearity} it is a $\C$-linear automorphism on the $\C$-vector space $V_L$,
\item \label{cond:preserve-OPE-up-to-cocycle} it can be written in the form of (\ref{eq:autom-of-lattice-VOA}):
\begin{align*}
F(\alpha_1(-m_1)\cdots\alpha_l(-m_l) e^k) = \bar{f}(\alpha_1)(-m_1) \cdots \bar{f}(\alpha_l)(-m_l) f(e^k),
\end{align*}
by using some map $f:\hat{L}\to\hat{L}$ such that $\bar{f}:L\to L; k\mapsto \overline{f(e^k)}$ is an isometry on $L$,
\item \label{cond:preserve-Virasoro} it preserves the Virasoro element; $F(\omega)=\omega$,
\item \label{cond:preserve-Virasoro-OPE} it preserves the OPE action by $\omega$; $F(\omega_m(v')) = \omega_m(F(v'))$ for any $v'\in V_L$.
\end{enumerate}
In the condition (\ref{cond:preserve-OPE-up-to-cocycle}),
such $f:\hat{L}\to\hat{L}$ satisfies
\begin{align}
f(e^ke^{k'})=\kappa^\text{0 or 1}f(e^k)f(e^{k'})
\end{align}
in $\hat{L}$.
This means that $F$ preserves OPE of monomial elements
up to cocycle factors;
more precisely, for any $v,v'\in V_L$ in the form of $\alpha_1(-m_1)\cdots\alpha_l(-m_l) e^k$,
\begin{align}
F(v_m(v'))=\pm(F(v))_m(F(v'))
\end{align}
holds,
which can be seen from the explicit formula \cite[\S3.4]{MR1745258} of $(e^k)_m(e^{k'})$.
The set of all maps $F:V_L\to V_L$ satisfying these conditions (\ref{cond:C-linearity})--(\ref{cond:preserve-Virasoro-OPE}) constitute a group under the composition,
which we write as $\LinAut(V_L;\Vir)_{O(L)}$.

The last condition (\ref{cond:preserve-Virasoro-OPE})  guarantees that the action of this group $\LinAut(V_L;\Vir)_{O(L)}$ does not change the $L_0$-eigenvalue of states of $V_L$,
where $L_0$ is the 0-th mode of the operator $T(z)$ corresponding to the Virasoro element $\omega$,
so that we can expect that $\LinAut(V_L;\Vir)_{O(L)}$ accounts for some properties of the decomposition of a partition function $\Tr_{V_L}q^{L_0-\frac{c}{24}}$ into the Virasoro characters.
Here we considered a simple partition function rather than $\Tr_{V_L}(-1)^Fy^{J_0}q^{L_0-\frac{c}{24}}$ respecting the $\cN=2$ SCA,
just for simplicity.
Fortunately,
this condition (\ref{cond:preserve-Virasoro-OPE}) automatically follows from the other conditions (\ref{cond:C-linearity})--(\ref{cond:preserve-Virasoro}),
because the action of modes $\omega_m$ of the operator $T(z) \propto -\frac{1}{2}\sum_i:(\partial X^i(z))^2:$ does not change the lattice vector $k$ of an element $\alpha_1(-m_1)\cdots\alpha(-m_l)e^k\in V$,
and hence this condition (\ref{cond:preserve-Virasoro-OPE}) does not lead to any nontrivial constraint caused by cocycle factors.

As a result, this group $\LinAut(V_L;\Vir)_{O(L)}$ turns out to be $\Map(L,\Z_2):O(L)$,
where $\Map(L,\Z_2)$ is the set of all the maps $L\to\Z_2$,
which constitutes a group under the multiplication.
It is greatly enlarged from $\Hom(L,\Z_2).O(L)$,
and in particular contains $M_{24}\subset O(L)$ preserving the momentum lattice $L=O_{24}$.

\vspace{\vertspace}

\noindent \textbf{Taking the $\cN=2$ SCA into account}

We only focused on the Virasoro algebra so far.
Can we apply a similar discussion to $\cN=2$ SCA?
The generators of $\cN=2$ SCA are the energy-momentum tensor $T(z)$, a $U(1)$ current $J(z)$ of weight 1, and supercurrents $G^\pm(z)$ of weight $\frac{3}{2}$.
As discussed in \cite{Benjamin:2015ria} and reviewed around (\ref{eq:lower-order-terms-of-partition-func}),
we can choose the $U(1)$ current $J(z)$ as $2\sqrt{-2}\partial X^0(z)$,
to reproduce the extremal elliptic genus as the partition function $\Tr_\rR(-1)^Fy^{J_0}q^{L_0-\frac{c}{24}}$ of the R sector.
This $U(1)$ current is exactly $j_1^{M_{23}}=\bbalpha^0_{-1}|0\rangle$ in (\ref{eq:M23-inv-weight-1-physics}) up to scalar multiplication,
and invariant under the action of $2^{23}.M_{23} \subset \Aut(V_L)$.
However, there is no weight-$\frac{3}{2}$ current invariant under $2^{23}.M_{23}$ as we have seen in Section \ref{subsec:inv-current-odd-Leech}.
In addition, any weight-$\frac{3}{2}$ currents invariant under the subgroup $M_{23}$ of $\LinAut(V_L;\Vir)_{O(L)}$ 
does not satisfy the OPEs of supercurrents of $\cN=2$ SCA,
as we will show in Remark of Appendix \ref{sec:N=2-SCA-in-odd-Leech}.
The subgroup of $2^{24}.M_{24}\subset\Aut(V_L)$ preserving the supercurrents we will construct in Appendix \ref{sec:N=2-SCA-in-odd-Leech} would be too small to provide any useful implications
(see the last paragraph of Appendix \ref{sec:N=2-SCA-in-odd-Leech}).

Therefore, what we can do is to deal with the subalgebra $\cA$ of the $\cN=2$ SCA generated by $T(z)$ and $J(z)$ only,
and consider a group $\LinAut(V_L;\cA)_{O(L)}$ consisting of maps $F:V_L\to V_L$ such that
\begin{enumerate}
\setlength{\itemsep}{-3pt}
\leftskip -8pt
\item it is a $\C$-linear automorphism on the $\C$-vector space $V_L$,
\item it can be written in the form of (\ref{eq:autom-of-lattice-VOA}):
\begin{align*}
F(\alpha_1(-m_1)\cdots\alpha_l(-m_l) e^k) = \bar{f}(\alpha_1)(-m_1) \cdots \bar{f}(\alpha_l)(-m_l) f(e^k),
\end{align*}
by using some map $f:\hat{L}\to\hat{L}$ such that $\bar{f}:L\to L; k\mapsto \overline{f(e^k)}$ is an isometry on $L$,
\item it preserves the Virasoro element $F(\omega)=\omega$
and the $U(1)$ current $F(j_1^{M_{23}}) = j_1^{M_{23}}$,
\item it preserves the OPE action by $\omega$ and $j_1^{M_{23}}$; $F(\omega_m(v')) = \omega_m(F(v'))$ and $F((j_1^{M_{23}})_m(v')) = (j_1^{M_{23}})_m(F(v'))$ for any $v'\in V_L$.
\end{enumerate}
Then the group $\LinAut(V_L;\cA)_{O(L)}$ contains $M_{23}$ as a subgroup, through a similar discussion to that on $M_{24}\subset\LinAut(V_L;\Vir)_{O(L)}$,
and it might account for the appearance of representation dimensions of $M_{23}$ in the decomposition of the partition function $\Tr_{V_L}(-1)^Fy^{J_0}q^{L_0-\frac{c}{24}}$ into the characters of the algebra $\cA$.

However, if we put together the characters of the algebra $\cA$ into the characters of $\cN=2$ SCA,
then we can only say that the resulting coefficients are virtual representation dimensions (where not only addition but also subtraction of irreducible representation dimensions are allowed) of $M_{23}$.
Moreover, in the first place, the whole group $\LinAut(V_L;\cA)_{O(L)}$ is much larger than $M_{23}$;
it is in fact $\Map(L,\Z_2):(2^{11}:M_{23})$,
where $2^{11}:M_{23}$ is the subgroup of $O(L)\cong 2^{12}:M_{24}$ stabilizing the first entry of the lattice vectors.
So it is unclear why the subgroup $M_{23}$ is only in effect,
compared with other moonshine phenomena such as the monstrous moonshine and the Conway moonshine,
where the automorphism groups of their VOAs are precisely the monster group $\mathbb{M}$ and the Conway group $\Co_1$.
Anyway, since we have destroyed the OPE structure, and hence the VOA structure,
it is difficult to imagine that the discussions here would yield so fruitful outcomes as, for example, the monstrous moonshine.

\vspace{\vertspace}

In conclusion, the situation surrounding the moonshine-type observation on the extremal elliptic genus seems more mysterious than we think.
We have only dealt with the VOA of odd Leech lattice CFT,
or the fermionized $(A_1)^{24}$ lattice CFT,
but of course there is a possibility that
another new theory explains the observation in a nice way. 
It would be interesting if that is the case,
and if it also provides some implications for other moonshine phenomena including the K3 Mathieu moonshine.

\section*{Acknowledgements}
\label{sec:acknowledge}
\addcontentsline{toc}{section}{\hspace{18pt}Acknowledgements}
The author thanks Yuji Tachikawa for a lot of discussions and comments on the draft,
and Kohki Kawabata and Shinichiro Yahagi for useful comments.
The author thanks Scott Carnahan, Gerald H\"{o}hn, and Martin Seysen for providing information on the existing results of the extension of the Conway group.

The author is supported by FoPM (WINGS Program of the University of Tokyo), 
JSPS Research Fellowship for Young Scientists,
JSPS KAKENHI Grant No.\ JP23KJ0650,
and in part by WPI Initiative, MEXT, Japan at Kavli IPMU, the University of Tokyo.

\appendix

\section{Group Extension}
\label{sec:group-extension}
In this Appendix \ref{sec:group-extension},
we give an elementary introduction to group extensions and a few related topics.
In Section \ref{subsec:grp-ext-and-grp-coh},
we review basic notions of group extensions following \cite[Ch.\ IV]{MR0672956},
and see that the equivalence classes of group extensions of a group $G$ by a $G$-module $N$ are in one-to-one correspondence with the cohomology classes of the second group cohomology $H^2(G,N)$.
In the next Section \ref{subsec:cent-ext-and-comm-map}, we focus on central extensions of abelian groups,
and establish another one-to-one correspondence with commutator maps
in the case of free abelian groups,
following \cite[\S5.2]{MR0996026}.
The last Section \ref{subsec:thm-on-aut-grp-of-central-ext} is just a review of a certain theorem \cite[Prop.\ 5.4.1]{MR0996026} on the automorphism group of a central extension,
which plays a vital role in the discussion of the main text.

\subsection{Group Extensions and Group Cohomology}
\label{subsec:grp-ext-and-grp-coh}
An \emph{extension} of a group $G$ by a group $N$ is a short exact sequence of groups and homomorphisms
\begin{align}
1 \to N \overset{i}{\to} \hat{G} \overset{\pi}{\to} G \to 1.
\label{eq:def-of-group-extension}
\end{align}
(Be aware that some literature calls it an extension \emph{of $N$ by $G$}.)
When there is no risk of confusion in the homomorphisms constituting the short exact sequence, we only say that $\hat{G}$ is an extension of $G$ by $N$.
(But note that the equivalence of group extension, defined later, classifies not only $\hat{G}$, but also the whole short exact sequence.)
We also use the notation $N.G$ for any extension of $G$ by $N$.
Since $i(N)$ is the kernel of $\pi$, $N$ can be regarded as a normal subgroup of $\hat{G}$,
and the quotient group $\hat{G}/i(N)$ is isomorphic to $G$.

Let us take a set-theoretical section (not necessarily a group homomorphism) $s:G\to\hat{G}$.
Since $s(g)s(g')s(gg')^{-1} \in \ker\pi = i(N)$ for $g,g' \in G$, we can define a function $\epsilon:G\times G\to N$ which measures how $s$ failures to be a homomorphism by
\begin{align}
s(g)s(g') = i(\epsilon(g,g'))s(gg').
\label{eq:section-vs-homomorphism}
\end{align}

To look into the structure of $\hat{G}$, we use the bijection $N\times G\to \hat{G}\ ;(n,g)\mapsto i(n)s(g)$.
Recalling that $i(N)$ is a normal subgroup of $\hat{G}$, we can define the action of $h\in \hat{G}$ on $N$ by
\begin{align}
\hat{\varphi}_{h}(n) = i^{-1}(hi(n)h^{-1}),
\label{eq:action-of-Ghat-on-N}
\end{align}
and then the multiplication law of $\hat{G}$ can be calculated as
\begin{align}
i(n)s(g) \cdot i(n')s(g') = i(n \cdot \hat{\varphi}_{s(g)}(n') \cdot \epsilon(g,g'))s(gg').
\label{eq:multi-law-of-grp-ext}
\end{align}
This multiplication law (i.e.\ the group structure of $\hat{G}$) becomes simpler in the following cases:
\begin{itemize}
\item If we can take a section $s$ which is a group homomorphism (i.e.\ the short exact sequence (\ref{eq:def-of-group-extension}) \emph{splits}), then the cocycle becomes trivial $\epsilon(g,g')=1_N$.
In addition, we can define a $G$-action $\sigma$ on $N$ by $\sigma_g:=\hat{\varphi}_{s(g)}$.
As a result, $\hat{G}$ is isomorphic to the semidirect product $N \rtimes_\sigma G$,
and it is also denoted by $N:G$.
In particular, $G$ can be regarded as a subgroup of $\hat{G}$ by $s$.
\begin{itemize}
\item Moreover, if the $G$-action $\sigma$ on $N$ is trivial,
then $\hat{G}$ is isomorphic to the direct product $N \times G$.
\end{itemize}
\item If $N$ is abelian, then for a given $g\in G$, any element $h\in\hat{G}$ such that $\pi(h)=g$ defines the same action $\hat{\varphi}_h$ on $N$, so the $\hat{G}$-action $\hat{\varphi}$ reduces to a $G$-action $\varphi$ on $N$.
\begin{itemize}
\item Moreover, if (and only if) $i(N)$ is in the center of $\hat{G}$, the $G$-action $\varphi$ becomes trivial.
In this case, $\hat{G}$ is called a \emph{central extension}.
\end{itemize}
\end{itemize}
In the following, \underline{we assume $N$ is abelian}.
An abelian group on which a group $G$ acts is called a \emph{$G$-module}.
Since there is the $G$-action $\varphi$ on $N$,
$N$ is a $G$-module.

\vspace{\vertspace}

From the associativity $(s(g)s(g'))s(g'')=s(g)(s(g')s(g''))$, the function $\epsilon$ turns out to be an $N$-valued 2-cocycle:
\begin{align}
\epsilon(g,g')\epsilon(gg',g'') = \varphi_g(\epsilon(g',g''))\epsilon(g,g'g'').
\label{eq:cocycle-condition-general}
\end{align}
If we take another section $s'$, then the cocycle $\epsilon'$ for it differs from $\epsilon$ by only coboundary.
To see it, we define $\zeta:G\to N$ as the difference of $s(g)$ and $s'(g)$:
\begin{align}
s'(g) = i(\zeta(g))s(g).
\end{align}
Then we can calculate $\epsilon'$ based on the definition (\ref{eq:section-vs-homomorphism}) as
\begin{align}
\epsilon'(g,g') = \epsilon(g,g')\varphi_g(\zeta(g'))\zeta(gg')^{-1}\zeta(g),
\label{eq:cohomologous-cocycles}
\end{align}
which shows that $\epsilon'$ and $\epsilon$ differ by the coboundary $d\zeta$.
Conversely, if a cocycle $\epsilon'$ is cohomologous to the cocycle $\epsilon$ of the section $s$ as $\epsilon' = \epsilon d\zeta$,
then $\epsilon'$ is the cocycle of the section $s'$ such that $s'(g)=i(\zeta(g))s(g)$.

\vspace{\vertspace}

Two extensions $1\to N\overset{i}{\to}\hat{G}\overset{\pi}{\to}G\to1$ and $1\to N\overset{i'}{\to}\hat{G}'\overset{\pi'}{\to}G\to1$ are said to be \emph{equivalent} if there exists a homomorphism $\psi:\hat{G}\to\hat{G}'$ such that the diagram
\begin{align}
\xymatrix{
 & & \hat{G} \ar[dd]^-\psi \ar[rd]^-\pi & &\\
1 \ar[r] & N \ar[ru]^-i \ar[rd]_-{i'} & & G \ar[r] & 1\\
 & & \hat{G}' \ar[ru]_-{\pi'} & &
}
\label{eq:group-ext-equiv-diagram}
\end{align}
commutes.
Such $\psi$ is an isomorphism by the short five lemma.\footnote{
For the short five lemma to be applied,
$\psi$ should be a homomorphism, not just a map.
In fact, for any two extensions $\hat{G},\hat{G}'$ of $G$ by $N$, there exists a map $\psi$ such that the diagram (\ref{eq:group-ext-equiv-diagram}) commutes (e.g.\ just $i(n)s(g)\mapsto i'(n)s'(g)$ for sections $s:G\to\hat{G}$ and $s':G\to\hat{G}'$ normalized as in (\ref{eq:normalized-section})), but it is of course not an isomorphism in general.
}
The difference of the cocycle $\epsilon$ for a section $s:G\to\hat{G}$ and the cocycle $\epsilon'$ for a section $s':G\to\hat{G}'$ is again a coboundary,
because $\epsilon$ is also the cocycle for the section $\psi\circ s:G\to\hat{G}'$.
Note that two extensions can be non-equivalent even if $\hat{G}$ and $\hat{G}'$ are isomorphic as groups.\footnote{
For example, for additive groups $G=\Z_2\times\Z_2$ and $N=\Z_2$,
(i) $\hat{G}=\Z_4\times\Z_2$ with $i(1)=(2,0)$ and $\pi((a,b)) = (a,b) \! \mod 2$, and 
(ii) $\hat{G}=\Z_4\times\Z_2$ with $i(1)=(2,0)$ and $\pi((a,b)) = (b,a) \! \mod 2$
are inequivalent extensions.
}
In this sense, the equivalence of group extension classifies not only $\hat{G}$, but also how $N$ and $G$ are incorporated in it.

So far, for a given $G$-module $N$, we have established a map
\begin{align}
\flat : \left\{
\begin{array}{c}
\text{extensions of $G$ by $N$}\\
\text{compatible with the action $G\curvearrowright N$}
\end{array}
\right\} / \text{equivalence} \to H^2(G,N).
\label{eq:equiv-extensions-vs-2nd-cohomology}
\end{align}

\vspace{\vertspace}

We can construct a map $\sharp$ in the inverse direction of (\ref{eq:equiv-extensions-vs-2nd-cohomology}) as follows.
For a given $G$-module $N$ and a 2-cocycle $\epsilon:G\times G\to N$,
we can construct a group extension $\hat{G}_\epsilon$ as a set $N\times G$ with the multiplication\footnote{
Under this multiplication,
the identity element of $\hat{G}_\epsilon$ is $(\epsilon(1_G,1_G)^{-1},1_G)$,
and the inverse of $(n,g)$ is $(\epsilon(1_G,1_G)^{-1}g^{-1}(n^{-1})\epsilon(g^{-1},g)^{-1}, g^{-1})$.
We can check them by using
$\epsilon(1_G,g)=\epsilon(1_G,1_G)$,
$\epsilon(g,1_G)=g(\epsilon(1_G,1_G))$, and
$\epsilon(g,g^{-1})\epsilon(1_G,g)=g(\epsilon(g^{-1},g))\epsilon(g,1_G)$,
which all follow from the cocycle condition (\ref{eq:cocycle-condition-general}).
\label{fn:identity-and-inverse-of-grp-ext}
}
\begin{align}
(n,g) \cdot (n',g') = (n \cdot g(n') \cdot \epsilon(g,g'), gg'),
\label{eq:construct-extension-from-cocycle}
\end{align}
together with group homomorphisms
\begin{align}
& i:N\to\hat{G}_\epsilon\ ;n\mapsto(n\cdot\epsilon(1_G,1_G)^{-1},1_G),\\
& \pi:\hat{G}_\epsilon\to G\ ;(n,g)\mapsto g,
\end{align}
constituting the short exact sequence.

For another 2-cocycle $\epsilon'$, if it differs from $\epsilon$ by the coboundary $d\zeta$ as in (\ref{eq:cohomologous-cocycles}),
then
\begin{align}
\psi: \hat{G}_\epsilon \to \hat{G}_{\epsilon'} \ ; (n,g) \mapsto (n\zeta(g)^{-1},g)
\end{align}
defines an homomorphism to show the equivalence of $\hat{G}_\epsilon$ and $\hat{G}_{\epsilon'}$.
Therefore, we have established the map $\sharp$ in the inverse direction of (\ref{eq:equiv-extensions-vs-2nd-cohomology}).

\vspace{\vertspace}

In fact, $\flat$ and $\sharp$ are the inverse of each other.
To see it, we first note that any 2-cocycle $\epsilon$ satisfies
\begin{align}
\epsilon(1_G,g) = \epsilon(1_G,1_G) \quad \text{for any $g\in G$},
\label{eq:cocycle-for-1-general}
\end{align}
which follow from the cocycle condition (\ref{eq:cocycle-condition-general}).
We also have
\begin{align}
s(1_G) = i(\epsilon(1_G,1_G)),
\label{eq:cocycle-for-section-1-general}
\end{align}
from (\ref{eq:section-vs-homomorphism}),
for any section $s:G\to\hat{G}$ of an extension and the cocycle $\epsilon$ for it.

To see $\sharp\circ\flat$ is an identity map,
for a given extension $\hat{G}$, take a section $s:G\to\hat{G}$ 
and construct the extension $\hat{G}_\epsilon$ from the cocycle $\epsilon$ for $s$.
Then it is equivalent to the original extension $\hat{G}$ by
\begin{align}
\psi : \hat{G}_\epsilon \to \hat{G} \ ; (n,g) \mapsto i(n)s(g).
\end{align}
The most nontrivial part is $\psi\circ\text{($i$ for $\hat{G}_\epsilon$)} = \text{($i$ for $\hat{G}$)}$,
which follows from
\begin{align}
\psi((n\epsilon(1_G,1_G)^{-1},1_G)) = i(n\epsilon(1_G,1_G)^{-1})s(1_G) \overset{\mathrm{(\ref{eq:cocycle-for-section-1-general})}}{=} i(n).
\end{align}

To see $\flat\circ\sharp$ is an identity map, starting from a given cocycle $\epsilon$, construct the extension $\hat{G}_\epsilon$ from it.
Then the section $s:G\to\hat{G}_{\epsilon}\ ;g\mapsto(1_N,g)$ gives back the cocycle $\epsilon$ because
\begin{align}
s(g)s(g') & = (\epsilon(g,g'),gg') = (\epsilon(g,g')\epsilon(1_G,gg')^{-1},1_G) \cdot (1_N,gg')\\
& \overset{\mathrm{(\ref{eq:cocycle-for-1-general})}}{=} (\epsilon(g,g')\epsilon(1_G,1_G)^{-1},1_G) \cdot (1_N,gg') = i(\epsilon(g,g'))s(gg').
\end{align}

To summarize above discussions, we finally obtained the following theorem.
\begin{thm}
\label{thm:grp-ext-and-grp-coh}
For a group $G$ and a $G$-module $N$, there exists a one-to-one correspondence
\begin{align}
\left\{
\begin{array}{c}
\text{extensions of $G$ by $N$}\\
\text{compatible with the action $G\curvearrowright N$}
\end{array}
\right\} / \text{equivalence} \overset{\flat}{\underset{\sharp}{\rightleftarrows}} H^2(G,N).
\end{align}
\end{thm}

\vspace{\vertspace}

Lastly, we mention the normalization of sections and cocycles.
If we take a section $s:G\to\hat{G}$ satisfying the normalization condition
\begin{align}
s(1_G) = 1_{\hat{G}},
\label{eq:normalized-section}
\end{align}
then the cocycle $\epsilon$ for it satisfies the normalization condition
\begin{align}
\epsilon(1_G,1_G)=1_N.
\label{eq:normalized-cocycle}
\end{align}
Therefore, by Theorem \ref{thm:grp-ext-and-grp-coh}, any cohomology class in $H^2(G,N)$ has at least one normalized cocycle.
In fact, we can construct it explicitly;
for any cocycle $\epsilon'$, if we define $\zeta(g):=\epsilon'(g,g)^{-1}$, then the modified cocycle $\epsilon:=\epsilon'\cdot d\zeta$ satisfies the normalization (\ref{eq:normalized-cocycle}).
Hence, in many cases, we can restrict our attention to the normalized sections and cocycles, without loss of generality.

\subsection{Central Extensions of Free Abelian Groups and Commutator Maps}
\label{subsec:cent-ext-and-comm-map}
Let \underline{$G$ be an abelian group} and $1 \to N \overset{i}{\to} \hat{G} \overset{\pi}{\to} G \to 1$ be a central extension of $G$.
Now, any commutator $[h,h']:=hh'h^{-1}h'^{-1}$ of $h,h'\in \hat{G}$ is in $\ker\pi = i(N)$, and hence in the center of $\hat{G}$.

If we take a section $s:G\to\hat{G}$, we can define a function $c:G\times G\to N$ by
\begin{align}
c(g,g') = i^{-1}([s(g),s(g')]).
\label{eq:def-of-comm-map}
\end{align}
If we take another section $s':G\to\hat{G}$, 
then we have $[s'(g),s'(g')]=[s(g),s(g')]$,
which follows from the fact that $s'(g)s(g)^{-1}$ is in $\ker\pi=i(N)$ and hence in the center of $\hat{G}$.
Therefore, the function $c$ does not depend on the choice of the section.
This $c$ is called the \emph{commutator map} associated to the central extension.

The properties of the commutator, $[h,h]=1$ and $[h,h']=[h',h]^{-1}$, respectively translates to the \emph{alternating} property
\begin{align}
c(g,g) = 1_N,
\label{eq:comm-map-alter}
\end{align}
and the \emph{antisymmetric} property
\begin{align}
c(g,g') = c(g',g)^{-1},
\label{eq:comm-map-antisym}
\end{align}
of the commutator map $c$.

From the fact that any commutator of $\hat{G}$ is in the center of $\hat{G}$, we also have the properties
\begin{align}
& [hh',h''] = [h,h''][h',h''],\\
& [h,h'h''] = [h,h'][h,h''],
\end{align}
for any $h,h',h''\in\hat{G}$.
They translates to the \emph{bilinearity}
\begin{align}
& c(gg',g'') = c(g,g'') c(g',g''), \label{eq:comm-map-linearity-1} \\
& c(g,g'g'') = c(g,g') c(g,g''), \label{eq:comm-map-linearity-2}
\end{align}
of the commutator map $c$.
Under the bilinearity (\ref{eq:comm-map-linearity-1}, \ref{eq:comm-map-linearity-2}), the antisymmetric property (\ref{eq:comm-map-antisym}) follows from the alternating property (\ref{eq:comm-map-alter}).\footnote{
The converse (``antisymmetric $\Rightarrow$ alternating'' under bilinearity) holds if $N$ does not have order-2 elements.
}

In terms of the cocycle $\epsilon$ of the section $s$, defined in (\ref{eq:section-vs-homomorphism}),
the commutator map is
\begin{align}
c(g,g') = \epsilon(g,g') \epsilon(g',g)^{-1}.
\label{eq:comm-map-vs-cocycle}
\end{align}
This (\ref{eq:comm-map-vs-cocycle}) holds even if we replace $\epsilon$ with another cocycle $\epsilon'$ cohomologous to $\epsilon$,
because such $\epsilon'$ is, as we have seen around (\ref{eq:cohomologous-cocycles}), just a cocycle for another section.
(Or, we can explicitly calculate $\epsilon'(g,g') \epsilon'(g',g)^{-1}=\epsilon(g,g') \epsilon(g',g)^{-1}$.)
Therefore, (\ref{eq:comm-map-vs-cocycle}) establishes a map
\begin{align}
\flat_c : H^2(G,N) \to \left\{\begin{array}{c}
\text{alternating bilinear maps}\\
\text{$c:G\times G\to N$}
\end{array}\right\},
\end{align}
for abelian groups $G$ and $N$, where $G$ acts on $N$ trivially.

\vspace{\vertspace}

To obtain the inverse $\sharp_c$ of this map $\flat_c$, 
we further assume that \underline{{$G$ be a free abelian group}} of finite rank $r$, and take a $\Z$-basis $e_0,\ldots,e_{r-1}$ of $G$.
For a given alternating $\Z$-bilinear map $c:G\times G\to N$, we define a $\Z$-bilinear map $\epsilon_c:G\times G\to N$ by the bilinear extension of
\begin{align}
\epsilon_c(e_i,e_j) = \left\{\begin{array}{ll}
c(e_i,e_j) & (i>j),\\
1_N & (i\leq j).
\end{array}\right.
\label{eq:cocycle-from-comm-map}
\end{align}
Since the action of $G$ on $N$ is trivial in the current situation,
a $\Z$-bilinear map $G \times G\to N$ is automatically a 2-cocycle.
Now, this cocycle $\epsilon_c$ satisfies (\ref{eq:comm-map-vs-cocycle}) for the given $c$.
Therefore, $c\mapsto\epsilon_c$ establishes a map
\begin{align}
\sharp_c : \left\{\begin{array}{c}
\text{alternating $\Z$-bilinear maps}\\
c:G\times G\to N
\end{array}\right\} \to H^2(G,N),
\end{align}
for a free abelian group $G$ of finite rank and an abelian group $N$,
and $\flat_c\circ\sharp_c$ is an identity map.

To see that $\sharp_c\circ\flat_c$ is also an identity map,
we start from a given cocycle $\epsilon$,
consider the commutator map $c_\epsilon$ for $\epsilon$ as in (\ref{eq:comm-map-vs-cocycle}),
and show that the cocycle $\epsilon_{c_\epsilon}$ constructed from this $c_\epsilon$ as in (\ref{eq:cocycle-from-comm-map}) is in the same cohomology class as the original cocycle $\epsilon$ belongs to.
Let $\hat{G}_\epsilon$ be the central extension constructed from $\epsilon$ as in (\ref{eq:construct-extension-from-cocycle}).
Recall that the cocycle for the section $s:G\to\hat{G}_\epsilon\ ;g\mapsto(1_N,g)$ is $\epsilon$ itself,
and hence the commutator map associated to $\hat{G}_\epsilon$ is of course $c_\epsilon$.
Take another section $s':G\to\hat{G}_\epsilon$ as follows:
set $s'(e_0), \ldots, s'(e_{r-1})$ to any element satisfying $\pi(s'(e_i))=e_i$, say $s'(e_i)=s(e_i)$,
and for a general element $e_0^{k^0} \cdots e_{r-1}^{k^{r-1}}\in G$, define
\begin{align}
s'(e_0^{k^0} \cdots e_{r-1}^{k^{r-1}}) = s'(e_0)^{k^0} \cdots s'(e_{r-1})^{k^{r-1}}.
\label{eq:section-to-prove-1to1-with-comm-maps}
\end{align}
This $s':G\to\hat{G}_\epsilon$ is well-defined because $G$ is a free abelian group, that is, torsion-free.\footnote{
If $G$ is a finitely generated abelian group,
where torsion is allowed,
then we cannot take a well-defined section $s'$ in the form of (\ref{eq:section-to-prove-1to1-with-comm-maps}) in general,
and in particular,
we cannot transform (\ref{eq:eqs-to-prove-1to1-with-comm-maps-4}) into (\ref{eq:eqs-to-prove-1to1-with-comm-maps-5}).
In fact, there are inequivalent extensions of $G$ with the same commutator map in this case.
For example, central extensions of an elementary abelian 2-group $G=(\Z_2)^n$ by $\Z_2$ up to equivalence
are in one-to-one correspondence with quadratic forms $G \to \Z_2$ \cite[Prop.\ 5.3.3]{MR0996026}.
The bilinear form associated to the quadratic form is the commutator map,
and there are multiple quadratic forms with the same associated bilinear form;
its degree of freedom is adding linear forms $G \to \Z_2$ to the quadratic form \cite[Remark 5.3.2]{MR0996026}.
If we further assume that $N$ is a divisible group,
then we can take a well-defined section $s'$ in the form of (\ref{eq:section-to-prove-1to1-with-comm-maps})
even if $G$ has a torsion part,
and therefore we can still establish the one-to-one correspondence with commutator maps.
For example, \cite[Prop.\ 2.6]{MR1776075} deals with the case of $N=k^\times$,
which is the multiplicative group of a field $k$.
}
Note that this $s'$ is not necessarily linear, or a homomorphism, because $s'(e_i)$'s are not necessarily commutative.
Then the cocycle $\epsilon'$ for this section $s'$ coincides with $\epsilon_{c_\epsilon}$, because
\begin{align}
& s'(e_0^{k^0} \cdots e_{r-1}^{k^{r-1}}) \cdot s'(e_0^{l^0} \cdots e_{r-1}^{l^{r-1}}) \\
& = s'(e_0)^{k^0} \cdots s'(e_{r-1})^{k^{r-1}} \cdot s'(e_0)^{l^0} \cdots s'(e_{r-1})^{l^{r-1}}\\
& = s'(e_0)^{k^0} \cdots s'(e_{r-2})^{k^{r-2}} \cdot s'(e_0)^{l^0} \cdots s'(e_{r-1})^{k^{r-1}+l^{r-1}} i(\prod_{r-1>j}c_\epsilon(e_{r-1},e_j)^{k^{r-1}l^j})\\
& = \cdots = s'(e_0)^{k^0+l^0} \cdots s'(e_{r-1})^{k^{r-1}+l^{r-1}} i(\prod_{i>j}c_\epsilon(e_i,e_j)^{k^il^j}) \label{eq:eqs-to-prove-1to1-with-comm-maps-4}\\
& = s'(e_0^{k^0+l^0} \cdots e_{r-1}^{k^{r-1}+l^{r-1}}) i(\epsilon_{c_\epsilon}(e_0^{k^0} \cdots e_{r-1}^{k^{r-1}}, e_0^{l^0} \cdots e_{r-1}^{l^{r-1}})), \label{eq:eqs-to-prove-1to1-with-comm-maps-5}
\end{align}
where in the second equation, we used the fact that the commutator map $c_\epsilon$ does not depend on the choice of the section $s$ or $s'$,
as we saw above.
Therefore, $\epsilon$ and $\epsilon_{c_\epsilon}$ are the cocycles of the different sections $s$ and $s'$ of the same extension $\hat{G}_\epsilon$,
and hence they differ only by the coboundary,
as we have seen around (\ref{eq:cohomologous-cocycles}).

\vspace{\vertspace}

We finally obtained the following theorem.
\begin{thm}
\label{thm:cent-ext-and-comm-map}
For a free abelian group $G$ of finite rank and an abelian group $N$, there exist one-to-one correspondences
\begin{align}
\left\{ \text{central extensions of $G$ by $N$} \right\} / \text{equivalence} & \overset{\flat}{\underset{\sharp}{\rightleftarrows}} H^2(G,N)\\
& \overset{\flat_c}{\underset{\sharp_c}{\rightleftarrows}} \left\{\begin{array}{c}
\text{alternating $\Z$-bilinear maps}\\
c:G\times G\to N
\end{array}\right\},
\end{align}
where in $H^2(G,N)$, $N$ is regarded as a $G$-module by the trivial $G$-action.
\end{thm}

\subsection{A Theorem on Automorphism Group of Central Extension}
\label{subsec:thm-on-aut-grp-of-central-ext}
This Section \ref{subsec:thm-on-aut-grp-of-central-ext} is a review of \cite[Prop.\ 5.4.1]{MR0996026}.

Let $1 \to N \overset{i}{\to} \hat{G} \overset{\pi}{\to} G \to 1$ be a central extension of a free abelian group $G$ of finite rank by an abelian group $N$,
and $c:G\times G\to N$ be the associated commutator map.

We first define some necessary objects.
We define a subgroup $\Aut(\hat{G},N)$ of the automorphism group $\Aut(\hat{G})$ as
\begin{align}
\Aut(\hat{G},N) := \{\chi\in\Aut(\hat{G}) \mid \chi(i(n)) = i(n) \ \text{for any $n\in N$}\}.
\end{align}
Since $N$ is abelian, $\Hom(G,N)$ has a group structure under the multiplication of functions: for $\eta,\eta'\in\Hom(G,N)$, $(\eta\cdot\eta')(g)=\eta(g)\eta'(g)$.
Now we can define a group homomorphism $I:\Hom(G,N)\to\Aut(\hat{G},N)$ which maps $\eta\in\Hom(G,N)$ to
\begin{align}
I(\eta): \hat{G} & \to \hat{G}\\
h & \mapsto i(\eta(\pi(h)))h.
\end{align}
In fact, $I(\eta)$ is an element of $\Aut(\hat{G},N)$, and $I$ satisfies $I(\eta\cdot\eta')=I(\eta)\circ I(\eta')$.

We also define a subgroup $\Aut(G,c)$ of $\Aut(G)$ as
\begin{align}
\Aut(G,c) := \{\phi\in\Aut(G) \mid c(\phi(g),\phi(g')) = c(g,g') \ \text{for any $g,g'\in G$}\},
\end{align}
and a group homomorphism $\Pi:\Aut(\hat{G},N)\to\Aut(G,c)$ which maps $\chi\in\Aut(\hat{G},N)$ to
\begin{align}
\Pi(\chi): G & \to G\\
g & \mapsto \pi(\chi(\hat{g})),
\end{align}
where $\hat{g}$ is an element of $\hat{G}$ such that $\pi(\hat{g})=g$.
Here, $\Pi(\chi)(g)$ does not depend on the choice of $\hat{g}$, because such another $\hat{g}'$ is an element of $\hat{g}i(N)$, and hence
\begin{align}
\pi(\chi(\hat{g}')) \in \pi(\chi(\hat{g}i(N))) = \pi(\chi(\hat{g})i(N)) = \{\pi(\chi(\hat{g}))\}.
\end{align}
To see that $\Pi(\chi)$ is in fact an element of $\Aut(G,c)$,
it suffices to calculate
\begin{align}
& c(\Pi(\chi)(g),\Pi(\chi)(g')) = i^{-1}([\widehat{\pi(\chi(\hat{g}))},\widehat{\pi(\chi(\widehat{g'}))}])\\
& \hspace{30pt} = i^{-1}([\chi(\hat{g}),\chi(\widehat{g'})]) = i^{-1}(\chi([\hat{g},\widehat{g'}])) = i^{-1}([\hat{g},\widehat{g'}]) = c(g,g'),
\end{align}
where we used the definition (\ref{eq:def-of-comm-map}) of the commutator map $c$ in the form of $c(-,-)=i^{-1}([\hat{-},\hat{-}])$; recall that it does not depend on the choice of $\hat{-}$.
Lastly, it is easy to see that $\Pi$ is in fact a group homomorphism: $\Pi(\chi\circ\chi')=\Pi(\chi)\circ\Pi(\chi')$.

\vspace{\vertspace}

Now, here is the theorem.
\begin{thm}[{\cite[Prop.\ 5.4.1]{MR0996026}}]
\label{thm:aut-grp-of-central-ext}
The following sequence is exact.
\begin{align}
1 \to \Hom(G,N) \overset{I}{\to} \Aut(\hat{G},N) \overset{\Pi}{\to} \Aut(G,c) \to 1.
\end{align}
\end{thm}

\begin{proof}
$\ker I = 1_N$ (constant function in $\Hom(G,N)$) and $\im I\subset\ker\Pi$ are easy.

$\ker \Pi\subset\im I$ can be checked as follows.
Let $\chi\in\Aut(\hat{G},N)$ satisfy $\Pi(\chi)=\id_G$.
Since $\pi(\chi(\hat{g}))=g$ for any $g\in G$,
we have $\pi(\chi(h))=\pi(h)$ for any $h\in\hat{G}$,
and hence the difference $\chi(h)h^{-1}$ of $\chi(h)$ and $h$ is in $i(N)$.
Furthermore, it only depends on $\pi(h)$, because if we take another $h'\in \hat{G}$ such that $\pi(h')=\pi(h)$, then $h'\in hi(N)$ and hence $\chi(h')(h')^{-1} = \chi(h)h^{-1}$.
Therefore, there exists a map $\eta:G\to N$ such that $\chi(h)h^{-1}=i(\eta(\pi(h)))$.
The linearity of $\eta$ follows from, for any $h,h'\in\hat{G}$,
\begin{align}
\chi(hh')(hh')^{-1} = \chi(h)\chi(h')(h')^{-1}h^{-1} = \chi(h)h^{-1} \chi(h')(h')^{-1},
\end{align}
where the last equation follows from $\chi(h')(h')^{-1} \in i(N) \subset \mathrm{Center}(\hat{G})$.
As a result, $\eta\in\Hom(G,N)$ and $I(\eta) = \chi$.

The surjectivity of $\Pi$ can be shown as follows.
For $\phi\in\Aut(G,c)$, consider a new central extension
\begin{align}
1 \to N \overset{i}{\to} \hat{G} \overset{\phi\circ\pi}{\to} G \to 1.
\label{eq:new-central-extension-for-phi}
\end{align}
If $s:G\to\hat{G}$ is a section of the original central extension, i.e.\ $\pi\circ s = \id_G$, then $s\circ \phi^{-1}$ is a section of (\ref{eq:new-central-extension-for-phi}).
Then, the commutator map $c_\phi$ associated to the new central extension coincides with the original commutator map $c$, because
\begin{align}
c_\phi(g,g') = c(\phi^{-1}(g),\phi^{-1}(g)) = c(g,g').
\end{align}
Therefore, by Theorem \ref{thm:cent-ext-and-comm-map}, there exists $\psi\in\Aut(\hat{G})$ such that the diagram
\begin{align}
\xymatrix{
 & & \hat{G} \ar[dd]^-\psi \ar[rd]^-{\phi\circ\pi} & &\\
1 \ar[r] & N \ar[ru]^i \ar[rd]_i & & G \ar[r] & 1\\
 & & \hat{G} \ar[ru]_-\pi & &
}
\end{align}
commutes.
It is obvious that $\psi\in\Aut(\hat{G},N)$ from the diagram,
and $\Pi(\psi)=\phi$ because
\begin{align}
\Pi(\psi)(g) = \pi(\psi(s(g))) = \phi(\pi(s(g))) = \phi(g). 
\end{align}
\end{proof}

\section{\texorpdfstring{$\cN=2$}{N=2} SCA in Odd Leech Lattice CFT}
\label{sec:N=2-SCA-in-odd-Leech}
In this Appendix \ref{sec:N=2-SCA-in-odd-Leech},
we discuss how to realize an $\cN=2$ superconformal algebra (SCA) of central charge 24 in the odd Leech lattice CFT,
such that the partition function of the corresponding R sector coincides with the $\cN=2$ extremal elliptic genus.

The $\cN=2$ SCA of central charge $c$ consists of the operators $T(z), J(z), G^+(z), G^-(z)$ satisfying the following OPEs: 
\begin{align}
& T(z_1)T(z_2) \sim \frac{c/2}{(z_1-z_2)^4} + \frac{2}{(z_1-z_2)^2}T(z_2) + \frac{1}{z_1-z_2}\partial T(z_2), \label{N=2-SCA-TT-OPE}\\
& T(z_1)J(z_2) \sim \frac{1}{(z_1-z_2)^2} J(z_2) + \frac{1}{z_1-z_2}\partial J(z_2),\\
& T(z_1)G^\pm(z_2) \sim \frac{3/2}{(z_1-z_2)^2}G^\pm(z_2) + \frac{1}{z_1-z_2}\partial G^\pm(z_2),\\
& J(z_1)J(z_2) \sim \frac{c/3}{(z_1-z_2)^2},\\
& J(z_1)G^\pm(z_2) \sim \pm\frac{1}{z_1-z_2}G^\pm(z_2), \label{N=2-SCA-JG-OPE}\\
& G^+(z_1)G^-(z_2) \sim \frac{2c/3}{(z_1-z_2)^3} + \frac{2}{(z_1-z_2)^2}J(z_2) + \frac{1}{z_1-z_2}(2T(z_2)+\partial J(z_2)), \label{N=2-SCA-G+G--OPE}\\
& G^\pm(z_1)G^\pm(z_2) \sim 0. \label{N=2-SCA-G+G+-OPE}
\end{align}
These operators, $T(z), J(z),$ and $G^\pm(z)$ are called the \emph{energy-momentum tensor}, the \emph{$U(1)$ current}, and the \emph{supercurrents}, respectively.

\vspace{\vertspace}

Let us first review the candidate of $c=24$ $\cN=2$ SCA proposed by \cite{Benjamin:2015ria}.
The theory of \cite{Benjamin:2015ria} uses the reflection $\Z_2$ symmetry $X(z)\to-X(z)$ of the chiral bosonic lattice CFT constructed from the $(A_1)^{24}$ lattice.
The twisted states with respect to the $\Z_2$ symmetry
are built up on the twisted ground states,
which constitute the $2^{\frac{n}{2}=12}$-dimensional irreducible representation of a certain gamma matrix algebra,
and these twisted ground states are of weight $\frac{n}{16}=\frac{3}{2}$ and get a sign $(-1)^{\frac{n}{8}=3}$ under the $\Z_2$ symmetry \cite[\S5.3]{Dolan:1994st}.
Currents of weight $\leq\frac{3}{2}$ in this $(A_1)^{24}$ theory are
\begin{align}
\begin{array}{c|cc}
\text{reflection $\Z_2$} & \text{untwisted} & \text{twisted} \\\hline
\text{even} & V_{\sqrt{2}\bbe_i}(z)+V_{-\sqrt{2}\bbe_i}(z) & \\
\text{odd} & V_{\sqrt{2}\bbe_i}(z)-V_{-\sqrt{2}\bbe_i}(z),\ \partial X^i(z) & \text{the twisted ground states}
\end{array} \quad ,
\label{sectors-of-A1^24-under-reflection-Z2}
\end{align}
where $\{\bbe_i\}_{i=0,\ldots,23}$ is the standard basis.
The authors of \cite{Benjamin:2015ria} fermionized this $(A_1)^{24}$ theory by the reflection $\Z_2$ symmetry;
the NS sector of the resulting theory is the even-untwisted and odd-twisted sectors above,
and the R sector is the odd-untwisted and even-twisted sectors.
They further proposed a $c=24$ $\cN=2$ SCA in the NS sector
such that the partition function of the R sector
coincides with the extremal $\cN=2$ elliptic genus of central charge 24.

We can translate their theory in terms of the odd Leech lattice CFT,
by considering the shift $\Z_2$ instead of the reflection $\Z_2$.
With respect to the shift $\Z_2$ symmetry (\ref{eq:shift-Z2-of-vertex-op}), 
the twisted sector of the $(A_1)^{24}$ theory has the lattice description as in Section \ref{subsec:lattice-CFT}.
This time, currents of weight $\leq\frac{3}{2}$ in each sector of the $(A_1)^{24}$ theory are
\begin{align}
\begin{array}{c|cc}
\text{shift $\Z_2$} & \text{untwisted} & \text{twisted} \\\hline
\text{even} & \partial X^i(z) & \\
\text{odd} & V_{\sqrt{2}\bbe_i}(z), V_{-\sqrt{2}\bbe_i}(z) & V_{(-1)^w/2\sqrt{2}}(z), w\in G_{24} \text{ in (\ref{eq:wt-3/2-in-odd-Leech})}
\end{array} \quad .
\label{sectors-of-A1^24-under-shift-Z2}
\end{align}
The fermionized theories by reflection $\Z_2$ and shift $\Z_2$ are equivalent (see footnote \ref{fn:equiv-of-reflection-and-shift-sym}),
and the NS sector of the fermionized theory by this shift $\Z_2$ is the odd Leech lattice CFT.

By comparing (\ref{sectors-of-A1^24-under-reflection-Z2}) and (\ref{sectors-of-A1^24-under-shift-Z2}),
we can translate the candidate of $c=24$ $\cN=2$ SCA proposed by \cite{Benjamin:2015ria}
in terms of the odd Leech lattice CFT as
\begin{align}
&T(z) = -\frac{1}{2}\sum_{i=0}^{23}:(\partial X^i(z))^2:, \label{eq:BDFK-T}\\
&J(z) = 2\sqrt{-2}\partial X^0(z), \label{eq:BDFK-J}\\
&G^+(z) = \frac{1}{8\sqrt{2}} \sum_{w\in G_{24} \text{ such that } w_0=0} V_{\frac{1}{2\sqrt{2}}(-1)^w}(z), \label{eq:BDFK-G+}\\
&G^-(z) = \frac{1}{8\sqrt{2}} \sum_{w\in G_{24} \text{ such that } w_0=1} V_{\frac{1}{2\sqrt{2}}(-1)^w}(z), \label{eq:BDFK-G-}
\end{align}
where $w=(w_0,\ldots,w_{23})$ and $(-1)^w = ((-1)^{w_0},\ldots,(-1)^{w_{23}})$ as in (\ref{eq:wt-3/2-in-odd-Leech}).

These operators in fact satisfy the $TT$, $TJ$, $JJ$, $TG^\pm$, and $JG^\pm$ OPEs (\ref{N=2-SCA-TT-OPE})--(\ref{N=2-SCA-JG-OPE}).
However, if we calculate the $G^\pm G^\pm$ OPE and $G^+ G^-$ OPE for them,
while taking into account the cocycle factor given by (\ref{eq:varepsilon}) and the choice of the basis (\ref{eq:odd-Leech-basis}),
then they fail to satisfy the proper OPEs (\ref{N=2-SCA-G+G--OPE}) and (\ref{N=2-SCA-G+G+-OPE});
their OPEs contain unnecessary weight-2 currents $V_{\frac{1}{\sqrt{2}}w^\pm}(z_2)$, where $w^\pm$ denotes an octad whose non-zero entries are $1$ or $-1$, in the $(z_1-z_2)^{-1}$ term.
See \texttt{check\_supercurrent\_OPE.ipynb} for the explicit calculation by computer.
In fact,
the discussion of \cite[Appendix A.2]{Benjamin:2015ria}
on how these unnecessary weight-2 currents decouple from the OPEs
is based on the assumption that the subgroup $M_{24}$ of the automorphism group of the lattice directly lifts to a subgroup of the automorphism group of the lattice CFT,
but in reality, it is disturbed by the cocycle factors as in Theorem \ref{thm:main-oddLeech}.

Their proposal is still valuable in the sense that the choice of the $U(1)$ current $J(z)$ in (\ref{eq:BDFK-J}) is enough for the partition function of the R sector to reproduce the $\cN=2$ extremal elliptic genus of central charge 24.
The $U(1)$ current (\ref{eq:BDFK-J}) determines the first few terms of the partition function as
\begin{align}
\Tr_\mathrm{R} (-1)^F y^{J_0} q^{L_0-\frac{c}{24}} = y^{-4} + 46 +y^4 + \cdots,
\label{eq:lower-order-terms-of-partition-func}
\end{align}
and such a weak Jacobi form\footnote{
It is known \cite{Kawai:1993jk} that if the $U(1)$ charges of the states in the NS sector of an $\cN=2$ SCFT (in a non-chiral case, say $\cN=(2,2)$ SCFT) of central charge $c$ are all integers,
then the partition function (elliptic genus) $\Tr_\text{$\mathrm{R}$-$\tilde{\mathrm{R}}$}(-1)^{F+\tilde{F}}y^{J_0}q^{L_0-\frac{c}{24}}\bar{q}^{\tilde{L}_0-\frac{c}{24}}$ is a weak Jacobi form of weight 0 and index $\frac{c}{6}$.
The odd Leech lattice CFT in the main text with the $U(1)$ current (\ref{eq:BDFK-J}) satisfies this condition.
}
of weight 0 and index $\frac{c}{6}=4$ is uniquely determined to be the $\cN=2$ extremal elliptic genus of central charge 24.
So we only have to retake the supercurrents $G^\pm(z)$ so that these operators satisfy the correct OPEs (\ref{N=2-SCA-TT-OPE})--(\ref{N=2-SCA-G+G+-OPE}) for the $\cN=2$ SCA.

\vspace{\vertspace}

\begin{rem*}
\textit{(A long remark; ends at (\ref{eq:2nd-sgl-term-of-G16+-G8--OPE}).)}
One of the advantageous points of the candidate (\ref{eq:BDFK-G+}, \ref{eq:BDFK-G-}) of the supercurrents
was that they are invariant under the $M_{23}$ symmetry fixing the first entry of $G_{24}$,
if we forget the effect of cocycle factors of the odd Leech lattice CFT.
Unfortunately, however, we cannot find supercurrents while retaining this property.

To state it more precisely, we introduce some notations.
Let $G_{24}^+$ and $G_{24}^-$ denote the set of all codewords of $G_{24}$ with the first entry 0 and 1, respectively.
Let $G_{24}^{(x)}$ denote the set of all codewords of $G_{24}$ with Hamming weight $x$,
and also decompose it into $G_{24}^{(x),\pm}$ according to the first entry.
Note that $G_{24}^{(0)}=G_{24}^{(0),+}=\{\underline{0}\}$
and $G_{24}^{(24)}=G_{24}^{(24),-}=\{\underline{1}\}$.
The orbit decomposition of $G_{24}^+$ and $G_{24}^-$ under the action of $M_{23}$ is $\{\underline{0}\} \sqcup G_{24}^{(8),+} \sqcup G_{24}^{(12),+} \sqcup G_{24}^{(16),+}$ and $G_{24}^{(8),-} \sqcup G_{24}^{(12),-} \sqcup G_{24}^{(16),-} \sqcup \{\underline{1}\}$, respectively.\footnote{
\label{fn:M23-orbit}
This can be seen as follows.
$M_{23}$ acts on $G_{24}^{(x)}$ as a permutation group,
so we can consider a $|G_{24}^{(x)}|$-dimensional representation $R^{(x)}:=\bigoplus_{w\in G_{24}^{(x)}}\C|w\rangle$ of $M_{23}$.
If $O \subset G_{24}^{(x)}$ is an orbit of some element under the action of $M_{23}$,
then $\C(\sum_{w\in O}|w\rangle)$ is the 1-dimensional irreducible subrepresentation of $R^{(x)}$.
Therefore, we have (the number of orbits in $G_{24}^{(x)}$) $\leq$ (the multiplicity of the 1-dimensional irrep.\ in $R^{(x)}$).
By using the character table of $M_{23}$ from the ATLAS book \cite{MR827219},
we can calculate by computer the multiplicity of the 1-dimensional irrep.\ in $R^{(x)}$ as $1,2,2,2,1$ for $x=0,8,12,16,24$, respectively.
This justifies the orbit decompositions in the main text.
In particular, $M_{23}$ acts transitively on each $G_{24}^{(x),\pm}$.
In a similar way, we can also check that $M_{24}$ acts transitively on each $G_{24}^{(x)}$.
}
Therefore, if a current $G^\pm(z)$ of weight $\frac{3}{2}$ and $U(1)$ charge $\pm1$ is invariant under the $M_{23}$ symmetry,
then it should be in the form of the linear combination
\begin{align}
& G^+(z) = c^{(0),+} G^{(0),+}(z) + c^{(8),+} G^{(8),+}(z) + c^{(12),+} G^{(12),+}(z) + c^{(16),+} G^{(16),+}(z), \label{eq:M23-inv-G+}\\
& G^-(z) = c^{(8),-} G^{(8),-}(z) + c^{(12),-} G^{(12),-}(z) + c^{(16),-} G^{(16),-}(z) + c^{(24),-} G^{(24),-}(z), \label{eq:M23-inv-G-}
\end{align}
where we defined
\begin{align}
G^{(x),\pm}(z) := \sum_{w \in G_{24}^{(x),\pm}} V_{\frac{1}{2\sqrt{2}}(-1)^w}(z).
\end{align}
However, we cannot find the coefficients $c^{(x),\pm}$ such that these $G^+(z),G^-(z)$ together with $T(z)$ in (\ref{eq:BDFK-T}) and $J(z)$ in (\ref{eq:BDFK-J})
satisfy the correct $\cN=2$ SCA OPEs (\ref{N=2-SCA-TT-OPE})--(\ref{N=2-SCA-G+G+-OPE}).

To see it, we first use the $G^+G^+$ OPE.
In the calculation of the $G^+G^+$ OPE for (\ref{eq:M23-inv-G+}),
the only sources of singular terms are
\begin{align}
V_{\frac{1}{2\sqrt{2}}(-1)^w}(z_1)V_{\frac{1}{2\sqrt{2}}(-1)^{w'}}(z_2) \sim \varepsilon(\tfrac{1}{2\sqrt{2}}(-1)^w,\tfrac{1}{2\sqrt{2}}(-1)^{w'})\frac{1}{z_1-z_2}V_{\frac{1}{2\sqrt{2}}((-1)^w+(-1)^{w'})}(z_2),
\end{align}
for $w,w'\in G_{24}^+$ such that $(-1)^w+(-1)^{w'}$ has precisely eight non-zero entries $\pm2$, and hence $V_{\frac{1}{2\sqrt{2}}((-1)^w+(-1)^{w'})}(z_2)$ is of weight 2.
By focusing on such terms, while taking into account the cocycle factors,
we obtain the conditions on $c^{(x),+}$ for $G^+(z)$ in (\ref{eq:M23-inv-G+}) to satisfy the correct OPE $G^+G^+\sim0$ as follows.
\begin{itemize}
\setlength{\itemsep}{-3pt}
\leftskip -15pt
\item $c^{(0),+} c^{(16),+} - (c^{(8),+})^2 = 0$ from the terms of $w,w'\in G_{24}^+$ such that\\
\hspace*{60pt} $(-1)^w+(-1)^{w'}=(2, 0, 2, 0, 0, 0, 0, 0, 0, 0, 0, 2, 0, 0, 2, 0, 0, 0, 0, 0, 2, 2, 2, 2)$.
\item $c^{(0),+} c^{(16),+} + 15 (c^{(8),+})^2 = 0$ from the terms of $w,w'\in G_{24}^+$ such that\\
\hspace*{60pt} $(-1)^w+(-1)^{w'}=(2, 0, 0, 0, 2, 2, 0, 0, 0, 0, 0, 0, 0, 0, 0, 0, 0, 2, 0, 2, 0, 2, 2, 2)$.
\item $(c^{(12),+})^2 - c^{(8),+} c^{(16),+} = 0$ from the terms of $w,w'\in G_{24}^+$ such that\\
\hspace*{60pt} $(-1)^w+(-1)^{w'}=(2, 0, 2, 0, 0, -2, 0, -2, 0, -2, 0, 0, -2, 0, 0, 0, 0, 0, 0, 0, 0, 0, 2, 2)$.
\end{itemize}
See \texttt{check\_supercurrent\_OPE.ipynb} again for the computation of these conditions.
The solution to them is $c^{(8),+}=c^{(12),+}=0$ and $c^{(0),+}c^{(16),+}=0$.

As a result, if $G^+(z)$ in (\ref{eq:M23-inv-G+}) satisfies the OPE $G^+G^+\sim0$,
then it must be $G^{(0),+}(z)$ or $G^{(16),+}$ up to scalar multiplication.
If we assume $G^+(z) \propto G^{(0),+}(z)$,
then the OPE $G^+G^+\sim0$ is achieved,
and $G^-(z)$ in (\ref{eq:M23-inv-G-}) must contain $G^{(24),-}(z)$
for their $G^+G^-$ OPE to have a $(z_1-z_2)^{-3}$ term.
However, its $(z_1-z_2)^{-2}$ term is proportional to
\begin{align}
\sum_{w\in G_{24}^{(0),+}} (-1)^w \cdot \partial X(z) = \sum_{i=0}^{23} \partial X^i(z),
\end{align}
where $(-1)^w\cdot X(z):=\sum_{i=0}^{23}(-1)^{w_i}\partial X^i(z)$,
and this is not $J(z)$ in (\ref{eq:BDFK-J}).
So the $G^+G^-$ OPE fails.
Let us assume $G^+(z) \propto G^{(16),+}(z)$, then.
The OPE $G^+G^+\sim0$ is achieved,
because it is easy to see that $(-1)^w+(-1)^{w'}$ cannot have precisely eight non-zero entries when $w,w'\in G^{(16),+}$.
Now, $G^-(z)$ in (\ref{eq:M23-inv-G-}) must contain $G^{(8),-}(z)$,
for their $G^+G^-$ OPE to have a $(z_1-z_2)^{-3}$ term.
However, its $(z_1-z_2)^{-2}$ term is proportional to
\begin{align}
\sum_{w\in G_{24}^{(16),+}} (-1)^w \cdot \partial X(z) = 253 \, \partial X^0(z) - 99\sum_{i=1}^{23} \partial X^i(z),
\label{eq:2nd-sgl-term-of-G16+-G8--OPE}
\end{align}
so the $G^+G^-$ OPE fails again.
This concludes that we have to give up the $M_{23}$ symmetry to retake the supercurrents so that they satisfy the correct $\cN=2$ SCA OPE.
\textit{(Remark ends.)}
\end{rem*}

\vspace{\vertspace}

To retake the supercurrents,
we first point out that we can construct a $c=1$ $\cN=2$ SCA from a codeword $w=(w_0,\ldots,w_{23})$ of the binary Golay code $G_{24}$ as
\begin{align}
&T^{w}(z) = -\frac{1}{48}:\bigl((-1)^w\cdot\partial X(z)\bigr)^2:,\\
&J^{w}(z) = \frac{\sqrt{-1}}{6\sqrt{2}}(-1)^w\cdot\partial X(z),\\
&G^{+,w}(z) = \sqrt{\frac{2}{3}} V_{\frac{1}{2\sqrt{2}}(-1)^w}(z),\\
&G^{-,w}(z) = \sqrt{\frac{2}{3}} V_{-\frac{1}{2\sqrt{2}}(-1)^w}(z),
\end{align}
where $(-1)^w\cdot X(z):=\sum_{i=0}^{23}(-1)^{w_i}\partial X^i(z)$.
This is just a well-known construction of a $c=1$ $\cN=2$ SCA from one free boson \cite{Waterson:1986ru},
if we change the variables as $Y(z):=\frac{1}{2\sqrt{6}}(-1)^w\cdot X(z)$.
The cocycle factor does not cause a problem here because
\begin{align}
\varepsilon(\frac{1}{2\sqrt{2}}(-1)^w, -\frac{1}{2\sqrt{2}}(-1)^w) = 1,
\end{align}
which follows from
(i) $\varepsilon(k,-k) = (-1)^{k\ast(-k)} = (-1)^{k\ast k} = \varepsilon(k,k)$,
(ii) Lemma \ref{lemma:property-of-varep-and-zeta} (3) $\varepsilon(k,k)=\varepsilon(k',k')$ for $|k|^2=|k'|^2$,
(iii) $|\frac{1}{2\sqrt{2}}(-1)^w|^2=|e_0|^2=3$ where $e_0=\frac{1}{2\sqrt{2}}\underline{1}$ in the basis (\ref{eq:odd-Leech-basis}) of the odd Leech lattice,
and (iv) $\varepsilon(e_0,e_0)=1$.

In addition, if $w,w'\in G_{24}$ satisfy the orthogonality condition
\begin{align}
(-1)^w \cdot (-1)^{w'} = \sum_{i=0}^{23}(-1)^{w_i+w'_i} = 0,
\label{eq:ortho-condition-of-codeword}
\end{align}
then any OPE between $T^w,J^w,G^{\pm,w}$ and $T^{w'},J^{w'},G^{\pm,w'}$ vanishes.
Therefore, $T^w+T^{w'}$, $J^w+J^{w'}$, and $G^{\pm,w}+G^{\pm,w'}$ constitute an $c=1+1$ $\cN=2$ SCA.

We can find the following 24 codewords $w^{(1)},\ldots,w^{(24)}\in G_{24}$
\begin{align}
\scalebox{0.7}{$
\begin{array}{r rrrrrrrrrrrrrrrrrrrrrrrr}
(-1)^{w^{(1)}}=\ (& 1 & 1 & 1 & 1 & 1 & 1 & 1 & 1 & 1 & 1 & 1 & 1 & 1 & 1 & 1 & 1 & 1 & 1 & 1 & 1 & 1 & 1 & 1 & 1),\\
(-1)^{w^{(2)}}=\ (& 1 & -1 & 1 & -1 & -1 & -1 & 1 & -1 & -1 & -1 & -1 & -1 & -1 & 1 & -1 & -1 & 1 & 1 & 1 & 1 & 1 & 1 & 1 & 1),\\
(-1)^{w^{(3)}}=\ (& 1 & 1 & -1 & -1 & -1 & 1 & -1 & 1 & 1 & -1 & -1 & 1 & 1 & -1 & -1 & -1 & -1 & -1 & -1 & 1 & 1 & 1 & 1 & 1),\\
(-1)^{w^{(4)}}=\ (& 1 & -1 & -1 & 1 & 1 & -1 & -1 & -1 & -1 & 1 & 1 & -1 & -1 & -1 & 1 & 1 & -1 & -1 & -1 & 1 & 1 & 1 & 1 & 1),\\
(-1)^{w^{(5)}}=\ (& 1 & -1 & -1 & -1 & -1 & -1 & 1 & 1 & 1 & 1 & 1 & -1 & 1 & 1 & 1 & -1 & -1 & 1 & -1 & -1 & -1 & -1 & 1 & 1),\\
(-1)^{w^{(6)}}=\ (& 1 & 1 & -1 & 1 & 1 & 1 & 1 & -1 & -1 & -1 & -1 & 1 & -1 & 1 & -1 & 1 & -1 & 1 & -1 & -1 & -1 & -1 & 1 & 1),\\
(-1)^{w^{(7)}}=\ (& 1 & 1 & -1 & -1 & -1 & -1 & -1 & -1 & -1 & 1 & 1 & 1 & 1 & -1 & -1 & 1 & 1 & 1 & 1 & -1 & -1 & 1 & -1 & 1),\\
(-1)^{w^{(8)}}=\ (& 1 & -1 & -1 & 1 & 1 & 1 & -1 & 1 & 1 & -1 & -1 & -1 & -1 & -1 & 1 & -1 & 1 & 1 & 1 & -1 & -1 & 1 & -1 & 1),\\
(-1)^{w^{(9)}}=\ (& 1 & -1 & 1 & 1 & 1 & -1 & 1 & 1 & -1 & 1 & -1 & 1 & 1 & -1 & -1 & -1 & -1 & -1 & 1 & 1 & -1 & -1 & -1 & 1),\\
(-1)^{w^{(10)}}=\ (& 1 & 1 & 1 & -1 & -1 & 1 & 1 & -1 & 1 & -1 & 1 & -1 & -1 & -1 & 1 & 1 & -1 & -1 & 1 & 1 & -1 & -1 & -1 & 1),\\
(-1)^{w^{(11)}}=\ (& 1 & 1 & 1 & 1 & -1 & -1 & -1 & -1 & 1 & 1 & -1 & 1 & -1 & 1 & 1 & -1 & 1 & -1 & -1 & -1 & 1 & -1 & -1 & 1),\\
(-1)^{w^{(12)}}=\ (& 1 & -1 & 1 & -1 & 1 & 1 & -1 & 1 & -1 & -1 & 1 & -1 & 1 & 1 & -1 & 1 & 1 & -1 & -1 & -1 & 1 & -1 & -1 & 1),\\
(-1)^{w^{(13)}}=\ (& 1 & -1 & -1 & -1 & -1 & 1 & 1 & 1 & -1 & 1 & -1 & 1 & -1 & -1 & 1 & 1 & 1 & 1 & -1 & 1 & 1 & -1 & -1 & -1),\\
(-1)^{w^{(14)}}=\ (& 1 & 1 & -1 & 1 & 1 & -1 & 1 & -1 & 1 & -1 & 1 & -1 & 1 & -1 & -1 & -1 & 1 & 1 & -1 & 1 & 1 & -1 & -1 & -1),\\
(-1)^{w^{(15)}}=\ (& 1 & 1 & 1 & 1 & -1 & -1 & -1 & 1 & -1 & -1 & -1 & -1 & 1 & 1 & 1 & 1 & -1 & 1 & -1 & 1 & -1 & 1 & -1 & -1),\\
(-1)^{w^{(16)}}=\ (& 1 & -1 & 1 & -1 & 1 & 1 & -1 & -1 & 1 & 1 & 1 & 1 & -1 & 1 & -1 & -1 & -1 & 1 & -1 & 1 & -1 & 1 & -1 & -1),\\
(-1)^{w^{(17)}}=\ (& 1 & 1 & -1 & -1 & 1 & 1 & 1 & -1 & -1 & 1 & -1 & -1 & 1 & 1 & 1 & -1 & -1 & -1 & 1 & -1 & 1 & 1 & -1 & -1),\\
(-1)^{w^{(18)}}=\ (& 1 & -1 & -1 & 1 & -1 & -1 & 1 & 1 & 1 & -1 & 1 & 1 & -1 & 1 & -1 & 1 & -1 & -1 & 1 & -1 & 1 & 1 & -1 & -1),\\
(-1)^{w^{(19)}}=\ (& 1 & 1 & 1 & 1 & -1 & 1 & -1 & 1 & -1 & 1 & 1 & -1 & -1 & -1 & -1 & -1 & -1 & 1 & 1 & -1 & 1 & -1 & 1 & -1),\\
(-1)^{w^{(20)}}=\ (& 1 & -1 & 1 & -1 & 1 & -1 & -1 & -1 & 1 & -1 & -1 & 1 & 1 & -1 & 1 & 1 & -1 & 1 & 1 & -1 & 1 & -1 & 1 & -1),\\
(-1)^{w^{(21)}}=\ (& 1 & 1 & -1 & -1 & 1 & -1 & -1 & 1 & 1 & 1 & -1 & -1 & -1 & 1 & -1 & 1 & 1 & -1 & 1 & 1 & -1 & -1 & 1 & -1),\\
(-1)^{w^{(22)}}=\ (& 1 & -1 & -1 & 1 & -1 & 1 & -1 & -1 & -1 & -1 & 1 & 1 & 1 & 1 & 1 & -1 & 1 & -1 & 1 & 1 & -1 & -1 & 1 & -1),\\
(-1)^{w^{(23)}}=\ (& 1 & -1 & 1 & 1 & -1 & 1 & 1 & -1 & 1 & 1 & -1 & -1 & 1 & -1 & -1 & 1 & 1 & -1 & -1 & -1 & -1 & 1 & 1 & -1),\\
(-1)^{w^{(24)}}=\ (& 1 & 1 & 1 & -1 & 1 & -1 & 1 & 1 & -1 & -1 & 1 & 1 & -1 & -1 & 1 & -1 & 1 & -1 & -1 & -1 & -1 & 1 & 1 & -1),
\end{array}
$}
\end{align}
by searching 23 dodecads $w^{(2)},\ldots,w^{(24)}$ such that any two of the 24 vectors, $(-1)^{w^{(1)}}=\underline{1}$ and $(-1)^{w^{(i)}}$ for those dodecads, satisfy the orthogonality condition (\ref{eq:ortho-condition-of-codeword}).
As a result, the sum
\begin{align}
T(z) = \sum_{i=1}^{24}T^{w^{(i)}}(z), \quad J(z) = \sum_{i=1}^{24}J^{w^{(i)}}(z), \quad G^\pm(z) = \sum_{i=1}^{24}G^{\pm,w^{(i)}}(z),
\label{eq:N=2-SCA-from-ternary}
\end{align}
constitute a $c=24$ $\cN=2$ SCA.
Moreover, we searched the 23 dodecads under the condition that the first entry of $(-1)^{w^{(i)}}$ is $(-1)^{w^{(i)}_0}=1$,
so the first column vector of $[(-1)^{w^{(i)}_j}]_{i,j}$ is $\underline{1}$,
and the other column vectors are all orthogonal to $\underline{1}$
because it is an orthogonal matrix (up to scalar multiplication by $\frac{1}{2\sqrt{6}}$).
As a result, the sum $J(z)=\sum_i J^{w^{(i)}}(z)$ turns out to be
\begin{align}
& J(z) = 2\sqrt{-2}\partial X^0(z), \label{eq:J-of-N=2-SCA-from-ternary}
\end{align}
which coincide with $J(z)$ in (\ref{eq:BDFK-J}) proposed by \cite{Benjamin:2015ria}.
We can also check that
\begin{align}
& T(z) = -\frac{1}{2} \sum_{i=0}^{23} :(\partial X^i(z))^2:, \label{eq:T-of-N=2-SCA-from-ternary}
\end{align}
by explicit calculation.
See \texttt{find\_supercurrent.ipynb} for the source code to find $w^{(1)},\ldots,w^{(24)}$ and to check (\ref{eq:J-of-N=2-SCA-from-ternary}, \ref{eq:T-of-N=2-SCA-from-ternary}).

\vspace{\vertspace}

This construction of $c=24$ $\cN=2$ SCA can be explained in connection with ternary codes as follows.
We have constructed an odd Leech lattice $O_{24}^\mathrm{bin}$ from the binary Golay code $G_{24}$,
but we can also construct odd Leech lattices $O_{24}^\mathrm{tern}$ from certain ternary codes 
as we reviewed in Section \ref{subsec:Leech-lattice} (see also \cite[Example 4.5]{Gaiotto:2018ypj}).
Let $W$ denote the orthogonal matrix $[\frac{1}{2\sqrt{6}}(-1)^{w^{(i)}_j}]_{i,j}$.
Then $O_{24}^\mathrm{tern}:=O_{24}^\mathrm{bin}W^T$
is in fact an odd Leech lattice constructed from a ternary code,
and $W$ can be regarded as an isometry between $O_{24}^\mathrm{bin}$ and $O_{24}^\mathrm{tern}$;
for example, the row vector $\frac{1}{2\sqrt{2}}(-1)^{w^{(i)}}\in O_{24}^\mathrm{bin}$ is mapped to
$\frac{1}{2\sqrt{2}}(-1)^{w^{(i)}}W^T = \sqrt{3}\bbe_i$, which should be contained in any lattice constructed from a ternary code by Construction A.
To see that the lattice $O_{24}^\mathrm{tern}$ in fact can be constructed from some ternary code,
it suffices to observe that the lattice $O_{24}^\mathrm{tern}$ is a subset of $\frac{1}{\sqrt{3}}\Z^{24}$,
and contains $\{\sqrt{3}\bbe_i\}_{i=0,\ldots,23}$
(which is a \emph{3-frame} in the terminology of \cite{MR2501075, MR2522420}).
In a lattice CFT constructed from ternary codes,
\cite{Gaiotto:2018ypj} pointed out that we can always find a $c=1$ $\cN=2$ SCA associated to each direction of the standard basis $\bbe_i$,
and hence the sum of them constitute a $c=n=24$ $\cN=2$ SCA with the $U(1)$ current in the direction of $\sum_i\bbe_i=(1,\ldots,1)$.
This direction in $O_{24}^\mathrm{tern}$ is mapped by $W$ to the direction of $\bbe_0=(1,0,\ldots,0)$ in $O_{24}^\mathrm{bin}$,
and therefore we obtained the $c=24$ $\cN=2$ SCA with the $U(1)$ current $J(z)$ in the direction of $\bbe_0$ as above (\ref{eq:J-of-N=2-SCA-from-ternary}).

\vspace{\vertspace}

What is the subgroup of $U(1)^{24}.M_{24} \subset \Aut(V_{O_{24}})$ preserving the $\cN=2$ SCA (\ref{eq:N=2-SCA-from-ternary}) we have constructed?
Since the supercurrents have been constructed from a specific choice of 23 dodecads with the first entry 0,
it should be a subgroup of some extension of $M'$,
where $M'$ is the subgroup of $M_{24}$ stabilizing the set of these 23 dodecads.
The group $M_{23}$ acts transitively on the set $G_{24}^{(12),+}$ of all the dodecads with the first entry 0
(see footnote \ref{fn:M23-orbit}),
so such $M'$ is at least smaller than $M_{23}$,
and would be too small to provide any useful implications.

\section{Data of the Equations to Show the Lift Does Not Preserve the Group Structures}
\label{sec:data-of-eqs}
Here, we record the coefficient matrix $M$ and the vector $\vec{\mu}$ of the equation (\ref{eq:eq-to-find-zeta}) for the cases we mentioned in Section \ref{sec:answer-to-the-question}.

\begin{landscape}
For $M_{24}$ of odd Leech lattice,
\begin{align}
& M = \scalebox{0.45}{$\left(\begin{array}{rrrrrrrrrrrrrrrrrrrrrrrrrrrrrrrrrrrrrrrrrrrrrrrr}
2 & 0 & 0 & 0 & 0 & 0 & 0 & 0 & 0 & 0 & 0 & 0 & 0 & 0 & 0 & 0 & 0 & 0 & 0 & 0 & 0 & 0 & 0 & 0 & 0 & 0 & 0 & 0 & 0 & 0 & 0 & 0 & 0 & 0 & 0 & 0 & 0 & 0 & 0 & 0 & 0 & 0 & 0 & 0 & 0 & 0 & 0 & 0 \\
2 & 0 & 0 & 0 & -1 & -1 & 2 & 1 & -3 & 0 & 4 & -1 & -2 & 3 & 0 & -3 & 1 & 0 & -2 & -1 & 0 & -2 & -2 & 4 & 0 & 0 & 0 & 0 & 0 & 0 & 0 & 0 & 0 & 0 & 0 & 0 & 0 & 0 & 0 & 0 & 0 & 0 & 0 & 0 & 0 & 0 & 0 & 0 \\
2 & -1 & 1 & 1 & -1 & -2 & 2 & 1 & -4 & 0 & 6 & -3 & -3 & 5 & 0 & -4 & 3 & 1 & -3 & 0 & 0 & -2 & -2 & 3 & 0 & 0 & 0 & 0 & 0 & 0 & 0 & 0 & 0 & 0 & 0 & 0 & 0 & 0 & 0 & 0 & 0 & 0 & 0 & 0 & 0 & 0 & 0 & 0 \\
0 & 0 & 1 & 1 & -1 & 0 & 2 & 0 & -1 & 1 & 0 & -4 & 0 & 3 & -1 & 0 & 2 & -1 & -1 & 2 & 2 & 2 & 2 & -4 & 0 & 0 & 0 & 0 & 0 & 0 & 0 & 0 & 0 & 0 & 0 & 0 & 0 & 0 & 0 & 0 & 0 & 0 & 0 & 0 & 0 & 0 & 0 & 0 \\
0 & 0 & 1 & 0 & 0 & 0 & 1 & -1 & 0 & 2 & -1 & -3 & 1 & 2 & -1 & 1 & 2 & -1 & 0 & 2 & 1 & 1 & 2 & -4 & 0 & 0 & 0 & 0 & 0 & 0 & 0 & 0 & 0 & 0 & 0 & 0 & 0 & 0 & 0 & 0 & 0 & 0 & 0 & 0 & 0 & 0 & 0 & 0 \\
2 & -1 & -1 & 1 & 0 & -1 & 0 & 2 & -1 & -2 & 4 & 2 & -3 & 1 & 2 & -3 & -1 & 1 & -2 & -3 & -2 & -2 & -3 & 7 & 0 & 0 & 0 & 0 & 0 & 0 & 0 & 0 & 0 & 0 & 0 & 0 & 0 & 0 & 0 & 0 & 0 & 0 & 0 & 0 & 0 & 0 & 0 & 0 \\
0 & 0 & 0 & 0 & 0 & 0 & 2 & 0 & 0 & 1 & 0 & -1 & 0 & 1 & -1 & 0 & 0 & -1 & -1 & 0 & 0 & 0 & 1 & 1 & 0 & 0 & 0 & 0 & 0 & 0 & 0 & 0 & 0 & 0 & 0 & 0 & 0 & 0 & 0 & 0 & 0 & 0 & 0 & 0 & 0 & 0 & 0 & 0 \\
0 & 0 & 1 & 1 & 0 & -1 & 0 & 1 & 0 & 1 & 0 & -3 & -1 & 1 & -1 & 0 & 2 & 0 & 0 & 1 & 1 & 1 & 2 & -2 & 0 & 0 & 0 & 0 & 0 & 0 & 0 & 0 & 0 & 0 & 0 & 0 & 0 & 0 & 0 & 0 & 0 & 0 & 0 & 0 & 0 & 0 & 0 & 0 \\
2 & 0 & -1 & -1 & 1 & 0 & -2 & 0 & 3 & -1 & -2 & 4 & 1 & -3 & 2 & 2 & -3 & 0 & 1 & -2 & -2 & -1 & -1 & 2 & 0 & 0 & 0 & 0 & 0 & 0 & 0 & 0 & 0 & 0 & 0 & 0 & 0 & 0 & 0 & 0 & 0 & 0 & 0 & 0 & 0 & 0 & 0 & 0 \\
0 & 1 & 0 & -1 & 1 & 2 & -1 & -2 & 2 & 2 & -5 & 0 & 3 & -2 & -1 & 3 & -1 & -1 & 3 & 2 & 1 & 2 & 3 & -6 & 0 & 0 & 0 & 0 & 0 & 0 & 0 & 0 & 0 & 0 & 0 & 0 & 0 & 0 & 0 & 0 & 0 & 0 & 0 & 0 & 0 & 0 & 0 & 0 \\
2 & 0 & 0 & -1 & 0 & 1 & -1 & -2 & 2 & 2 & -3 & -1 & 3 & -1 & -1 & 3 & 0 & -1 & 1 & 1 & 0 & 1 & 2 & -4 & 0 & 0 & 0 & 0 & 0 & 0 & 0 & 0 & 0 & 0 & 0 & 0 & 0 & 0 & 0 & 0 & 0 & 0 & 0 & 0 & 0 & 0 & 0 & 0 \\
0 & 1 & 0 & -1 & 1 & 2 & -1 & -2 & 2 & 0 & -5 & 2 & 3 & -3 & 0 & 3 & -1 & 0 & 3 & 2 & 1 & 2 & 2 & -6 & 0 & 0 & 0 & 0 & 0 & 0 & 0 & 0 & 0 & 0 & 0 & 0 & 0 & 0 & 0 & 0 & 0 & 0 & 0 & 0 & 0 & 0 & 0 & 0 \\
0 & 0 & 2 & 0 & -2 & 0 & 2 & -2 & -2 & 4 & 0 & -8 & 2 & 5 & -3 & 1 & 7 & -1 & -1 & 4 & 3 & 2 & 4 & -9 & 0 & 0 & 0 & 0 & 0 & 0 & 0 & 0 & 0 & 0 & 0 & 0 & 0 & 0 & 0 & 0 & 0 & 0 & 0 & 0 & 0 & 0 & 0 & 0 \\
0 & 0 & 0 & 0 & 0 & 0 & 0 & 0 & 0 & 0 & 0 & 0 & 0 & 1 & 1 & 0 & 1 & 0 & 0 & 0 & 0 & 0 & 0 & -1 & 0 & 0 & 0 & 0 & 0 & 0 & 0 & 0 & 0 & 0 & 0 & 0 & 0 & 0 & 0 & 0 & 0 & 0 & 0 & 0 & 0 & 0 & 0 & 0 \\
0 & 0 & 0 & 0 & 0 & 0 & 0 & 0 & 0 & 0 & 0 & 0 & 0 & 1 & 1 & 0 & 1 & 0 & 0 & 0 & 0 & 0 & 0 & -1 & 0 & 0 & 0 & 0 & 0 & 0 & 0 & 0 & 0 & 0 & 0 & 0 & 0 & 0 & 0 & 0 & 0 & 0 & 0 & 0 & 0 & 0 & 0 & 0 \\
0 & 0 & 0 & 0 & 0 & 0 & 0 & 0 & 2 & 0 & -2 & 0 & 1 & -1 & 0 & 2 & 0 & -1 & 0 & 0 & 0 & 1 & 1 & -1 & 0 & 0 & 0 & 0 & 0 & 0 & 0 & 0 & 0 & 0 & 0 & 0 & 0 & 0 & 0 & 0 & 0 & 0 & 0 & 0 & 0 & 0 & 0 & 0 \\
0 & 0 & 0 & 0 & 0 & 0 & 0 & 0 & 0 & 0 & 0 & 0 & 0 & 0 & 0 & 0 & 2 & 0 & 0 & 0 & 0 & 0 & 0 & 0 & 0 & 0 & 0 & 0 & 0 & 0 & 0 & 0 & 0 & 0 & 0 & 0 & 0 & 0 & 0 & 0 & 0 & 0 & 0 & 0 & 0 & 0 & 0 & 0 \\
0 & 0 & 0 & 0 & 0 & 0 & 0 & 0 & 0 & 0 & 0 & 0 & 0 & 0 & 0 & 0 & 1 & 1 & 0 & 0 & 0 & 0 & 1 & -1 & 0 & 0 & 0 & 0 & 0 & 0 & 0 & 0 & 0 & 0 & 0 & 0 & 0 & 0 & 0 & 0 & 0 & 0 & 0 & 0 & 0 & 0 & 0 & 0 \\
0 & 0 & 0 & 0 & 0 & 0 & 0 & 0 & 0 & 2 & 0 & -2 & 0 & 1 & -1 & 0 & 2 & -1 & 0 & 0 & 0 & 0 & 1 & 0 & 0 & 0 & 0 & 0 & 0 & 0 & 0 & 0 & 0 & 0 & 0 & 0 & 0 & 0 & 0 & 0 & 0 & 0 & 0 & 0 & 0 & 0 & 0 & 0 \\
0 & 0 & 0 & 0 & 2 & 0 & -2 & 0 & 2 & -2 & -2 & 4 & 0 & -4 & 1 & 1 & -2 & 1 & 2 & 0 & -1 & 0 & -1 & 1 & 0 & 0 & 0 & 0 & 0 & 0 & 0 & 0 & 0 & 0 & 0 & 0 & 0 & 0 & 0 & 0 & 0 & 0 & 0 & 0 & 0 & 0 & 0 & 0 \\
0 & 2 & 0 & -2 & 0 & 2 & -2 & -2 & 4 & 0 & -8 & 2 & 5 & -4 & 1 & 6 & -1 & -1 & 4 & 2 & 2 & 3 & 3 & -10 & 0 & 0 & 0 & 0 & 0 & 0 & 0 & 0 & 0 & 0 & 0 & 0 & 0 & 0 & 0 & 0 & 0 & 0 & 0 & 0 & 0 & 0 & 0 & 0 \\
0 & 0 & 0 & 0 & 0 & 2 & 0 & -2 & 0 & 2 & -2 & -2 & 2 & 0 & -2 & 1 & 2 & -1 & 1 & 2 & 1 & 2 & 2 & -4 & 0 & 0 & 0 & 0 & 0 & 0 & 0 & 0 & 0 & 0 & 0 & 0 & 0 & 0 & 0 & 0 & 0 & 0 & 0 & 0 & 0 & 0 & 0 & 0 \\
0 & 0 & 0 & 0 & 0 & 0 & 0 & 0 & 0 & 0 & 0 & 0 & 0 & 0 & 0 & 0 & 1 & 1 & 0 & 0 & 0 & 0 & 1 & -1 & 0 & 0 & 0 & 0 & 0 & 0 & 0 & 0 & 0 & 0 & 0 & 0 & 0 & 0 & 0 & 0 & 0 & 0 & 0 & 0 & 0 & 0 & 0 & 0 \\
0 & 0 & 0 & 0 & 0 & 0 & 0 & 0 & 0 & 0 & 0 & 0 & 0 & 0 & 0 & 0 & 2 & 0 & 0 & 0 & 0 & 0 & 0 & 0 & 0 & 0 & 0 & 0 & 0 & 0 & 0 & 0 & 0 & 0 & 0 & 0 & 0 & 0 & 0 & 0 & 0 & 0 & 0 & 0 & 0 & 0 & 0 & 0 \\
\hline
18 & 0 & 0 & 0 & 0 & 0 & 0 & 0 & 0 & 0 & 0 & 0 & 0 & 0 & 0 & 0 & 0 & 0 & 0 & 0 & 0 & 0 & 0 & 0 & 27 & 0 & 0 & 0 & 0 & 0 & 0 & 0 & 0 & 0 & 0 & 0 & 0 & 0 & 0 & 0 & 0 & 0 & 0 & 0 & 0 & 0 & 0 & 0 \\
14 & 2 & -2 & 0 & -1 & -1 & -2 & 5 & -1 & -7 & 2 & 8 & -4 & -3 & 6 & -3 & -4 & 1 & 1 & -5 & 0 & -2 & -7 & 8 & 20 & 1 & -1 & 2 & -2 & -4 & 0 & 8 & -5 & -10 & 8 & 8 & -10 & 0 & 8 & -8 & -2 & 3 & -1 & -6 & 1 & -3 & -10 & 12 \\
12 & 3 & -1 & 0 & -2 & 0 & -1 & 8 & -1 & -12 & -1 & 11 & -4 & -5 & 8 & -3 & -7 & 3 & 2 & -4 & 3 & 1 & -5 & 3 & 18 & 1 & 0 & 7 & -3 & -8 & 1 & 12 & -12 & -14 & 19 & 5 & -19 & 5 & 9 & -15 & 3 & 11 & -5 & -7 & 5 & -5 & -13 & 15 \\
10 & 2 & 2 & 1 & -7 & 0 & 9 & 1 & -11 & 3 & 7 & -13 & -2 & 15 & -2 & -6 & 13 & -2 & -5 & 4 & 8 & 2 & 2 & -9 & 20 & 3 & 1 & -2 & -11 & 1 & 11 & -1 & -15 & 3 & 11 & -16 & -2 & 20 & 0 & -6 & 17 & -2 & -6 & 6 & 9 & 1 & 1 & -14 \\
12 & 0 & 1 & 3 & -2 & -3 & 2 & 1 & -6 & -1 & 8 & -5 & -5 & 7 & 1 & -5 & 7 & 3 & -3 & 1 & 2 & -1 & -2 & -1 & 18 & 2 & 1 & -2 & -1 & 1 & 0 & 2 & 0 & -5 & -3 & 4 & 0 & -2 & 4 & 1 & -2 & 1 & 1 & 0 & 2 & 2 & -1 & -4 \\
12 & -1 & -1 & 0 & -1 & 1 & 5 & -1 & -6 & 3 & 7 & -3 & -2 & 6 & -3 & -6 & 3 & -1 & -4 & -1 & -1 & -3 & -2 & 8 & 22 & -2 & -2 & -1 & -2 & 0 & 6 & -1 & -9 & 1 & 11 & -2 & -3 & 8 & -2 & -8 & 3 & 1 & -5 & -2 & -2 & -5 & -6 & 11 \\
10 & -1 & 3 & 2 & -2 & -4 & 3 & -2 & -3 & 8 & 5 & -15 & -1 & 12 & -3 & 0 & 13 & -1 & -4 & 4 & 2 & 0 & 5 & -9 & 12 & 5 & 5 & -6 & -2 & 5 & -1 & -7 & 11 & 8 & -21 & -11 & 15 & -1 & -2 & 18 & 5 & -7 & 8 & 12 & 6 & 12 & 18 & -38 \\
10 & 0 & 2 & 1 & 4 & -3 & -7 & 2 & 8 & -4 & -7 & 5 & 0 & -8 & 4 & 6 & -4 & 3 & 5 & 1 & -1 & 2 & 2 & -6 & 10 & 3 & 4 & 1 & 7 & -1 & -12 & 1 & 18 & -4 & -20 & 8 & 5 & -18 & 4 & 15 & -9 & 2 & 11 & 2 & 1 & 8 & 9 & -16 \\
12 & -2 & 0 & -1 & 3 & 2 & -3 & -5 & 5 & 5 & -5 & -1 & 5 & -3 & -3 & 5 & -1 & -1 & 2 & 2 & -3 & 1 & 4 & -4 & 18 & -2 & 0 & -2 & 4 & 4 & -4 & -8 & 8 & 8 & -10 & 0 & 10 & -6 & -6 & 9 & -3 & -3 & 4 & 2 & -5 & 1 & 6 & -5 \\
14 & -3 & -2 & -1 & 6 & 3 & -3 & -5 & 5 & 3 & -3 & 6 & 3 & -9 & -5 & 2 & -8 & -3 & 2 & -1 & -7 & -2 & -3 & 13 & 22 & -6 & -3 & 0 & 8 & 3 & -2 & -6 & 4 & 6 & -1 & 5 & 3 & -7 & -9 & -2 & -9 & -5 & 0 & -3 & -10 & -5 & -3 & 22 \\
10 & -1 & 2 & 0 & 2 & 2 & -2 & -5 & 5 & 7 & -8 & -7 & 6 & -2 & -6 & 7 & 2 & -2 & 3 & 6 & 1 & 5 & 9 & -13 & 16 & -2 & 0 & 0 & 4 & 3 & -3 & -5 & 8 & 8 & -9 & -3 & 8 & -5 & -7 & 7 & -3 & -4 & 2 & 3 & -3 & 2 & 8 & -4 \\
14 & -1 & -2 & -3 & 0 & 3 & 3 & -1 & -1 & -1 & 1 & 6 & 2 & -1 & 0 & -2 & -6 & -2 & -2 & -3 & -3 & -3 & -4 & 10 & 24 & -4 & -2 & 1 & -4 & -2 & 10 & -3 & -15 & 7 & 21 & -12 & -6 & 18 & -7 & -13 & 11 & -1 & -13 & -1 & -2 & -8 & -6 & 16 \\
4 & 4 & 2 & 0 & -4 & -2 & 0 & 8 & -2 & -10 & -2 & 4 & -3 & 1 & 9 & -1 & -1 & 4 & 2 & 0 & 7 & 3 & 0 & -11 & 4 & 6 & 4 & 4 & -6 & -6 & 0 & 8 & -6 & -10 & 6 & -4 & -10 & 6 & 10 & -2 & 10 & 10 & 0 & 2 & 12 & 4 & 0 & -19 \\
8 & 0 & -2 & -4 & -2 & 4 & 4 & 2 & -4 & -4 & 2 & 8 & 1 & 1 & 3 & -5 & -6 & -2 & -2 & -4 & -1 & -3 & -6 & 10 & 16 & -2 & -2 & 2 & -6 & -4 & 8 & -2 & -18 & 2 & 24 & -8 & -10 & 18 & -1 & -14 & 14 & 5 & -11 & -3 & 0 & -9 & -11 & 13 \\
8 & 0 & -2 & 0 & -2 & 0 & 4 & 2 & -8 & -4 & 10 & 4 & -5 & 5 & 3 & -9 & 0 & 2 & -4 & -4 & -1 & -5 & -8 & 12 & 12 & 2 & 0 & -2 & -6 & 2 & 6 & -2 & -10 & 0 & 6 & -6 & -1 & 10 & 1 & -4 & 9 & 1 & -3 & 3 & 5 & 0 & -2 & -8 \\
4 & 0 & 4 & -2 & -6 & 0 & 6 & -2 & -4 & 8 & 0 & -16 & 5 & 15 & -2 & 2 & 13 & -3 & -4 & 6 & 6 & 3 & 9 & -20 & 4 & 6 & 6 & -4 & -8 & 2 & 2 & -10 & 2 & 12 & -10 & -22 & 12 & 12 & -2 & 16 & 20 & -1 & 2 & 14 & 10 & 10 & 18 & -48 \\
8 & 0 & 0 & 0 & -4 & -4 & 4 & 4 & -8 & -4 & 12 & 0 & -6 & 10 & 6 & -8 & 4 & 4 & -6 & -4 & 0 & -6 & -6 & 6 & 8 & 4 & 4 & -2 & -8 & -2 & 6 & 0 & -8 & 0 & 6 & -12 & -2 & 14 & 3 & 0 & 14 & 3 & -4 & 5 & 8 & 3 & 3 & -20 \\
4 & 0 & 0 & -2 & 4 & 0 & -4 & 2 & 8 & -4 & -8 & 10 & 2 & -9 & 4 & 5 & -10 & 1 & 4 & -2 & -3 & 0 & 1 & 1 & 4 & 2 & 2 & 0 & 6 & -2 & -8 & -2 & 12 & -2 & -12 & 6 & 4 & -12 & 3 & 11 & -5 & 6 & 8 & 1 & -1 & 4 & 5 & -13 \\
8 & -2 & 0 & -2 & 2 & 0 & -2 & 0 & 2 & 0 & 0 & 4 & 1 & -2 & 1 & 0 & -4 & -1 & 0 & -2 & -4 & -3 & -3 & 7 & 8 & 0 & 2 & 0 & 2 & 0 & -4 & -6 & 6 & 6 & -6 & -4 & 4 & -2 & -3 & 7 & 3 & -1 & 3 & 3 & -1 & 2 & 5 & -9 \\
8 & -2 & 0 & -2 & 0 & 0 & 0 & 0 & 2 & -2 & 0 & 4 & 2 & -1 & 2 & 1 & -4 & 1 & -2 & -2 & -3 & -2 & -1 & 3 & 8 & 0 & 2 & 0 & -2 & 0 & 0 & -6 & 0 & 8 & 0 & -10 & 5 & 6 & -4 & 6 & 9 & 1 & -3 & 4 & 1 & 1 & 7 & -14 \\
8 & 2 & -2 & -4 & -2 & 4 & 0 & 0 & 0 & -4 & -4 & 8 & 4 & -2 & 5 & 1 & -5 & -1 & 2 & -2 & 1 & 0 & -3 & -2 & 12 & 2 & 0 & 0 & -6 & 0 & 4 & -2 & -10 & 0 & 8 & -6 & -3 & 10 & 2 & -4 & 11 & 3 & -3 & 2 & 5 & -1 & -3 & -8 \\
8 & -2 & 0 & -2 & -4 & 2 & 10 & -2 & -10 & 6 & 10 & -8 & 0 & 12 & -6 & -8 & 5 & -4 & -8 & 0 & 0 & -4 & -2 & 9 & 16 & -4 & -2 & 0 & -6 & -2 & 14 & -4 & -20 & 8 & 26 & -16 & -7 & 24 & -8 & -16 & 15 & 0 & -15 & 0 & -1 & -9 & -7 & 15 \\
4 & 0 & 0 & -2 & 4 & 0 & -4 & 2 & 8 & -4 & -8 & 10 & 2 & -9 & 4 & 5 & -10 & 1 & 4 & -2 & -3 & 0 & 1 & 1 & 4 & 2 & 2 & 0 & 4 & -2 & -6 & -2 & 12 & 2 & -12 & 0 & 5 & -8 & 1 & 11 & -2 & 1 & 5 & 3 & 0 & 5 & 8 & -14 \\
8 & 0 & 0 & -4 & -4 & 0 & 4 & 4 & -4 & -4 & 4 & 4 & 0 & 6 & 6 & -4 & -2 & 0 & -4 & -4 & 0 & -4 & -4 & 4 & 8 & 4 & 4 & 0 & -8 & -4 & 4 & -4 & -8 & 4 & 8 & -16 & -2 & 16 & 2 & 2 & 20 & 6 & -4 & 6 & 8 & 2 & 4 & -25 
\end{array}\right)$}, \label{eq:M24-coef-mat}\\
& \vec{\mu} = \scalebox{0.7}{$\left(\begin{array}{cccccccccccccccccccccccccccccccccccccccccccccccc}
0 & 1 & 0 & 0 & 0 & 1 & 1 & 1 & 0 & 0 & 0 & 0 & 0 & 0 & 0 & 0 & 0 & 0 & 0 & 0 & 0 & 0 & 0 & 0 & 0 & 0 & 1 & 1 & 1 & 0 & 1 & 1 & 1 & 0 & 1 & 1 & 0 & 0 & 0 & 0 & 0 & 0 & 0 & 0 & 0 & 0 & 0 & 0
\end{array}\right)$}^T. \label{eq:M24-RHS-vec}
\end{align}

For $M_{23}$ of odd Leech lattice,
\begin{align}
& M = \scalebox{0.45}{$\left(\begin{array}{rrrrrrrrrrrrrrrrrrrrrrrrrrrrrrrrrrrrrrrrrrrrrrrr}
 15 & 0 & 0 & 0 & 0 & 0 & 0 & 0 & 0 & 0 & 0 & 0 & 0 & 0 & 0 & 0 & 0 & 0 & 0 & 0 & 0 & 0 & 0 & 0 & 28 & 0 & 0 & 0 & 0 & 0 & 0 & 0 & 0 & 0 & 0 & 0 & 0 & 0 & 0 & 0 & 0 & 0 & 0 & 0 & 0 & 0 & 0 & 0 \\
0 & 2 & 6 & 4 & -2 & 4 & 4 & -8 & -4 & 14 & 0 & -24 & 1 & 11 & -12 & 3 & 17 & -2 & -1 & 12 & 11 & 7 & 14 & -23 & 0 & 5 & 11 & 4 & 1 & 9 & 2 & -11 & 6 & 16 & -15 & -27 & 7 & 2 & -16 & 11 & 13 & -5 & 4 & 18 & 16 & 16 & 24 & -36 \\
0 & 5 & 4 & 0 & 0 & 7 & 0 & -6 & 4 & 5 & -14 & -8 & 8 & -3 & -5 & 10 & 4 & -2 & 8 & 12 & 10 & 11 & 13 & -30 & 0 & 10 & 10 & 0 & 1 & 11 & 0 & -11 & 11 & 10 & -28 & -20 & 14 & -6 & -11 & 19 & 7 & -6 & 12 & 24 & 19 & 22 & 27 & -53 \\
0 & 1 & 5 & 1 & 2 & 4 & 3 & -8 & 7 & 11 & -13 & -16 & 9 & 0 & -11 & 10 & 7 & -6 & 4 & 10 & 8 & 12 & 18 & -25 & 0 & 3 & 9 & 5 & 6 & 3 & 2 & -4 & 8 & 4 & -13 & -14 & 1 & -6 & -10 & 8 & 5 & -2 & 7 & 11 & 15 & 17 & 19 & -27 \\
0 & 3 & 3 & -1 & 4 & 10 & 0 & -11 & 8 & 8 & -15 & -12 & 11 & -3 & -11 & 11 & 4 & -7 & 8 & 12 & 10 & 12 & 17 & -28 & 0 & 6 & 5 & 0 & 11 & 15 & -4 & -14 & 20 & 2 & -28 & -4 & 13 & -19 & -12 & 18 & -7 & -8 & 17 & 15 & 13 & 20 & 21 & -33 \\
0 & 3 & 2 & 2 & 5 & 6 & -2 & -4 & 9 & 0 & -13 & -4 & 4 & -8 & -5 & 9 & -2 & -4 & 8 & 8 & 8 & 10 & 12 & -18 & 0 & 6 & 5 & 1 & 11 & 14 & -8 & -12 & 23 & 0 & -34 & 1 & 15 & -23 & -9 & 22 & -9 & -7 & 21 & 14 & 13 & 22 & 22 & -39 \\
0 & 3 & 3 & -1 & 4 & 5 & -1 & -5 & 9 & 4 & -15 & -4 & 9 & -7 & -4 & 9 & -1 & -4 & 9 & 7 & 5 & 11 & 13 & -21 & 0 & 6 & 9 & 2 & 2 & 6 & 4 & -7 & 5 & 8 & -18 & -20 & 7 & -1 & -11 & 11 & 12 & -4 & 11 & 19 & 18 & 22 & 26 & -48 \\
0 & 9 & 5 & 1 & 0 & 2 & -5 & 3 & 10 & -6 & -23 & -2 & 5 & -11 & 3 & 13 & -2 & -1 & 14 & 14 & 12 & 18 & 15 & -38 & 0 & 19 & 8 & -6 & 1 & 11 & -16 & -6 & 31 & -6 & -63 & 9 & 30 & -34 & 8 & 44 & -10 & -4 & 35 & 27 & 19 & 35 & 32 & -90 \\
0 & 7 & 3 & -3 & 2 & 11 & -6 & -7 & 13 & 2 & -27 & 0 & 15 & -13 & -3 & 19 & -5 & -5 & 16 & 13 & 10 & 15 & 16 & -38 & 0 & 14 & 6 & -6 & 7 & 19 & -18 & -17 & 37 & 2 & -64 & 8 & 34 & -36 & -2 & 46 & -13 & -8 & 35 & 24 & 15 & 32 & 33 & -80 \\
0 & 8 & 5 & -1 & -2 & 3 & 1 & 0 & 3 & 2 & -17 & -12 & 8 & -1 & -1 & 9 & 7 & -3 & 8 & 16 & 13 & 18 & 16 & -40 & 0 & 15 & 11 & -1 & -6 & 5 & 4 & -2 & 2 & 4 & -31 & -22 & 14 & 0 & -3 & 18 & 15 & -3 & 17 & 32 & 26 & 30 & 32 & -80 \\
0 & 9 & 2 & -6 & 5 & 12 & -10 & -5 & 26 & -7 & -43 & 12 & 22 & -27 & 2 & 28 & -18 & -8 & 23 & 13 & 9 & 22 & 19 & -45 & 0 & 17 & 2 & -8 & 13 & 19 & -24 & -12 & 47 & -16 & -75 & 31 & 36 & -55 & 8 & 51 & -31 & -7 & 46 & 21 & 14 & 36 & 28 & -79 \\
0 & 6 & 3 & -1 & 3 & 9 & -3 & -6 & 9 & 1 & -20 & -4 & 10 & -9 & -4 & 12 & -1 & -4 & 11 & 13 & 11 & 14 & 15 & -32 & 0 & 10 & 6 & 0 & 6 & 14 & -4 & -10 & 16 & 0 & -36 & -6 & 17 & -16 & -8 & 21 & -3 & -7 & 22 & 23 & 19 & 26 & 26 & -56 \\
0 & 2 & 4 & 0 & -2 & 2 & 4 & -4 & -2 & 8 & -4 & -16 & 4 & 7 & -6 & 2 & 11 & -3 & 0 & 10 & 8 & 8 & 11 & -21 & 0 & 6 & 6 & 0 & -2 & 2 & 4 & -2 & 0 & 4 & -8 & -14 & 3 & 2 & -4 & 5 & 8 & -3 & 3 & 13 & 11 & 11 & 13 & -26 \\
0 & 0 & 4 & 2 & -4 & 4 & 6 & -8 & -6 & 16 & 2 & -26 & 3 & 15 & -12 & 3 & 18 & -5 & -5 & 11 & 10 & 5 & 12 & -20 & 0 & 0 & 6 & 4 & 2 & 4 & 4 & -6 & 0 & 10 & -2 & -18 & 0 & 5 & -12 & 1 & 8 & -5 & -2 & 8 & 9 & 6 & 11 & -9 \\
0 & 0 & 6 & 2 & -6 & -2 & 12 & -4 & -14 & 18 & 10 & -36 & -1 & 25 & -12 & -6 & 28 & -4 & -10 & 13 & 11 & 7 & 13 & -21 & 0 & 0 & 10 & 8 & -10 & -8 & 24 & 6 & -32 & 14 & 30 & -46 & -18 & 39 & -14 & -25 & 36 & -1 & -18 & 14 & 17 & 4 & 11 & -9 \\
0 & 0 & 4 & 2 & 0 & -2 & 2 & -2 & 6 & 6 & -8 & -12 & 3 & 1 & -4 & 8 & 5 & -4 & 0 & 5 & 4 & 8 & 12 & -15 & 0 & 0 & 6 & 4 & 4 & -4 & 2 & 2 & 4 & 0 & -6 & -6 & -3 & -4 & -3 & 4 & 2 & -1 & 2 & 3 & 9 & 9 & 9 & -11 \\
0 & 0 & 4 & 2 & 0 & 2 & 4 & -6 & -6 & 10 & 4 & -18 & -1 & 10 & -9 & -3 & 14 & -1 & -2 & 8 & 7 & 3 & 9 & -12 & 0 & 0 & 6 & 4 & -2 & 4 & 10 & -6 & -10 & 14 & 6 & -28 & -2 & 17 & -15 & -6 & 19 & -5 & -4 & 12 & 11 & 7 & 14 & -16 \\
0 & 2 & 4 & 2 & -4 & -2 & 6 & 0 & -8 & 8 & 4 & -18 & -2 & 12 & -3 & -3 & 15 & 0 & -4 & 8 & 7 & 5 & 7 & -15 & 0 & 4 & 10 & 2 & -10 & -2 & 14 & -2 & -16 & 18 & 4 & -38 & -1 & 25 & -11 & -3 & 30 & -2 & -6 & 19 & 16 & 11 & 20 & -38 \\
0 & 4 & 6 & -2 & -6 & 2 & 4 & -4 & -2 & 12 & -10 & -22 & 9 & 10 & -6 & 8 & 15 & -5 & 2 & 15 & 11 & 13 & 16 & -37 & 0 & 8 & 8 & 0 & -8 & 0 & 6 & 0 & -6 & 8 & -10 & -22 & 5 & 10 & -3 & 6 & 18 & -3 & 4 & 19 & 15 & 15 & 19 & -45 \\
0 & 6 & 6 & 2 & -6 & -2 & 2 & 2 & -4 & 4 & -6 & -18 & -1 & 8 & -1 & 4 & 14 & 0 & 1 & 14 & 12 & 11 & 12 & -31 & 0 & 12 & 8 & -2 & -8 & 2 & 2 & 0 & 0 & 4 & -22 & -16 & 11 & 2 & 1 & 14 & 11 & -3 & 9 & 23 & 17 & 18 & 20 & -55 \\
0 & 2 & 4 & 0 & -2 & 2 & 4 & -4 & -2 & 8 & -4 & -16 & 4 & 7 & -6 & 3 & 11 & -3 & -1 & 10 & 8 & 8 & 11 & -21 & 0 & 4 & 6 & 2 & 0 & 0 & 4 & 0 & 0 & 2 & -8 & -12 & 1 & 1 & -4 & 4 & 6 & -2 & 3 & 11 & 12 & 11 & 12 & -23 \\
0 & 2 & 4 & 0 & -2 & 2 & 4 & -4 & -2 & 8 & -4 & -16 & 4 & 7 & -6 & 3 & 11 & -3 & -1 & 10 & 8 & 8 & 11 & -21 & 0 & 2 & 6 & 4 & -4 & 2 & 14 & -2 & -16 & 10 & 8 & -30 & -4 & 21 & -13 & -10 & 21 & -3 & -4 & 16 & 15 & 9 & 14 & -22 \\
0 & 2 & 4 & 0 & -2 & 2 & 4 & -4 & -2 & 8 & -4 & -16 & 4 & 7 & -6 & 2 & 11 & -3 & 0 & 10 & 8 & 8 & 11 & -21 & 0 & 2 & 6 & 4 & -4 & 0 & 10 & 0 & -12 & 8 & 6 & -24 & -4 & 16 & -9 & -7 & 19 & -2 & -2 & 12 & 12 & 9 & 13 & -21 \\
0 & 0 & 8 & 4 & -8 & -4 & 12 & -4 & -16 & 20 & 12 & -40 & -4 & 28 & -12 & -6 & 32 & -2 & -12 & 14 & 12 & 6 & 14 & -23 & 0 & 0 & 12 & 8 & -12 & -8 & 24 & 4 & -32 & 20 & 28 & -52 & -16 & 42 & -16 & -22 & 40 & -2 & -20 & 16 & 18 & 4 & 14 & -14 \\
\hline
15 & 0 & 0 & 0 & 0 & 0 & 0 & 0 & 0 & 0 & 0 & 0 & 0 & 0 & 0 & 0 & 0 & 0 & 0 & 0 & 0 & 0 & 0 & 0 & 28 & 0 & 0 & 0 & 0 & 0 & 0 & 0 & 0 & 0 & 0 & 0 & 0 & 0 & 0 & 0 & 0 & 0 & 0 & 0 & 0 & 0 & 0 & 0 \\
0 & 2 & 4 & 0 & 3 & 7 & -1 & -8 & 9 & 9 & -13 & -10 & 9 & -3 & -9 & 9 & 2 & -6 & 3 & 8 & 6 & 10 & 15 & -18 & 0 & 5 & 8 & 3 & 5 & 12 & -4 & -12 & 14 & 10 & -21 & -14 & 10 & -9 & -13 & 15 & 4 & -6 & 11 & 14 & 12 & 17 & 22 & -33 \\
0 & 5 & 8 & 0 & -4 & 4 & 3 & -8 & 3 & 15 & -14 & -26 & 9 & 8 & -10 & 12 & 16 & -6 & 3 & 17 & 13 & 15 & 22 & -41 & 0 & 7 & 13 & 4 & -4 & 6 & 4 & -8 & 1 & 18 & -17 & -33 & 8 & 8 & -14 & 12 & 21 & -6 & 5 & 23 & 20 & 20 & 29 & -52 \\
0 & 2 & 2 & 2 & 4 & 7 & -1 & -7 & 7 & 6 & -9 & -7 & 5 & -4 & -9 & 6 & 1 & -4 & 5 & 6 & 6 & 7 & 11 & -12 & 0 & 3 & 6 & 4 & 9 & 15 & -4 & -14 & 15 & 8 & -25 & -11 & 12 & -11 & -16 & 15 & 0 & -7 & 15 & 15 & 14 & 19 & 23 & -35 \\
0 & 2 & 4 & 2 & 1 & 4 & 2 & -7 & 4 & 9 & -4 & -13 & 3 & 3 & -9 & 5 & 7 & -3 & 1 & 6 & 5 & 6 & 11 & -13 & 0 & 2 & 9 & 5 & 3 & 10 & 4 & -12 & 3 & 16 & -4 & -27 & 2 & 5 & -19 & 3 & 14 & -4 & 1 & 13 & 12 & 12 & 20 & -21 \\
0 & 5 & 9 & -3 & -6 & 9 & 4 & -14 & 4 & 21 & -19 & -28 & 15 & 9 & -13 & 17 & 17 & -8 & 5 & 20 & 14 & 14 & 25 & -49 & 0 & 7 & 16 & -1 & -8 & 15 & 6 & -25 & 3 & 42 & -27 & -54 & 22 & 18 & -27 & 24 & 35 & -13 & 4 & 34 & 22 & 23 & 44 & -78 \\
0 & 4 & 5 & 1 & -1 & 4 & 1 & -6 & 5 & 8 & -11 & -9 & 6 & -1 & -6 & 8 & 5 & -1 & 4 & 8 & 6 & 8 & 12 & -21 & 0 & 7 & 13 & 0 & -4 & 12 & 4 & -18 & 5 & 28 & -23 & -39 & 16 & 9 & -20 & 19 & 24 & -7 & 6 & 27 & 19 & 21 & 35 & -63 \\
0 & 5 & 9 & 1 & -5 & 3 & 5 & -7 & -3 & 13 & -12 & -24 & 7 & 8 & -9 & 10 & 18 & -1 & 4 & 19 & 15 & 14 & 20 & -45 & 0 & 10 & 15 & 0 & -8 & 6 & 8 & -11 & -1 & 24 & -26 & -42 & 15 & 13 & -15 & 19 & 28 & -4 & 7 & 33 & 26 & 26 & 38 & -80 \\
0 & 2 & 5 & 5 & -1 & -1 & 4 & 2 & -2 & 0 & 5 & -9 & -6 & 5 & -3 & -6 & 6 & 2 & -3 & 5 & 5 & 5 & 4 & -4 & 0 & 9 & 11 & 0 & -3 & 13 & 4 & -15 & 4 & 20 & -23 & -33 & 15 & 6 & -17 & 16 & 19 & -7 & 9 & 28 & 20 & 22 & 32 & -61 \\
0 & 4 & 5 & 3 & 0 & 1 & 2 & 0 & -2 & 1 & -6 & -5 & 0 & -2 & -3 & 2 & 5 & 4 & 4 & 9 & 8 & 8 & 8 & -19 & 0 & 10 & 12 & 1 & -4 & 3 & 6 & -4 & 1 & 10 & -22 & -27 & 9 & 4 & -8 & 14 & 17 & 2 & 9 & 26 & 23 & 24 & 29 & -65 \\
0 & 4 & 2 & 4 & 5 & 3 & -2 & 2 & 4 & -8 & -8 & 1 & -1 & -9 & -1 & 1 & -3 & 1 & 8 & 8 & 9 & 11 & 7 & -15 & 0 & 11 & 4 & 1 & 8 & 9 & -4 & -1 & 15 & -14 & -29 & 6 & 9 & -22 & 0 & 14 & -11 & -2 & 22 & 17 & 17 & 24 & 17 & -41 \\
0 & 6 & 2 & 0 & 3 & 5 & -2 & 0 & 7 & -5 & -13 & 2 & 4 & -9 & 0 & 6 & -4 & -1 & 9 & 8 & 7 & 11 & 8 & -19 & 0 & 12 & 6 & -1 & 4 & 11 & -4 & -6 & 15 & -4 & -32 & -3 & 14 & -17 & -3 & 19 & -3 & -4 & 21 & 21 & 17 & 25 & 23 & -53 \\
0 & 8 & 4 & -6 & -2 & 8 & -4 & -8 & 12 & 8 & -32 & -10 & 19 & -7 & -1 & 23 & 3 & -8 & 12 & 17 & 12 & 18 & 22 & -51 & 0 & 12 & 6 & -4 & -2 & 6 & -6 & -4 & 14 & 0 & -38 & -2 & 19 & -14 & 3 & 26 & -2 & -5 & 18 & 21 & 15 & 21 & 21 & -59 \\
0 & 8 & 2 & -8 & -2 & 10 & -6 & -8 & 18 & 6 & -28 & -2 & 18 & -8 & 1 & 22 & -5 & -10 & 9 & 9 & 4 & 12 & 15 & -32 & 0 & 12 & 4 & -8 & 0 & 20 & -12 & -20 & 24 & 12 & -44 & -4 & 30 & -18 & -6 & 34 & -4 & -12 & 20 & 22 & 8 & 18 & 24 & -56 \\
0 & 12 & 4 & -10 & -4 & 10 & -8 & -8 & 18 & 8 & -44 & -2 & 25 & -14 & 3 & 32 & -3 & -8 & 17 & 17 & 10 & 19 & 23 & -58 & 0 & 16 & 10 & -8 & -10 & 8 & -6 & -10 & 14 & 14 & -48 & -16 & 29 & -7 & 1 & 37 & 10 & -6 & 18 & 31 & 18 & 25 & 32 & -88 \\
0 & 8 & 2 & -8 & 0 & 12 & -8 & -10 & 20 & 6 & -38 & 0 & 23 & -15 & 0 & 28 & -7 & -10 & 16 & 13 & 8 & 16 & 20 & -45 & 0 & 12 & 4 & -8 & 2 & 16 & -14 & -16 & 28 & 6 & -52 & 4 & 31 & -25 & -1 & 39 & -9 & -10 & 25 & 21 & 10 & 22 & 25 & -64 \\
0 & 8 & 4 & -6 & -4 & 6 & -2 & -6 & 8 & 8 & -26 & -8 & 16 & -4 & 0 & 19 & 4 & -5 & 11 & 14 & 9 & 14 & 17 & -44 & 0 & 12 & 10 & -6 & -10 & 6 & 0 & -10 & 6 & 18 & -34 & -24 & 22 & 3 & -4 & 27 & 18 & -7 & 12 & 28 & 17 & 22 & 31 & -77 \\
0 & 8 & 8 & -8 & -8 & 8 & 0 & -12 & 8 & 16 & -32 & -16 & 22 & 0 & -2 & 26 & 10 & -6 & 10 & 20 & 12 & 16 & 24 & -61 & 0 & 12 & 14 & -8 & -14 & 8 & 2 & -18 & 6 & 34 & -40 & -40 & 29 & 11 & -11 & 35 & 28 & -8 & 8 & 35 & 20 & 23 & 40 & -94 \\
0 & 12 & 2 & -8 & -2 & 8 & -8 & -4 & 16 & 0 & -38 & 4 & 20 & -16 & 6 & 26 & -6 & -5 & 18 & 14 & 9 & 18 & 17 & -50 & 0 & 20 & 6 & -12 & -6 & 12 & -8 & -10 & 20 & 2 & -58 & -2 & 34 & -19 & 6 & 41 & -2 & -7 & 28 & 34 & 19 & 30 & 31 & -95 \\
0 & 8 & 2 & -4 & 0 & 4 & -4 & 2 & 8 & -6 & -22 & 4 & 8 & -11 & 6 & 13 & -5 & -2 & 11 & 11 & 8 & 13 & 10 & -31 & 0 & 16 & 4 & -8 & -2 & 8 & -6 & -4 & 16 & -6 & -44 & 4 & 23 & -19 & 7 & 29 & -6 & -4 & 24 & 25 & 16 & 25 & 22 & -70 \\
0 & 8 & 2 & -6 & 0 & 8 & -6 & -4 & 16 & 0 & -32 & 2 & 17 & -14 & 3 & 20 & -7 & -7 & 13 & 12 & 8 & 16 & 16 & -38 & 0 & 14 & 4 & -6 & 0 & 8 & -10 & -4 & 20 & -6 & -42 & 8 & 21 & -23 & 7 & 29 & -10 & -4 & 22 & 19 & 12 & 21 & 18 & -56 \\
0 & 8 & 2 & -6 & 0 & 8 & -6 & -4 & 12 & 0 & -28 & 2 & 15 & -12 & 3 & 19 & -5 & -5 & 14 & 12 & 8 & 14 & 14 & -38 & 0 & 14 & 4 & -8 & 0 & 10 & -10 & -8 & 22 & -2 & -46 & 6 & 25 & -23 & 5 & 33 & -8 & -6 & 24 & 21 & 12 & 23 & 22 & -64 \\
0 & 12 & 2 & -10 & 0 & 12 & -10 & -8 & 20 & 0 & -46 & 6 & 26 & -21 & 5 & 32 & -9 & -7 & 23 & 17 & 11 & 21 & 21 & -60 & 0 & 16 & 4 & -8 & -2 & 10 & -10 & -6 & 20 & -4 & -50 & 6 & 27 & -22 & 7 & 34 & -7 & -5 & 27 & 25 & 15 & 26 & 24 & -75 \\
0 & 16 & 4 & -16 & -4 & 16 & -12 & -12 & 28 & 8 & -64 & 0 & 38 & -22 & 6 & 46 & -8 & -14 & 26 & 24 & 14 & 28 & 32 & -83 & 0 & 24 & 8 & -16 & -8 & 16 & -16 & -16 & 32 & 8 & -80 & 0 & 48 & -28 & 8 & 60 & -4 & -12 & 36 & 40 & 20 & 36 & 40 & -120 
\end{array}\right)$}, \label{eq:M23-coef-mat}\\
& \vec{\mu} = \scalebox{0.7}{$\left(\begin{array}{cccccccccccccccccccccccccccccccccccccccccccccccc}
0 & 0 & 0 & 0 & 1 & 1 & 0 & 0 & 1 & 1 & 1 & 1 & 0 & 0 & 0 & 0 & 0 & 0 & 0 & 0 & 0 & 0 & 0 & 0 & 0 & 1 & 0 & 1 & 0 & 0 & 1 & 0 & 1 & 0 & 1 & 0 & 0 & 0 & 0 & 0 & 0 & 0 & 0 & 0 & 0 & 0 & 0 & 0
\end{array}\right)$}^T. \label{eq:M23-RHS-vec}
\end{align}

For $\Co_0$ of Leech lattice,
\begin{align}
& M = \scalebox{0.35}{$\left(\begin{array}{rrrrrrrrrrrrrrrrrrrrrrrrrrrrrrrrrrrrrrrrrrrrrrrr}
364 & 84 & 37 & -495 & -288 & -227 & -27 & -367 & 155 & 245 & 186 & -31 & 119 & 45 & 103 & 87 & -12 & -24 & -51 & -9 & 16 & 38 & -73 & 100 & -180 & -155 & -105 & -31 & -53 & 31 & -50 & -10 & 31 & 62 & 8 & 2 & -18 & 3 & -55 & 16 & -20 & -6 & -9 & 0 & -17 & 10 & -52 & -35 \\
-718 & 43 & 13 & 1070 & 690 & 424 & 199 & 774 & -327 & -524 & -418 & 95 & -232 & -122 & -149 & -136 & 109 & 88 & 103 & 71 & 17 & -89 & 125 & -96 & 49 & 332 & 190 & 415 & 363 & 104 & 240 & 290 & -213 & -310 & -175 & 70 & -118 & -81 & -12 & -31 & 32 & 69 & 44 & 31 & 8 & -72 & 130 & 0 \\
-496 & 24 & 10 & 762 & 477 & 322 & 123 & 535 & -237 & -375 & -277 & 44 & -159 & -87 & -108 & -105 & 74 & 62 & 67 & 31 & 5 & -79 & 100 & -72 & 107 & 209 & 123 & 166 & 169 & 23 & 140 & 105 & -96 & -164 & -71 & 41 & -55 & -42 & -1 & -10 & 3 & 35 & 15 & 20 & -4 & -41 & 81 & 7 \\
-850 & 151 & 60 & 1337 & 879 & 519 & 295 & 942 & -409 & -654 & -509 & 108 & -272 & -164 & -154 & -154 & 171 & 126 & 119 & 93 & 38 & -130 & 152 & -68 & -24 & 386 & 208 & 550 & 479 & 158 & 330 & 385 & -278 & -412 & -238 & 113 & -191 & -120 & -68 & -25 & 15 & 98 & 50 & 51 & -13 & -103 & 159 & -28 \\
-80 & 266 & 135 & 257 & 247 & 23 & 198 & 192 & -94 & -120 & -155 & 81 & -40 & -48 & 51 & 34 & 126 & 64 & 24 & 82 & 57 & -11 & 13 & 104 & -199 & 64 & -2 & 251 & 200 & 144 & 138 & 197 & -143 & -168 & -119 & 55 & -157 & -83 & -104 & 26 & -25 & 53 & 2 & 12 & -43 & -75 & 37 & -68 \\
470 & -103 & -54 & -757 & -506 & -286 & -174 & -547 & 243 & 370 & 309 & -78 & 159 & 90 & 80 & 79 & -106 & -75 & -71 & -61 & -23 & 64 & -93 & 34 & -44 & -214 & -113 & -224 & -211 & -68 & -164 & -156 & 135 & 197 & 99 & -47 & 97 & 64 & 25 & 0 & 1 & -47 & -13 & -16 & 15 & 63 & -85 & 13 \\
824 & -230 & -104 & -1332 & -908 & -483 & -346 & -954 & 415 & 649 & 541 & -142 & 272 & 167 & 126 & 128 & -203 & -139 & -123 & -123 & -56 & 115 & -145 & 30 & 87 & -391 & -198 & -608 & -521 & -203 & -354 & -436 & 317 & 445 & 266 & -119 & 233 & 142 & 92 & 15 & -11 & -110 & -48 & -47 & 24 & 125 & -163 & 48 \\
-354 & 127 & 50 & 575 & 402 & 197 & 172 & 407 & -172 & -279 & -232 & 64 & -113 & -77 & -46 & -49 & 97 & 64 & 52 & 62 & 33 & -51 & 52 & 4 & -131 & 177 & 85 & 384 & 310 & 135 & 190 & 280 & -182 & -248 & -167 & 72 & -136 & -78 & -67 & -15 & 12 & 63 & 35 & 31 & -9 & -62 & 78 & -35 \\
364 & 84 & 37 & -495 & -288 & -227 & -27 & -367 & 155 & 245 & 186 & -31 & 119 & 45 & 103 & 87 & -12 & -24 & -51 & -9 & 16 & 38 & -73 & 100 & -180 & -155 & -105 & -31 & -53 & 31 & -50 & -10 & 31 & 62 & 8 & 2 & -18 & 3 & -55 & 16 & -20 & -6 & -9 & 0 & -17 & 10 & -52 & -35 \\
-486 & 235 & 97 & 842 & 591 & 292 & 268 & 575 & -254 & -409 & -323 & 77 & -153 & -119 & -51 & -67 & 159 & 102 & 68 & 84 & 54 & -92 & 79 & 32 & -204 & 231 & 103 & 519 & 426 & 189 & 280 & 375 & -247 & -350 & -230 & 115 & -209 & -117 & -123 & -9 & -5 & 92 & 41 & 51 & -30 & -93 & 107 & -63 \\
-1188 & 146 & 67 & 1827 & 1196 & 710 & 373 & 1321 & -570 & -894 & -727 & 173 & -391 & -212 & -229 & -215 & 215 & 163 & 174 & 132 & 40 & -153 & 218 & -130 & 93 & 546 & 303 & 639 & 574 & 172 & 404 & 446 & -348 & -507 & -274 & 117 & -215 & -145 & -37 & -31 & 31 & 116 & 57 & 47 & -7 & -135 & 215 & -13 \\
-824 & 230 & 104 & 1332 & 908 & 483 & 346 & 954 & -415 & -649 & -541 & 142 & -272 & -167 & -126 & -128 & 203 & 139 & 123 & 123 & 56 & -115 & 145 & -30 & -87 & 391 & 198 & 608 & 521 & 203 & 354 & 436 & -317 & -445 & -266 & 119 & -233 & -142 & -92 & -15 & 11 & 110 & 48 & 47 & -24 & -125 & 163 & -48 \\
-718 & 43 & 13 & 1070 & 690 & 424 & 199 & 774 & -327 & -524 & -418 & 95 & -232 & -122 & -149 & -136 & 109 & 88 & 103 & 71 & 17 & -89 & 125 & -96 & 49 & 332 & 190 & 415 & 363 & 104 & 240 & 290 & -213 & -310 & -175 & 70 & -118 & -81 & -12 & -31 & 32 & 69 & 44 & 31 & 8 & -72 & 130 & 0 \\
132 & -108 & -47 & -267 & -189 & -95 & -96 & -168 & 82 & 130 & 91 & -13 & 40 & 42 & 5 & 18 & -62 & -38 & -16 & -22 & -21 & 41 & -27 & -28 & 73 & -54 & -18 & -135 & -116 & -54 & -90 & -95 & 65 & 102 & 63 & -43 & 73 & 39 & 56 & -6 & 17 & -29 & -6 & -20 & 21 & 31 & -29 & 28 \\
-824 & 230 & 104 & 1332 & 908 & 483 & 346 & 954 & -415 & -649 & -541 & 142 & -272 & -167 & -126 & -128 & 203 & 139 & 123 & 123 & 56 & -115 & 145 & -30 & -87 & 391 & 198 & 608 & 521 & 203 & 354 & 436 & -317 & -445 & -266 & 119 & -233 & -142 & -92 & -15 & 11 & 110 & 48 & 47 & -24 & -125 & 163 & -48 \\
-718 & 43 & 13 & 1070 & 690 & 424 & 199 & 774 & -327 & -524 & -418 & 95 & -232 & -122 & -149 & -136 & 109 & 88 & 103 & 71 & 17 & -89 & 125 & -96 & 49 & 332 & 190 & 415 & 363 & 104 & 240 & 290 & -213 & -310 & -175 & 70 & -118 & -81 & -12 & -31 & 32 & 69 & 44 & 31 & 8 & -72 & 130 & 0 \\
1310 & -465 & -201 & -2174 & -1499 & -775 & -614 & -1529 & 669 & 1058 & 864 & -219 & 425 & 286 & 177 & 195 & -362 & -241 & -191 & -207 & -110 & 207 & -224 & -2 & 291 & -622 & -301 & -1127 & -947 & -392 & -634 & -811 & 564 & 795 & 496 & -234 & 442 & 259 & 215 & 24 & -6 & -202 & -89 & -98 & 54 & 218 & -270 & 111 \\
26 & 79 & 44 & -5 & 29 & -36 & 51 & 12 & -6 & 5 & -32 & 34 & 0 & -3 & 28 & 26 & 32 & 13 & 4 & 30 & 18 & 15 & -7 & 38 & -63 & 5 & -10 & 58 & 42 & 45 & 24 & 51 & -39 & -33 & -28 & 6 & -42 & -22 & -24 & 10 & -4 & 12 & -2 & -4 & -11 & -22 & 4 & -20 \\
824 & -230 & -104 & -1332 & -908 & -483 & -346 & -954 & 415 & 649 & 541 & -142 & 272 & 167 & 126 & 128 & -203 & -139 & -123 & -123 & -56 & 115 & -145 & 30 & 87 & -391 & -198 & -608 & -521 & -203 & -354 & -436 & 317 & 445 & 266 & -119 & 233 & 142 & 92 & 15 & -11 & -110 & -48 & -47 & 24 & 125 & -163 & 48 \\
-26 & -79 & -44 & 5 & -29 & 36 & -51 & -12 & 6 & -5 & 32 & -34 & 0 & 3 & -28 & -26 & -32 & -13 & -4 & -30 & -18 & -15 & 7 & -38 & 63 & -5 & 10 & -58 & -42 & -45 & -24 & -51 & 39 & 33 & 28 & -6 & 42 & 22 & 24 & -10 & 4 & -12 & 2 & 4 & 11 & 22 & -4 & 20 \\
338 & 5 & -7 & -490 & -317 & -191 & -78 & -379 & 161 & 240 & 218 & -65 & 119 & 48 & 75 & 61 & -44 & -37 & -55 & -39 & -2 & 23 & -66 & 62 & -117 & -160 & -95 & -89 & -95 & -14 & -74 & -61 & 70 & 95 & 36 & -4 & 24 & 25 & -31 & 6 & -16 & -18 & -7 & 4 & -6 & 32 & -56 & -15 \\
-364 & -84 & -37 & 495 & 288 & 227 & 27 & 367 & -155 & -245 & -186 & 31 & -119 & -45 & -103 & -87 & 12 & 24 & 51 & 9 & -16 & -38 & 73 & -100 & 180 & 155 & 105 & 31 & 53 & -31 & 50 & 10 & -31 & -62 & -8 & -2 & 18 & -3 & 55 & -16 & 20 & 6 & 9 & 0 & 17 & -10 & 52 & 35 \\
338 & 5 & -7 & -490 & -317 & -191 & -78 & -379 & 161 & 240 & 218 & -65 & 119 & 48 & 75 & 61 & -44 & -37 & -55 & -39 & -2 & 23 & -66 & 62 & -117 & -160 & -95 & -89 & -95 & -14 & -74 & -61 & 70 & 95 & 36 & -4 & 24 & 25 & -31 & 6 & -16 & -18 & -7 & 4 & -6 & 32 & -56 & -15 \\
364 & 84 & 37 & -495 & -288 & -227 & -27 & -367 & 155 & 245 & 186 & -31 & 119 & 45 & 103 & 87 & -12 & -24 & -51 & -9 & 16 & 38 & -73 & 100 & -180 & -155 & -105 & -31 & -53 & 31 & -50 & -10 & 31 & 62 & 8 & 2 & -18 & 3 & -55 & 16 & -20 & -6 & -9 & 0 & -17 & 10 & -52 & -35 \\
\hline
375 & -19 & 3 & -549 & -350 & -226 & -120 & -359 & 161 & 250 & 192 & -35 & 112 & 65 & 80 & 85 & -39 & -40 & -50 & -24 & -18 & 50 & -51 & 56 & -125 & -55 & -58 & 66 & 17 & 33 & -2 & 36 & 11 & -6 & -34 & 10 & -3 & -21 & -32 & -8 & -4 & -6 & 9 & -4 & -5 & 13 & -26 & -10 \\
-216 & -10 & -21 & 301 & 184 & 130 & 59 & 181 & -82 & -133 & -90 & 14 & -58 & -37 & -54 & -56 & 9 & 16 & 31 & 3 & 5 & -31 & 27 & -40 & 138 & -11 & 21 & -165 & -108 & -66 & -52 & -109 & 44 & 75 & 73 & -26 & 43 & 37 & 37 & 12 & -5 & -13 & -19 & -5 & 2 & 8 & -5 & 15 \\
159 & -29 & -18 & -248 & -166 & -96 & -61 & -178 & 79 & 117 & 102 & -21 & 54 & 28 & 26 & 29 & -30 & -24 & -19 & -21 & -13 & 19 & -24 & 16 & 13 & -66 & -37 & -99 & -91 & -33 & -54 & -73 & 55 & 69 & 39 & -16 & 40 & 16 & 5 & 4 & -9 & -19 & -10 & -9 & -3 & 21 & -31 & 5 \\
318 & -58 & -36 & -496 & -332 & -192 & -122 & -356 & 158 & 234 & 204 & -42 & 108 & 56 & 52 & 58 & -60 & -48 & -38 & -42 & -26 & 38 & -48 & 32 & 26 & -132 & -74 & -198 & -182 & -66 & -108 & -146 & 110 & 138 & 78 & -32 & 80 & 32 & 10 & 8 & -18 & -38 & -20 & -18 & -6 & 42 & -62 & 10 \\
318 & -58 & -36 & -496 & -332 & -192 & -122 & -356 & 158 & 234 & 204 & -42 & 108 & 56 & 52 & 58 & -60 & -48 & -38 & -42 & -26 & 38 & -48 & 32 & 26 & -132 & -74 & -198 & -182 & -66 & -108 & -146 & 110 & 138 & 78 & -32 & 80 & 32 & 10 & 8 & -18 & -38 & -20 & -18 & -6 & 42 & -62 & 10 \\
0 & 0 & 0 & 0 & 0 & 0 & 0 & 0 & 0 & 0 & 0 & 0 & 0 & 0 & 0 & 0 & 0 & 0 & 0 & 0 & 0 & 0 & 0 & 0 & 0 & 0 & 0 & 0 & 0 & 0 & 0 & 0 & 0 & 0 & 0 & 0 & 0 & 0 & 0 & 0 & 0 & 0 & 0 & 0 & 0 & 0 & 0 & 0 \\
-204 & 136 & 114 & 390 & 296 & 124 & 126 & 350 & -152 & -202 & -228 & 56 & -100 & -38 & 4 & -4 & 102 & 64 & 14 & 78 & 42 & -14 & 42 & 16 & -328 & 286 & 106 & 726 & 580 & 264 & 320 & 510 & -308 & -426 & -302 & 116 & -246 & -138 & -94 & -40 & 46 & 102 & 78 & 46 & 8 & -100 & 134 & -50 \\
-102 & 68 & 57 & 195 & 148 & 62 & 63 & 175 & -76 & -101 & -114 & 28 & -50 & -19 & 2 & -2 & 51 & 32 & 7 & 39 & 21 & -7 & 21 & 8 & -164 & 143 & 53 & 363 & 290 & 132 & 160 & 255 & -154 & -213 & -151 & 58 & -123 & -69 & -47 & -20 & 23 & 51 & 39 & 23 & 4 & -50 & 67 & -25 \\
159 & -29 & -18 & -248 & -166 & -96 & -61 & -178 & 79 & 117 & 102 & -21 & 54 & 28 & 26 & 29 & -30 & -24 & -19 & -21 & -13 & 19 & -24 & 16 & 13 & -66 & -37 & -99 & -91 & -33 & -54 & -73 & 55 & 69 & 39 & -16 & 40 & 16 & 5 & 4 & -9 & -19 & -10 & -9 & -3 & 21 & -31 & 5 \\
375 & -19 & 3 & -549 & -350 & -226 & -120 & -359 & 161 & 250 & 192 & -35 & 112 & 65 & 80 & 85 & -39 & -40 & -50 & -24 & -18 & 50 & -51 & 56 & -125 & -55 & -58 & 66 & 17 & 33 & -2 & 36 & 11 & -6 & -34 & 10 & -3 & -21 & -32 & -8 & -4 & -6 & 9 & -4 & -5 & 13 & -26 & -10 \\
0 & 0 & 0 & 0 & 0 & 0 & 0 & 0 & 0 & 0 & 0 & 0 & 0 & 0 & 0 & 0 & 0 & 0 & 0 & 0 & 0 & 0 & 0 & 0 & 0 & 0 & 0 & 0 & 0 & 0 & 0 & 0 & 0 & 0 & 0 & 0 & 0 & 0 & 0 & 0 & 0 & 0 & 0 & 0 & 0 & 0 & 0 & 0 \\
318 & -58 & -36 & -496 & -332 & -192 & -122 & -356 & 158 & 234 & 204 & -42 & 108 & 56 & 52 & 58 & -60 & -48 & -38 & -42 & -26 & 38 & -48 & 32 & 26 & -132 & -74 & -198 & -182 & -66 & -108 & -146 & 110 & 138 & 78 & -32 & 80 & 32 & 10 & 8 & -18 & -38 & -20 & -18 & -6 & 42 & -62 & 10 \\
-102 & 68 & 57 & 195 & 148 & 62 & 63 & 175 & -76 & -101 & -114 & 28 & -50 & -19 & 2 & -2 & 51 & 32 & 7 & 39 & 21 & -7 & 21 & 8 & -164 & 143 & 53 & 363 & 290 & 132 & 160 & 255 & -154 & -213 & -151 & 58 & -123 & -69 & -47 & -20 & 23 & 51 & 39 & 23 & 4 & -50 & 67 & -25 \\
-159 & 29 & 18 & 248 & 166 & 96 & 61 & 178 & -79 & -117 & -102 & 21 & -54 & -28 & -26 & -29 & 30 & 24 & 19 & 21 & 13 & -19 & 24 & -16 & -13 & 66 & 37 & 99 & 91 & 33 & 54 & 73 & -55 & -69 & -39 & 16 & -40 & -16 & -5 & -4 & 9 & 19 & 10 & 9 & 3 & -21 & 31 & -5 \\
-216 & -10 & -21 & 301 & 184 & 130 & 59 & 181 & -82 & -133 & -90 & 14 & -58 & -37 & -54 & -56 & 9 & 16 & 31 & 3 & 5 & -31 & 27 & -40 & 138 & -11 & 21 & -165 & -108 & -66 & -52 & -109 & 44 & 75 & 73 & -26 & 43 & 37 & 37 & 12 & -5 & -13 & -19 & -5 & 2 & 8 & -5 & 15 \\
216 & 10 & 21 & -301 & -184 & -130 & -59 & -181 & 82 & 133 & 90 & -14 & 58 & 37 & 54 & 56 & -9 & -16 & -31 & -3 & -5 & 31 & -27 & 40 & -138 & 11 & -21 & 165 & 108 & 66 & 52 & 109 & -44 & -75 & -73 & 26 & -43 & -37 & -37 & -12 & 5 & 13 & 19 & 5 & -2 & -8 & 5 & -15 \\
273 & 49 & 60 & -354 & -202 & -164 & -57 & -184 & 85 & 149 & 78 & -7 & 62 & 46 & 82 & 83 & 12 & -8 & -43 & 15 & 3 & 43 & -30 & 64 & -289 & 88 & -5 & 429 & 307 & 165 & 158 & 291 & -143 & -219 & -185 & 68 & -126 & -90 & -79 & -28 & 19 & 45 & 48 & 19 & -1 & -37 & 41 & -35 \\
375 & -19 & 3 & -549 & -350 & -226 & -120 & -359 & 161 & 250 & 192 & -35 & 112 & 65 & 80 & 85 & -39 & -40 & -50 & -24 & -18 & 50 & -51 & 56 & -125 & -55 & -58 & 66 & 17 & 33 & -2 & 36 & 11 & -6 & -34 & 10 & -3 & -21 & -32 & -8 & -4 & -6 & 9 & -4 & -5 & 13 & -26 & -10 \\
0 & 0 & 0 & 0 & 0 & 0 & 0 & 0 & 0 & 0 & 0 & 0 & 0 & 0 & 0 & 0 & 0 & 0 & 0 & 0 & 0 & 0 & 0 & 0 & 0 & 0 & 0 & 0 & 0 & 0 & 0 & 0 & 0 & 0 & 0 & 0 & 0 & 0 & 0 & 0 & 0 & 0 & 0 & 0 & 0 & 0 & 0 & 0 \\
-375 & 19 & -3 & 549 & 350 & 226 & 120 & 359 & -161 & -250 & -192 & 35 & -112 & -65 & -80 & -85 & 39 & 40 & 50 & 24 & 18 & -50 & 51 & -56 & 125 & 55 & 58 & -66 & -17 & -33 & 2 & -36 & -11 & 6 & 34 & -10 & 3 & 21 & 32 & 8 & 4 & 6 & -9 & 4 & 5 & -13 & 26 & 10 \\
-216 & -10 & -21 & 301 & 184 & 130 & 59 & 181 & -82 & -133 & -90 & 14 & -58 & -37 & -54 & -56 & 9 & 16 & 31 & 3 & 5 & -31 & 27 & -40 & 138 & -11 & 21 & -165 & -108 & -66 & -52 & -109 & 44 & 75 & 73 & -26 & 43 & 37 & 37 & 12 & -5 & -13 & -19 & -5 & 2 & 8 & -5 & 15 \\
0 & 0 & 0 & 0 & 0 & 0 & 0 & 0 & 0 & 0 & 0 & 0 & 0 & 0 & 0 & 0 & 0 & 0 & 0 & 0 & 0 & 0 & 0 & 0 & 0 & 0 & 0 & 0 & 0 & 0 & 0 & 0 & 0 & 0 & 0 & 0 & 0 & 0 & 0 & 0 & 0 & 0 & 0 & 0 & 0 & 0 & 0 & 0 \\
-318 & 58 & 36 & 496 & 332 & 192 & 122 & 356 & -158 & -234 & -204 & 42 & -108 & -56 & -52 & -58 & 60 & 48 & 38 & 42 & 26 & -38 & 48 & -32 & -26 & 132 & 74 & 198 & 182 & 66 & 108 & 146 & -110 & -138 & -78 & 32 & -80 & -32 & -10 & -8 & 18 & 38 & 20 & 18 & 6 & -42 & 62 & -10 \\
-57 & -39 & -39 & 53 & 18 & 34 & -2 & 3 & -3 & -16 & 12 & -7 & -4 & -9 & -28 & -27 & -21 & -8 & 12 & -18 & -8 & -12 & 3 & -24 & 151 & -77 & -16 & -264 & -199 & -99 & -106 & -182 & 99 & 144 & 112 & -42 & 83 & 53 & 42 & 16 & -14 & -32 & -29 & -14 & -1 & 29 & -36 & 20 \\
\hline
35 & -1 & -2 & -57 & -34 & -25 & -8 & -40 & 19 & 26 & 20 & -2 & 10 & 5 & 6 & 9 & -5 & -7 & -5 & -2 & -1 & 5 & -7 & 7 & 16 & -13 & -8 & -44 & -34 & -12 & -14 & -33 & 16 & 21 & 15 & -1 & 9 & 2 & 0 & 6 & -7 & -4 & -7 & -2 & -6 & 6 & -6 & 4 \\
-2 & -1 & -1 & 8 & 1 & 6 & 0 & 1 & -2 & -1 & 3 & -6 & 2 & 0 & 1 & -5 & -1 & 2 & -1 & -3 & 1 & -4 & -1 & -2 & -2 & -2 & 0 & 1 & 4 & -4 & -2 & 6 & -1 & 3 & -3 & 1 & -1 & 4 & 3 & 1 & 3 & -2 & 4 & 2 & 3 & 2 & 0 & -2 \\
-54 & -3 & 1 & 79 & 50 & 32 & 10 & 59 & -25 & -34 & -31 & 5 & -14 & -7 & -9 & -11 & 8 & 7 & 9 & 5 & 3 & -4 & 9 & -9 & -12 & 16 & 9 & 42 & 32 & 14 & 15 & 28 & -16 & -23 & -13 & 1 & -9 & -3 & 0 & -5 & 5 & 6 & 5 & 0 & 4 & -7 & 8 & -4 \\
25 & 2 & 0 & -33 & -20 & -15 & -3 & -23 & 10 & 15 & 12 & -2 & 5 & 2 & 5 & 4 & -2 & -4 & -4 & -1 & 0 & 3 & -4 & 5 & 17 & -9 & -6 & -40 & -28 & -12 & -11 & -29 & 12 & 18 & 12 & 1 & 7 & 2 & 1 & 8 & -8 & -4 & -5 & -2 & -6 & 4 & -3 & 2 \\
10 & -4 & 2 & -16 & -10 & -8 & -7 & -5 & 2 & 7 & 0 & 3 & 1 & 4 & 2 & 3 & -1 & -2 & 0 & 0 & -1 & 5 & 1 & 0 & 27 & 6 & 4 & -25 & -16 & -10 & -9 & -14 & 7 & 10 & 6 & -5 & 7 & 2 & 10 & 3 & 2 & -4 & -2 & -2 & -1 & 2 & -1 & 4 \\
40 & -8 & -5 & -66 & -43 & -26 & -14 & -49 & 22 & 33 & 28 & -9 & 15 & 7 & 8 & 8 & -8 & -8 & -7 & -6 & 0 & 4 & -10 & 5 & -9 & -21 & -12 & -25 & -21 & -4 & -12 & -21 & 14 & 18 & 11 & -1 & 6 & 4 & -3 & 3 & -4 & -1 & -4 & -1 & -1 & 6 & -7 & 1 \\
0 & 3 & -1 & -3 & -1 & -2 & 4 & -4 & 3 & -1 & 2 & 2 & 1 & 1 & -1 & 2 & -2 & 0 & 1 & 1 & -1 & -1 & -1 & 3 & -27 & 6 & 1 & 40 & 26 & 14 & 14 & 24 & -12 & -18 & -12 & 3 & -9 & -7 & -9 & -7 & 4 & 4 & 3 & 4 & 2 & -5 & 1 & -2 \\
-15 & -3 & -2 & 19 & 11 & 9 & 4 & 9 & -5 & -5 & -3 & -2 & 1 & 2 & -2 & -5 & -3 & 3 & 4 & -2 & 3 & -4 & 0 & -3 & 1 & 1 & 1 & 0 & 1 & -4 & -3 & 3 & 1 & 4 & -1 & -1 & 0 & 3 & 4 & -1 & 6 & -3 & 3 & 3 & 2 & 1 & -1 & -1 \\
-16 & -3 & -3 & 14 & 10 & 4 & 4 & 8 & -1 & -3 & -5 & 2 & 0 & 0 & -3 & -2 & -1 & 0 & 4 & 1 & 2 & -2 & -1 & -1 & 26 & 2 & 2 & -32 & -21 & -13 & -10 & -20 & 8 & 13 & 8 & -1 & 5 & 3 & 6 & 5 & 0 & -5 & -3 & 0 & -4 & 4 & 0 & 3 \\
18 & -5 & -1 & -29 & -18 & -12 & -5 & -21 & 8 & 18 & 12 & -4 & 8 & 6 & 5 & 2 & -5 & -4 & -2 & -3 & 3 & 3 & -5 & 3 & 20 & -13 & -7 & -51 & -36 & -16 & -18 & -36 & 18 & 27 & 18 & -4 & 9 & 7 & 6 & 8 & -4 & -5 & -5 & -1 & -4 & 5 & -6 & 2 \\
-24 & -8 & 2 & 39 & 23 & 20 & 1 & 28 & -16 & -14 & -10 & -6 & -5 & 0 & -2 & -9 & 2 & 6 & 3 & -2 & 6 & -4 & 3 & -9 & -25 & -7 & -5 & 23 & 15 & 10 & 2 & 12 & -1 & -3 & -1 & -3 & -1 & 3 & 0 & -4 & 4 & 4 & 4 & -2 & 6 & -3 & 1 & -5 \\
-3 & -9 & 1 & 2 & 1 & 1 & -6 & 6 & -3 & 1 & -3 & -1 & -2 & 2 & 0 & -1 & 0 & 1 & 0 & 0 & 2 & 4 & 1 & -4 & 16 & 0 & 1 & -21 & -13 & -6 & -9 & -13 & 8 & 9 & 4 & -1 & 7 & 3 & 8 & 5 & 0 & -2 & -1 & -3 & -2 & 3 & 2 & 2 \\
-6 & -4 & 1 & 10 & 4 & 10 & -2 & 5 & -4 & -5 & 3 & -7 & -1 & 1 & -1 & -4 & 0 & 4 & -2 & -2 & 1 & -4 & 0 & -4 & -46 & -5 & -4 & 58 & 37 & 20 & 13 & 37 & -12 & -20 & -14 & 1 & -8 & -2 & -8 & -11 & 6 & 7 & 8 & 2 & 8 & -4 & 2 & -6 \\
43 & 7 & -2 & -57 & -39 & -27 & -5 & -42 & 20 & 22 & 20 & 1 & 10 & 4 & 7 & 8 & -6 & -6 & -5 & -2 & -6 & 4 & -5 & 10 & 28 & 11 & 6 & -25 & -12 & -16 & -5 & -8 & 1 & 8 & -4 & 4 & 3 & -2 & 5 & 6 & 1 & -9 & 0 & 4 & -4 & 4 & -2 & 5 \\
-20 & 3 & 4 & 33 & 22 & 10 & 6 & 29 & -13 & -17 & -21 & 9 & -10 & -2 & -4 & -3 & 3 & 4 & 5 & 4 & 0 & 2 & 7 & -4 & 21 & 24 & 15 & 8 & 13 & -1 & 5 & 13 & -10 & -9 & -10 & 0 & -3 & -2 & 11 & 1 & 7 & -2 & 3 & 2 & 1 & -3 & 7 & 2 \\
46 & 1 & -2 & -63 & -40 & -26 & -8 & -46 & 19 & 28 & 26 & -6 & 12 & 5 & 9 & 8 & -6 & -6 & -7 & -4 & -1 & 4 & -8 & 7 & -22 & -20 & -14 & -7 & -10 & 3 & -4 & -11 & 7 & 9 & 6 & 2 & 2 & -1 & -9 & 2 & -7 & 1 & -3 & -2 & -2 & 2 & -6 & -2 \\
-5 & -5 & -1 & 1 & 1 & 3 & -2 & 2 & 0 & -2 & 2 & -1 & -1 & 0 & -2 & 1 & 0 & 1 & -1 & 1 & 0 & 0 & -1 & -2 & -48 & -6 & -5 & 56 & 32 & 25 & 14 & 30 & -11 & -23 & -11 & 1 & -8 & -6 & -13 & -10 & 2 & 10 & 3 & -1 & 4 & -5 & 2 & -4 \\
-9 & 2 & 2 & 10 & 9 & 3 & 2 & 10 & -4 & -5 & -8 & 5 & -4 & -2 & -2 & 1 & 3 & 1 & 1 & 3 & 0 & 2 & 2 & -1 & -2 & 2 & 3 & 9 & 5 & 6 & 3 & 7 & -4 & -7 & -2 & -1 & -3 & -1 & 0 & -1 & 0 & 4 & -1 & -1 & 0 & -1 & 2 & 0 \\
-37 & -1 & -2 & 47 & 33 & 19 & 7 & 32 & -13 & -17 & -16 & 0 & -6 & -7 & -5 & -6 & 7 & 2 & 5 & 2 & 5 & -4 & 2 & -5 & -28 & -11 & -5 & 24 & 12 & 15 & 7 & 10 & -4 & -8 & 3 & -3 & -5 & 2 & -8 & -6 & -2 & 9 & -1 & -2 & 5 & -3 & -2 & -4 \\
16 & 3 & 1 & -19 & -12 & -4 & -1 & -18 & 6 & 8 & 13 & -8 & 3 & -1 & 3 & 1 & 0 & -2 & -6 & -3 & 1 & -3 & -4 & 2 & -21 & -20 & -12 & -3 & -6 & 4 & -2 & -9 & 6 & 6 & 9 & 0 & 1 & 2 & -9 & -1 & -7 & 3 & -3 & -1 & 0 & 0 & -3 & -2 \\
8 & -4 & 0 & -13 & -9 & 0 & -5 & -13 & 5 & 8 & 12 & -10 & 4 & 0 & 2 & 0 & 0 & 0 & -5 & -3 & 2 & -3 & -4 & -1 & -18 & -24 & -15 & -11 & -13 & 1 & -6 & -15 & 11 & 11 & 12 & -1 & 4 & 4 & -7 & 0 & -7 & 2 & -3 & -4 & 0 & 3 & -4 & -2 \\
23 & -2 & -1 & -29 & -22 & -9 & -8 & -22 & 8 & 13 & 15 & -7 & 8 & 4 & 5 & 2 & -4 & -2 & -4 & -5 & -1 & -1 & -3 & 2 & -25 & -7 & -6 & 21 & 13 & 7 & 4 & 12 & -2 & -4 & -4 & 0 & -3 & -1 & -5 & -4 & 2 & 2 & 3 & 1 & 4 & 0 & -3 & -3 \\
1 & 2 & 2 & 2 & 1 & 5 & 3 & -5 & -1 & -1 & 6 & -6 & 1 & 0 & 0 & -1 & -1 & 2 & -2 & -3 & 2 & -5 & -1 & 0 & -11 & -17 & -10 & -13 & -12 & -1 & -6 & -15 & 10 & 11 & 11 & -2 & 3 & 3 & -5 & 0 & -4 & 0 & -3 & -1 & -1 & 0 & -3 & -2 \\
45 & 2 & -2 & -65 & -40 & -31 & -11 & -44 & 18 & 29 & 18 & 2 & 11 & 5 & 8 & 11 & -5 & -10 & -4 & -4 & -4 & 8 & -4 & 8 & 49 & -6 & -1 & -77 & -53 & -26 & -21 & -50 & 23 & 32 & 22 & -4 & 15 & 5 & 8 & 11 & -8 & -9 & -10 & -2 & -7 & 7 & -9 & 8 
\end{array}\right)$}, \label{eq:Co0-coef-mat}\\
&\scalebox{0.8}{$\begin{array}{rccccccccccccccccccccccl}
\vec{\mu} = (\ 1 & 1 & 0 & 0 & 1 & 0 & 0 & 0 & 1 & 1 & 1 & 0 & 1 & 1 & 0 & 1 & 1 & 0 & 0 & 0 & 1 & 1 & 1 & 1 \\
0 & 0 & 0 & 0 & 0 & 0 & 0 & 0 & 0 & 0 & 0 & 0 & 0 & 0 & 0 & 0 & 0 & 0 & 0 & 0 & 0 & 0 & 0 & 0 \\
0 & 1 & 0 & 0 & 0 & 0 & 1 & 0 & 0 & 1 & 0 & 1 & 0 & 1 & 1 & 0 & 0 & 1 & 1 & 1 & 0 & 0 & 0 & 1\ )^T.
\end{array}$} \label{eq:Co0-RHS-vec}
\end{align}
\end{landscape}

\section{A Subgroup \texorpdfstring{$U(1)^{23}.M_{23}$}{U(1)23.M23} of \texorpdfstring{$U(1)^{24}.M_{24}$}{U(1)24.M24}}
\label{sec:proof-of-subgroup-of-extension}
Here, we write down the detailed proof that the group $\Stab_{\Hom(L,U(1))}(\bbe_0).M_{23}$, introduced around (\ref{eq:def-of-stabsubgrp}), is a subgroup of $\Hom(L,U(1)).M_{24} \subset O(\hat{L})$,
where $L$ is the odd Leech lattice $O_{24}$ constructed as in (\ref{eq:odd-Leech-from-Golay-code}).

Let us take a section $S:O(L)\to O(\hat{L}); g \mapsto S_g$ of (\ref{eq:lift-of-O(L)}).
For each $g\in O(L)$, we have defined $\zeta_g:L\to U(1)$ in (\ref{eq:def-of-zeta}) as
\begin{align}
S_g(e^k) = \zeta_g(k)e^{g(k)}.
\end{align}
Define another section $S':O(L)\to O(\hat{L})$ by
\begin{align}
S'_g = \zeta_g(r\bbe_0)^{-1}S_g,
\end{align}
where $r\bbe_0$ is a primitive vector of $L\cap\R\bbe_0$, say $r\bbe_0=\sqrt{2}(2,0,\ldots,0)$.
Then $\{\tilde{\eta} \circ S'_g \mid \eta\in\Stab_{\Hom(L,U(1))}(\bbe_0), g\in M_{23}\}$,
where $\tilde{\ }:\Hom(L,U(1))\to O(\hat{L})$ is defined as in (\ref{eq:def-of-tilde-pre}, \ref{eq:def-of-tilde}),
is a subgroup of $\Hom(L,U(1)).M_{24}$ as follows.

The multiplication of two elements $\tilde{\eta}_1 \circ S'_{g_1}$ and $\tilde{\eta}_2 \circ S'_{g_2}$ in $\Hom(L,U(1)).M_{24}$ is
\begin{align}
(\tilde{\eta}_1 \circ S'_{g_1}) \circ (\tilde{\eta}_2 \circ S'_{g_2}) = \tilde{\eta}_1 \circ \widetilde{\eta_2 \circ g_1^{-1}} \circ \widetilde{E'(g_1,g_2)} \circ S'_{g_1g_2},
\label{eq:multiplication-of-e0-stabilizing-elements}
\end{align}
where we used
\begin{align}
S'_{g_1} \circ \tilde{\eta}_2 \circ (S'_{g_1})^{-1} = \widetilde{\eta_2 \circ g_1^{-1}},
\end{align}
and the 2-cocycle $E':O(L) \times O(L) \to \Hom(L,U(1))$ is specified by
\begin{align}
S'_{g_1} \circ S'_{g_2} = \widetilde{E'(g_1,g_2)} \circ S'_{g_1g_2}.
\label{eq:2-cocycle-of-e0-stabilizing-section}
\end{align}
If $g_1\in M_{23}$ and $\eta_2 \in \Stab_{\Hom(L,U(1))}(\bbe_0)$,
then $\eta_2 \circ g_1^{-1} \in \Stab_{\Hom(L,U(1))}(\bbe_0)$.
If $g_1,g_2\in M_{23}$, the both-hand sides of (\ref{eq:2-cocycle-of-e0-stabilizing-section}) map $e^{r\bbe_0}$ as
\begin{align}
e^{r\bbe_0} = E'(g_1,g_2)(r\bbe_0)e^{r\bbe_0},
\end{align}
and hence $E'(g_1,g_2) \in \Stab_{\Hom(L,U(1))}(\bbe_0)$.
Therefore, the set $\{\tilde{\eta} \circ S'_g \mid \eta\in\Stab_{\Hom(L,U(1))}(\bbe_0),$
$g\in M_{23}\}$ is closed under the multiplication (\ref{eq:multiplication-of-e0-stabilizing-elements}),
and it turns out to form a subgroup $\Stab_{\Hom(L,U(1))}(\bbe_0).M_{23}$ of $\Hom(L,U(1)).M_{24}$.

\newpage

\bibliographystyle{ytamsalpha}
\def\arxivfont{\rm}
\baselineskip=.95\baselineskip
\bibliography{refs}

\end{document}